%% file: main.tex
\documentclass[12pt]{article}
\usepackage{arxiv}
\usepackage{hyperref}

\usepackage{times}
\usepackage{fullpage}

\usepackage{xurl}

\usepackage{amsmath, amsthm}
\newtheorem{theorem}{Theorem}
\theoremstyle{definition}
\newtheorem{definition}{Definition}

\usepackage{pifont}
\usepackage{url,tikz}
\usepackage{graphicx}
\usepackage{booktabs}
\usepackage{amsfonts}       
\usepackage{nicefrac}      
\usepackage{xcolor}
\usepackage{savesym}
\savesymbol{Bbbk}% colors
\usepackage{graphicx, amsfonts, amsmath,amsthm, amssymb} 

\usepackage{xcolor}
\usepackage{algorithmicx}
\usepackage{mathrsfs, comment,csquotes}
\usepackage{framed}
\usepackage{varwidth}
\usepackage{nccmath, mathtools}
\usepackage{subcaption}
\usepackage{comment}
\usepackage{multirow, multicol}
\usepackage{subcaption}
\usepackage{tikz}
\usetikzlibrary{positioning, arrows.meta, shapes.geometric, fit, backgrounds, decorations.pathmorphing}
\usepackage{fontawesome}  % for \faUser; alternatively swap for an emoji or a generic person symbol
\excludecomment{comment}

\newtheorem{remark}{Remark}

\usepackage{tcolorbox}
\usepackage{framed}

\newcommand{\mygraybox}[1]{%
  \fcolorbox{gray!80}{gray!10}{%
    \parbox{0.97\columnwidth}{#1}%
  }%
}

\newcommand{\sysname}{RootGuard}

\author{
 Divyam Anshumaan \\
  University of Wisconsin-Madison\\
   \And
 Sarthak Choudhary \\
  University of Wisconsin-Madison\\
   \And
Nils Palumbo \\
  University of Wisconsin-Madison\\
   \And
   Somesh Jha \\
  University of Wisconsin-Madison\\
}

\title{Dependency-Aware Privacy for Multi-Turn LLM Agents}
\begin{document}
\maketitle
\begin{abstract}

LLM-based agents routinely release private data across multiple services and interactions. Existing prompt sanitizers based on metric differential privacy treat each release independently. In multi-turn workflows, an adversary can combine these independent releases to recover non-trivial information about the user's private attributes, causing the effective privacy guarantee to degrade with every additional release. We show that this degradation is fundamental: if the user's private attributes are the \emph{roots} of a computation graph, independently noising a value derived from a root amplifies the root's distinguishability by up to the deriving function's Lipschitz constant $L$, which can greatly exceed the nominal privacy parameter for the nonlinear functions common in medical and financial workflows.

To address this, we propose RootGuard, which sanitizes root values once and computes all subsequent releases deterministically from the noised roots. By the post-processing theorem of differential privacy, the privacy guarantee depends only on the initial root sanitization --- regardless of the adversary's choice of functions or number of interaction turns. We show that derived values inherit privacy at zero marginal cost. RootGuard further exploits structural domain knowledge when available, such as dependencies among the declared roots (e.g., BMI is a function of height and weight) or knowledge of the task's target function, to allocate privacy budget across roots and improve the privacy-utility tradeoff.

A worst-case multi-turn adversary forcing $t$ turns of interaction increases the total privacy budget $B = t \cdot \varepsilon$ available to the user. RootGuard distributes this larger budget across roots (reducing noise per root), while independent noising spends $\varepsilon$ per release and gives the adversary $t$ independent observations to combine via MAP reconstruction. The result is a \emph{double asymmetry}: additional turns simultaneously improve RootGuard's utility and weaken independent noising's privacy. On eight medical diagnostic templates drawn from NHANES (the CDC's National Health and Nutrition Examination Survey), RootGuard achieves $2.6\times$ lower target error than independent noising at $\varepsilon = 0.1$ (7.8\% vs.\ 20.3\% wMAPE at $B = (2k{+}1)\varepsilon$), with the advantage growing as the adversary forces more turns. Under MAP reconstruction attacks, additional queries strengthen the adversary against independent noising while reconstruction under RootGuard remains invariant.
\end{abstract}

\input{sections/1_Introduction}
\input{sections/2_Background}
\input{sections/3_Method}
\input{sections/4_PrivacyAnalysis}
\input{sections/5_experiments_worst_case_adversary}
\input{sections/6_RelatedWork}
\input{sections/7_Conclusion}

\bibliographystyle{unsrt} 
\bibliography{references}

\appendix

\input{sections/Appendix}

\end{document}

%% file: sections/1_Introduction.tex
\section{Introduction}

\begin{figure*}
    \centering
    \includegraphics[width=\linewidth]{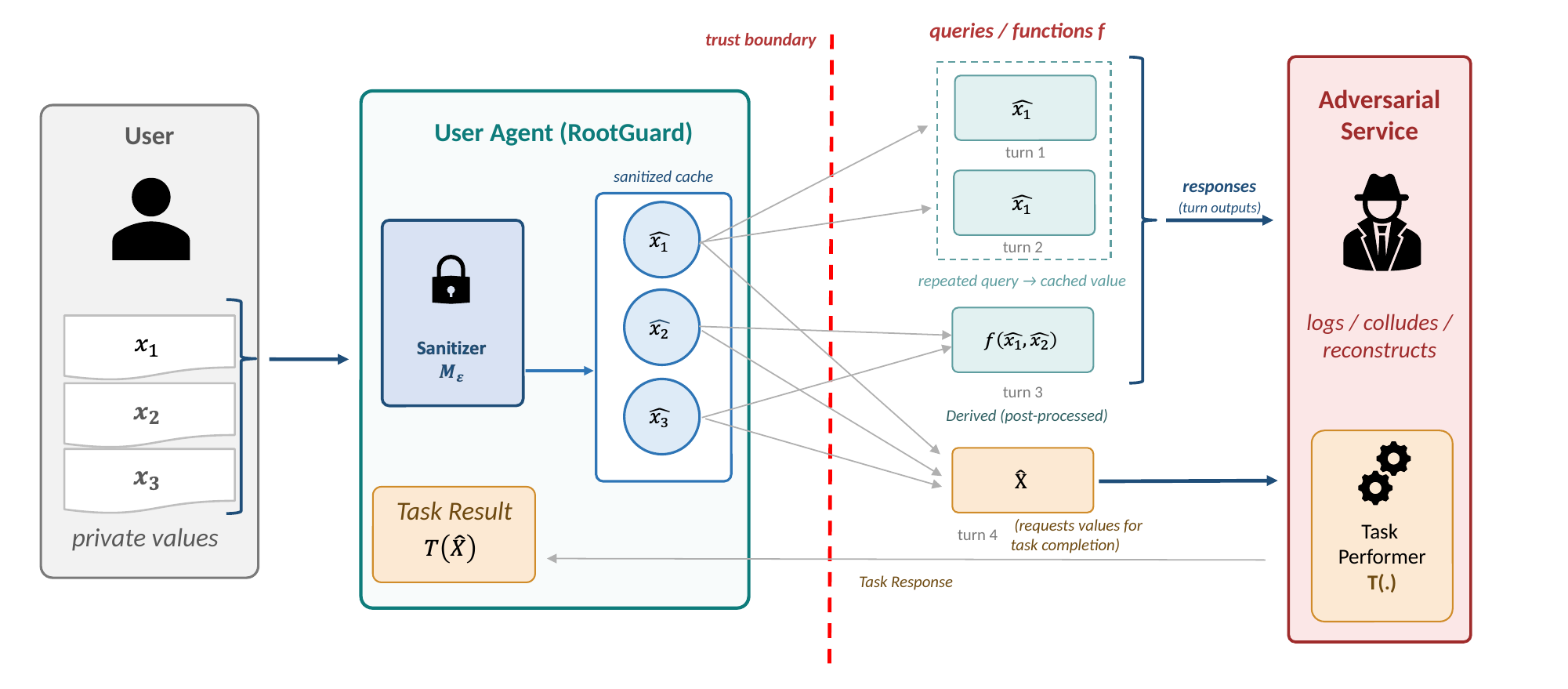}
    \caption{Overview of \sysname{} in an agentic deployment. At initialization, the user passes its private values $X = (x_1, x_2, x_3)$ to the user agent, which sanitizes each root once via $\mathcal{M}_\varepsilon$ and stores the result in the sanitized cache. All subsequent interactions with the adversarial service are answered from the cache: repeated requests for the same root return the identical cached value (turns 1, 2), derived requests are deterministic post-processing of cached roots (turn 3), and the request that supplies all roots for task computation (turn 4) carries no additional privacy cost. The trust boundary separates the user agent from the adversarial service, which performs the user's task by computing $T(\hat{X})$ and may concurrently log, collude with other parties, or attempt to reconstruct the underlying $X$. The task result $T(\hat{X})$ is returned to the user agent. By the post-processing theorem, the adversary's view across the boundary is no more informative than the initial root sanitization, regardless of which functions $f$ it requests.}
\label{fig:rootguard_overview}
\end{figure*}

The rapid integration of large language models into sensitive domains, such as healthcare, finance, legal consultation, etc. has created a new attack surface for privacy. Users increasingly delegating tasks to LLM agents that act on their behalf --- accessing external tools and services via protocols such as  MCP~\cite{anthropic2024mcp,hou2025mcp}, and executing multi-step workflows. For instance, a medical agent might submit lab values to a diagnostic service, forward results to a medication advisor, and send recommendations to insurance providers, all without human oversight of the private data released at each step.

Every intermediary in this chain is a potential adversary. Services can be malicious by design, collecting private data under the guise of providing a task. They can be compromised through prompt injection attacks~\cite{greshake2023indirect,zhan2024injecagent,debenedetti2024agentdojo}, where adversarial inputs force an agent to reveal user data or compute unintended functions on sensitive information, or they can be coerced through legal processes~\cite{openai-nyt} or breached through conventional exploits. Therefore, client-side privacy protection is not only a precaution against dishonest providers, it is a necessary defense against the compromise of honest ones. The user agent must therefore assume that \emph{any} value released to \emph{any} service may be exposed, and sanitize accordingly.

The standard defense is to add noise to each released value independently, using metric differential privacy~\cite{Chatzikokolakis2013BroadeningTS}. Methods such as Pr$\epsilon\epsilon$mpt~\cite{chowdhury2025} and CAPE~\cite{wu2025capecontextawarepromptperturbation} protect individual tokens in a prompt with formal guarantees. However, these methods treat each release in isolation, disregarding the relationships between released values. This independence assumption breaks down in practice --- even when all parties are honest.

% \paragraph{Benign leakage.}
% Consider two routine scenarios that arise naturally in agentic workflows:
 
% \begin{itemize}
%     \item \textbf{Cross-platform redundancy.} A user's agent shares the same lab value, say, fasting glucose, with a diabetes screening service and a metabolic health dashboard. Each service receives an independently noised copy. If either service is later breached or subpoenaed, the attacker holds two independent views of the same value. By the composition theorem~\cite{dwork2006calibrating}, the joint privacy cost is $2\varepsilon$, not $\varepsilon$. As agents interact with more services, the effective privacy budget scales linearly with the number of releases.
 
%     \item \textbf{Structural redundancy.} A user releases both a private root value and a derived value that depends on it --- for example, height and BMI. Each release is independently noised. But the adversary now holds two independent views of the same underlying root (i.e., height). We show (Theorem~\ref{thm:privacy-leakage}) that the effective distinguishability of the root is amplified by up to the Lipschitz constant of the deriving function --- for nonlinear functions common in medical and financial domains (ratios, reciprocals, logarithms), this amplification can exceed $100\times$ the nominal privacy parameter. 
%     % \sarthak{This could be too much detail, maybe say we formalize the leakage of privacy from this structure in Theorem in terms of the dependencies ....}
% \end{itemize}

\paragraph{Benign leakage.} Two routine scenarios suffice to break the per-release guarantee. \emph{Cross-platform redundancy:} a single root (e.g., fasting glucose) sent to two services arrives as two independently noised copies, costing $2\varepsilon$ rather than $\varepsilon$ by composition~\cite{dwork2006calibrating}; the budget scales linearly with services touched. \emph{Structural redundancy:} a root and a derived value (e.g., height and BMI) released separately each leak the root, and we show (Theorem~\ref{thm:privacy-leakage}) that for nonlinear derivations the effective distinguishability is amplified by up to the Lipschitz constant, exceeding $100\times$ the nominal $\varepsilon$ for common medical formulas. These are natural consequences of how independently-noising methods interact with multi-service, multi-value workflows. 

\paragraph{Adversarial leakage.} The problem is severe in adversarial settings. Services can be compromised through prompt injection~\cite{greshake2023indirect,zhan2024injecagent,debenedetti2024agentdojo}, giving the adversary control of the interaction. It can then request the same value repeatedly (averaging away protection at rate $1/\sqrt{q}$ over $q$ queries), force computation of arbitrary derived functions (each a fresh independent observation), or exploit the target release as a cross-root constraint linking all private values through the task function.

% \paragraph{Adversarial leakage.}
% The problem is severe in adversarial settings. Services can be compromised through prompt injection attacks~\cite{greshake2023indirect,zhan2024injecagent,debenedetti2024agentdojo}, where adversarial inputs force an agent to reveal user data or compute unintended functions on sensitive information. In the worst case, the adversary controls the interaction and can:
 
% \begin{itemize}
%     \item request the same private value \emph{repeatedly}, obtaining $q$ independent noise draws and averaging away the privacy protection (reconstruction error decreases as $1/\sqrt{q}$);
%     \item force the agent to compute \emph{arbitrary functions} of private data, each producing a fresh, independently noised observation that can be combined with previous releases;
%     \item exploit the \emph{target release} itself as a cross-root constraint, linking all private values through the known task function.
% \end{itemize}

\paragraph{The fundamental problem.} Both the benign and adversarial scenarios are instances of the same root cause: \emph{uncontrolled independent releases} of correlated private values. Without a mechanism to track dependencies between releases, every additional interaction, whether benign or adversarial, degrades the user's privacy beyond the nominal budget. Existing sanitizers~\cite{chowdhury2025,wu2025capecontextawarepromptperturbation} are stateless: they treat each value in isolation, discarding structural information that is already present in the interaction protocol.
 
\paragraph{Our approach.} We propose \sysname{} (Fig.~\ref{fig:rootguard_overview}), a dependency-aware privacy mechanism for multi-turn agentic interactions. We noise the private root values exactly once, and compute everything else deterministically from the noised roots. By the post-processing theorem of differential privacy~\cite{dwork2014algorithmic}, all derived values, regardless of the function the adversary chooses, inherit privacy at zero additional cost. Repeated queries return the same cached value. Cross-platform releases of the same root are identical. The adversary's view after $t$ turns is no more informative than after the first release.

This creates a \emph{double asymmetry}. When the adversary forces $t$ turns of interaction, the total privacy budget $B = t\varepsilon$ scales with $t$; \sysname{} redistributes this larger budget across the $k$ roots, while independent noising spends $\varepsilon$ per release. As a result, with additional turns, \sysname{}'s utility improves (more budget per root) while its privacy stays constant (all releases are post-processing). Independent noising gains nothing, since the target is always a single release at $\varepsilon$, but its privacy degrades with each fresh independent observation.

\paragraph{Exploiting domain knowledge.} 
\sysname{} requires the user agent to identify sensitive values (the \emph{roots}) --- a baseline assumption shared with prior sanitizers~\cite{chowdhury2025,wu2025capecontextawarepromptperturbation}, typically discharged via NER or explicit user declaration. Beyond this, \sysname{} can exploit two progressive levels of public domain knowledge to improve the privacy-utility tradeoff. The structural knowledge needed is increasingly available: MCP tool schemas~\cite{anthropic2024mcp} declare parameter types and computations (e.g., \texttt{computeBMI(height, weight)}), and many derived values rely on standardized public formulas (BMI, eGFR, HOMA-IR, FIB-4) that a domain-aware agent can recognize without service cooperation.
 
\begin{enumerate}
    \item \textbf{Inter-root dependencies} (Level 1): identifying which declared roots are deterministic functions of others (e.g., BMI from height and weight) lets dependent values inherit privacy from their parents at zero cost (Theorem~\ref{thm:implied-privacy}), freeing budget for truly independent roots.
    \item \textbf{Target sensitivity} (Level 2): any indication of the target function's form --- whether the exact published formula or an approximation $\hat{f}$ --- enables sensitivity-weighted budget allocation, $\varepsilon_i \propto (|h_i| \cdot D_i)^{1/2}$, which minimizes target error by directing budget to the roots that matter most (Sec.~\ref{sec:opt-budget}).
\end{enumerate}
 
\noindent Each level requires only public information (tool schemas, published formulas, population statistics) and improves the privacy-utility tradeoff without weakening the privacy guarantee.

\paragraph{Contributions.}
\begin{itemize}
    \item We formalize multi-turn agentic privacy and prove \sysname{}'s guarantee depends only on root sanitization parameters, independent of adversary function choices or query count (Thm.~\ref{thm:privacy-guarantee}).
    \item We prove independently noising a derived value amplifies parent-root distinguishability by the function's Lipschitz constant (Thm.~\ref{thm:privacy-leakage}); derived values from noised roots inherit privacy at zero marginal cost (Thm.~\ref{thm:implied-privacy}).
    \item On 8 medical templates: \sysname{} achieves $2.6\times$ lower target wMAPE than independent noising at $\varepsilon = 0.1$, $B = (2k{+}1)\varepsilon$ ($7.8\%$ vs.\ $20.3\%$), widening to $3.8\times$ at $(3k{+}1)\varepsilon$, with corresponding $2.2$--$3.3\times$ reductions in clinical Risk Class Error.
    % \item On 8 medical templates: \sysname{} achieves $2.6\times$ lower target wMAPE than independent noising at $\varepsilon = 0.1$, $B = (2k{+}1)\varepsilon$ ($7.8\%$ vs.\ $20.3\%$), widening to $3.8\times$ at $(3k{+}1)\varepsilon$.
    \item Under matched per-root noise ($\varepsilon_r = 0.1$), MAP reconstruction degrades $\mathcal{M}$-All from $21.8\%$ to $6.2\%$ wMAPE as queries grow $1{\to}16$ ($3.5\times$); \sysname{} is invariant.
    \item These findings transfer end-to-end: $21{,}600$ LLM tool-call sessions (GPT-5.4 nano) reproduce RQ1's ordering and per-template structure within bootstrap noise (Sec.~\ref{sec:experiments-rq4}).
\end{itemize}

We make the code to reproduce our experiments available at: \url{https://github.com/danshumaan/rootguard}

%% file: sections/2_Background.tex
\section{Background}

\subsection{Preliminaries}
\label{sec:preliminaries}
 
\begin{definition}[Metric Differential Privacy~\cite{Chatzikokolakis2013BroadeningTS}]
\label{def:mdp}
Let $(\mathcal{X}, d_{\mathcal{X}})$ be a metric space. A randomized mechanism $\mathcal{M}: \mathcal{X} \to \mathcal{Z}$ satisfies $\varepsilon$-metric differential privacy ($\varepsilon$-mDP) if for all $x, x' \in \mathcal{X}$ and all measurable $S \subseteq \mathcal{Z}$:
\[
    \Pr[\mathcal{M}(x) \in S] \leq e^{\varepsilon \cdot d_{\mathcal{X}}(x, x')} \Pr[\mathcal{M}(x') \in S].
\]
\end{definition}
 
\noindent The privacy loss scales with the distance $d_{\mathcal{X}}(x, x')$ between inputs, providing stronger guarantees for nearby values. This generalizes standard $\varepsilon$-DP (which corresponds to the discrete metric $d(x,x') = \mathbf{1}[x \neq x']$) to structured domains where the distance between values is meaningful.

\subsection{Threat Model}
\label{sec:threat-model}

A user agent $\mathcal{U}$, holding private values $X = (x_1, \ldots, x_k)$, engages a service $\mathcal{A}$ to obtain a useful output --- a diagnostic, recommendation, or classification computed from $X$. The user wants $\mathcal{A}$'s output but does not trust $\mathcal{A}$ with the raw values: $\mathcal{A}$ may log queries, share data with partners, be breached, or be coerced through legal processes~\cite{openai-nyt}. We model $\mathcal{A}$ as an adversary that performs the user's task while concurrently attempting to reconstruct $X$ from whatever values cross the trust boundary. Concretely, $\mathcal{A}$ accepts sanitized inputs from $\mathcal{U}$ and returns a task response $T(\hat{X})$ on those inputs (e.g., the FIB-4 liver-fibrosis risk). $T$ decomposes as $T = h \circ g$, where $g$ is some intermediate value (such as published clinical formulas) and $h$ is an opaque threshold, logic, or model the service uses to determine risk, and is not known to the user. The user agent's job is to control what crosses the trust boundary so that $\mathcal{A}$ can compute its task, but learns no more about $X$ than the post-processing theorem permits.

\paragraph{Adversary goal.} The adversary's goal is to reconstruct the user's true root values $X$ from sanitized releases within $t$ turns. It chooses a set of functions $F$ to maximize its reconstruction advantage, and may use all released values and knowledge of the sanitization pipeline to estimate the roots. In particular, if the user agent independently noises both a root $x_i$ and a derived value $f(x_i)$, the adversary obtains two independent noisy views of $x_i$, which it can combine to reduce uncertainty about the true value (Sec.~\ref{sec:redundancy-noising}). As the adversary must complete the task in $t$ turns, we assume without loss of generality that it gathers reconstruction data for the first $t - 1$ turns and obtains the roots needed to compute the target in the final turn.

\paragraph{Adversary capabilities.} $\mathcal{A}$ controls the interaction over $t$ turns. In each of the first $t - 1$ turns, it can: \textbf{(i)} request a private root value $x_i$ that has not yet been disclosed, \textbf{(ii)} force the user agent to compute $f(z_1, \dots, z_d)$ where $ z_i\in X$, for some arbitrary function $f$ specified by $\mathcal{A}$ and \textbf{(iii)} re-request a value that has already been disclosed. In the final turn, $\mathcal{A}$ requests the roots required for task computation and applies the target function $T$ on the sanitized values it receives. The user does not compute $T$ itself; computing $T(\hat{X})$ and the downstream clinical interpretation is $\mathcal{A}$'s service. This capability model captures both prompt injection attacks~\cite{greshake2023indirect,zhan2024injecagent}, where a compromised service forces the agent to execute arbitrary computations, and malicious-by-design services that deliberately probe the user's data.

\paragraph{Adversary knowledge.} The adversary has full knowledge of $\mathcal{U}$'s sanitization scheme: the mechanism, its parameters, and the per-root budget allocation (if any). Optionally, $\mathcal{A}$ may possess population-level statistics (e.g., mean and variance of each root in the general population), which it can use as a prior.

\paragraph{Public vs.\ private knowledge.} The service's task computation $T(\hat{X})$ typically decomposes into a known intermediate $g(\hat{X})$ --- a published clinical formula, an MCP tool schema~\cite{anthropic2024mcp}, or any well-known computation a domain-aware agent can identify (e.g., the FIB-4 score $g = \text{age} \cdot \text{AST} / (\text{PLT} \cdot \sqrt{\text{ALT}})$) --- followed by downstream interpretation that the user does not know (risk thresholds, classification rules, recommendation logic). The user agent exploits the public form of $g$ to allocate sanitization budget across roots (Sec.~\ref{sec:opt-budget}). The opaque downstream interpretation is irrelevant for allocation as $T$ is well-aligned with $g$'s accuracy: minimizing error in $g$ is a faithful proxy for task utility. This asymmetry --- public intermediate $g$, private root values, opaque service-side $T$ --- defines the operational regime where dependency-aware sanitization applies.
 
\paragraph{Adversary strategy.} Given the $t-1$ observations gathered during the interaction, the adversary computes the Maximum A Posteriori (MAP) estimate of the root values, combining the mechanism likelihood with any prior knowledge of the population. Under independent noising ($\mathcal{M}$-All), each observation contributes a fresh factor to the likelihood; under \sysname{}, repeated queries return identical cached values, collapsing the likelihood to a single factor regardless of $t$.

% Given the $t - 1$ observations gathered during the interaction, the adversary computes the Maximum A Posteriori (MAP) estimate of the root values:
% \[
%     \hat{X}_{\mathrm{MAP}} = \arg\max_{X} \prod_{i} p_{\mathcal{M}}(\tilde{y}_i \mid X) \cdot \pi(X),
% \]
% where $p_{\mathcal{M}}$ is the mechanism's likelihood and $\pi$ is the adversary's prior (uniform or informed). Under independent noising ($\mathcal{M}$-All), each observation contributes a fresh factor to the likelihood; under \sysname{}, repeated queries return identical cached values, collapsing the likelihood to a single factor regardless of $t$.
 
\paragraph{Benign setting.} In a less adversarial scenario, $\mathcal{A}$ requests only task-relevant values and does not attempt reconstruction. The privacy concern in this case is passive: services may log released values, which could be later exposed through data breaches or subpoenas. \sysname{}'s guarantees hold in both settings --- the post-processing property does not depend on the adversary's intent.

\paragraph{Scope of evaluation.} Our experiments (Sec.~\ref{sec:experiments-rq2}) evaluate repeated direct root queries. Under independent noising, nonlinear derived-function queries could further amplify the adversary's advantage (Sec.~\ref{sec:nonlinear-amplification}), making our reported numbers a lower bound on its vulnerability. Under \sysname{}, the adversary's function choice is irrelevant: all releases are post-processing of the initial noised roots (Theorem~\ref{thm:privacy-guarantee}).

\subsection{Setting}
\label{sec:setting}
Given the adversary model above, we formalize the interaction between $\mathcal{U}$ and $\mathcal{A}$ as a directed acyclic graph (DAG) $G = (V, E)$, constructed incrementally as the conversation progresses.
 
\begin{itemize}
    \item \textbf{Root nodes.} The user's private input values $X = (x_1, \ldots, x_k)$, which initialize the vertex set $V$. Each $x_i$ takes values in a bounded domain $\mathcal{X}_i = [\ell_i, u_i] \subset \mathbb{R}$ with domain width $D_i = u_i - \ell_i$.
    \item \textbf{Computed nodes.} Derived values $Y = \{y_1, \ldots, y_l\}$, where each $y_j = f_j(\mathrm{parents}(y_j))$ for some function $f_j$ specified by $\mathcal{A}$, with $\mathrm{parents}(y_j) \subseteq X \cup Y$. Each dependency $(p, y_j)$ forms a directed edge in $E$.
    \item \textbf{Target node.} The value $T(X)$ that completes the task --- a known, deterministic function of the roots (e.g., risk thresholding the liver-fibrosis score $\text{FIB4} = \text{age} \cdot \text{AST} / (\text{PLT} \cdot \sqrt{\text{ALT}})$). In the final turn, $\mathcal{A}$ requests the roots needed for task computation, computes $T(\hat{X})$ on the sanitized values it receives, applies its downstream interpretation, and returns the resulting clinical decision to the user.
\end{itemize}
 
\noindent During the interaction, each request from $\mathcal{A}$ adds a node to $G$. The user agent $\mathcal{U}$ computes the requested value and releases a sanitized version according to its sanitization scheme. The DAG serves as a stateful record of what has been released and how values depend on one another.
 
\paragraph{User goal.} The user agent must release sanitized values that enable accurate computation of $T(\hat{X})$ on $\mathcal{A}$'s side while minimizing the adversary's ability to reconstruct the true root values $X$. The user may exploit public domain knowledge (tool schemas, published formulas) to improve its privacy-utility tradeoff (Sec.~\ref{sec:exploiting-dependencies}).
 
\paragraph{Remark on NER} In a full deployment, a named entity recognition (NER) model identifies sensitive tokens in the user's prompt before sanitization. Following Pr$\epsilon\epsilon$mpt~\cite{chowdhury2025}, we assume perfect NER throughout this paper and focus on the sanitization mechanism itself, as it is orthogonal to the privacy-utility tradeoff we study.

%% file: sections/3_Method.tex
%%%%%%%%%%%%%%%%%%%%%%%%%%%%%%%%%%%%%%%%%%%%%%%%%%%%%
%%%%%%%%%%%%%%%%%%%%%%%%%%%%%%%%%%%%%%%%%%%%%%%%%%%%%
%%%%%%%%%%%%%%%%%%%%%%%%%%%%%%%%%%%%%%%%%%%%%%%%%%%%%
\section{Method}
\label{sec:method}
% \sarthak{We should definitely discuss the precomputation of DAG somewhere are modeling conversation between the user and remote LLM and precomputing it, implies fixing the conversation}

\subsection{System Overview}
\label{sec:system}
 
\sysname{} consists of a \emph{DAG controller} that manages privacy state across the conversation, and a \emph{prompt sanitizer} that formats each release into natural language. Figure~\ref{fig:rootguard_overview} illustrates the end-to-end flow. 
 
\paragraph{Initialization.} Before the conversation begins, the controller sanitizes each root independently (Sec.~\ref{sec:sanitization}) and caches the noised values $\hat{X} = (\hat{x}_1, \ldots, \hat{x}_k)$. If the user has knowledge of dependencies among the declared roots (e.g., BMI depends on height and weight), the controller identifies the true independent roots, recomputes the dependent roots from their noised parents, and updates the cache accordingly (Sec.~\ref{sec:exploiting-dependencies}). If the user knows the intermediate $g$ that the service's task $T$ depends on (Sec.~\ref{sec:threat-model}), the controller computes per-root sensitivities of $g$ and allocates budgets to minimize error in $g(\hat{X})$ (Sec.~\ref{sec:opt-budget}); otherwise, budget is distributed uniformly.
 
% \paragraph{Per-turn operation.} When $\mathcal{A}$ requests a value, the controller adds a node to $G$ and returns either the cached noised root or a derived value computed deterministically from cached parents. No fresh noise is ever drawn after initialization. By the post-processing theorem of differential privacy~\cite{dwork2006calibrating}, no additional privacy cost is incurred beyond the initial root sanitization. A separate \emph{prompt sanitizer} (not evaluated in this work) wraps each released value into a natural-language response; our experiments bypass this layer to isolate the privacy-utility tradeoff (Sec.~\ref{sec:exp-setup}).

\paragraph{Per-turn operation.} Each request from $\mathcal{A}$ is answered with the cached noised root or a derived value computed deterministically from cached parents; no fresh noise is drawn after initialization. By the post-processing theorem~\cite{dwork2006calibrating}, no additional privacy cost is incurred beyond the initial root sanitization.

%%%%%%%%%%%%%%%%%%%%%%%%%%%%%%%%%%%%%%%%%%%%%%%%%%%%%
%%%%%%%%%%%%%%%%%%%%%%%%%%%%%%%%%%%%%%%%%%%%%%%%%%%%%
%%%%%%%%%%%%%%%%%%%%%%%%%%%%%%%%%%%%%%%%%%%%%%%%%%%%%
\subsection{Sanitization}
\label{sec:sanitization}

\paragraph{Baseline: independent noising ($\mathcal{M}$-All).} Without structural knowledge of the roots, the default sanitization strategy is to treat each request independently. When $\mathcal{A}$ requests a value, the user computes it from the original (unsanitized) values and applies fresh, independent noise. This produces independent noisy releases of values that share the same underlying roots, giving the adversary multiple observations to combine. The structural redundancy this creates is exploitable (Sec.~\ref{sec:redundancy-noising}), and the total privacy cost grows with the number of releases.
 
\paragraph{\sysname{}.} Our proposed framework improves on this by exploiting the user's knowledge of the root nodes. It sanitizes the roots \emph{once} at the start of the interaction and computes all subsequent releases deterministically from the noised roots. By the post-processing theorem~\cite{dwork2006calibrating}, no additional privacy cost is incurred regardless of what $\mathcal{A}$ requests. The method's effectiveness improves further with additional structural knowledge (Sec.~\ref{sec:exploiting-dependencies}).
 
Both methods sanitize values according to two categories of root, following Pr$\epsilon\epsilon$mpt~\cite{chowdhury2025}.
 
\paragraph{Category-I roots ($\tau_I$).} Roots whose utility depends only on format (e.g., names, identifiers) are sanitized via format-preserving encryption: $\hat{x}_i \gets E_F(K, N_{\tau_i}, x_i)$.
 
\paragraph{Category-II roots ($\tau_{II}$).} Roots whose utility depends on numerical value (e.g., lab measurements) are sanitized via a differentially private mechanism: $\hat{x}_i \gets \mathcal{M}_{\epsilon_i}(x_i)$, with total budget $\sum_{i \in [\tau_{II}]} \epsilon_i = B$. Each root has domain $[\ell_i, u_i]$ with width $D_i$, discretized into $m$ candidates with spacing $\delta_i = D_i/(m{-}1)$.
 
\paragraph{Index-space normalization.}
All Category-II mechanisms operate in a shared index space $\{0, \ldots, m{-}1\}$. The true value $x_i$ is mapped to a normalized index $t_i = (m-1)(x_i - \ell_i)/D_i$, noise is drawn in index units according to the mechanism's distribution, the result is clamped to $[0, m{-}1]$ and rounded to the nearest integer $s$, and the output is mapped back to domain units as $x_i' = s \cdot \delta_i + \ell_i$. The sensitivity in index space is $1$ for all three mechanisms used in our experiments (Sec.~\ref{sec:exp-setup}), so a unit of $\epsilon$ provides the same worst-case index-space distinguishability for every node and is directly comparable across mechanisms, regardless of domain width. Without this normalization, the noise scale in domain units ($\delta_i = D_i/(m-1)$) would vary across nodes, making the same $\epsilon$ provide vastly different physical protection (e.g., $\delta = 0.013$ for hemoglobin vs.\ $\delta = 2.67$ for triglycerides). This enables direct comparisons across nodes and a meaningful total budget $B = \sum_r \epsilon_r$ for principled budget allocation.

% \sarthak{Examples for each category needed.}

\subsection{Exploiting Structural Knowledge}
\label{sec:exploiting-dependencies}
 
\sysname{}'s privacy and utility improve with the user's structural knowledge, exploited at two progressive levels.
 
\paragraph{Level 1: Inter-root dependencies.} The user knows that some declared roots depend on others (e.g., BMI depends on height and weight). True independent roots are identified, and dependent roots are recomputed from noised parents. This reduces the number of noise draws from $k$ to $k' < k$. By the Implied Privacy theorem (Theorem~\ref{thm:implied-privacy}), dependent roots inherit privacy for free.
 
\paragraph{Level 2: Target sensitivity.} The user wants the service's clinical decision $T(\hat{X})$ to be accurate, but does not know the threshold logic that maps the service's intermediate computation to the final decision. What the user does have is knowledge of an intermediate $g$ that $T$ depends on --- typically a published clinical formula (e.g., the FIB-4 score $g = \text{age} \cdot \text{AST} / (\text{PLT} \cdot \sqrt{\text{ALT}})$, on which liver-fibrosis risk classification depends) or any computation declared in an MCP tool schema. Because $T$ is well-aligned with $g$'s accuracy, minimizing error in $g(\hat{X})$ is a faithful proxy for task utility. This enables sensitivity-weighted budget allocation (Sec.~\ref{sec:opt-budget}).
 
\paragraph{Configurations evaluated.} We compare three configurations that progressively exploit structural knowledge: $\mathcal{M}$-All (independent noising; no structural knowledge), $\mathcal{M}$-Roots (root-only noising with uniform allocation $\epsilon_i = B/k$), and $\mathcal{M}$-Opt (root-only noising with sensitivity-weighted allocation; full \sysname{}). The progression $\mathcal{M}$-All~$\to$~$\mathcal{M}$-Roots isolates the gain from eliminating redundant noising; $\mathcal{M}$-Roots~$\to$~$\mathcal{M}$-Opt isolates the gain from dependency-aware allocation. The levels degrade gracefully: when structural knowledge is unavailable, the unknown dependencies are treated independently and the method reduces to the prior level.

%%%%%%%%%%%%%%%%%%%%%%%%%%%%%%%%%%%%%%%%%%%%%%%%%%%%%
%%%%%%%%%%%%%%%%%%%%%%%%%%%%%%%%%%%%%%%%%%%%%%%%%%%%%
%%%%%%%%%%%%%%%%%%%%%%%%%%%%%%%%%%%%%%%%%%%%%%%%%%%%%

\subsection{Sensitivity-Weighted Budget Allocation}
\label{sec:opt-budget}

\noindent The user's task utility is the accuracy of the service's clinical decision $T(\hat{X})$, which depends on an intermediate $g(\hat{X})$ via threshold logic the user does not know (Sec.~\ref{sec:threat-model}). Because $T$ is well-aligned with $g$'s accuracy, we minimize the expected absolute error of $g$ at allocation time, subject to the total budget constraint. Without loss of generality, let $x = (x_1, \dots,x_k)$ contain only $\tau_{II}$ roots (Category-I roots are sanitized via FPE under a separate security parameter and do not consume the DP budget). Let $g\colon \mathbb{R}^k \to \mathbb{R}$ be a known, differentiable intermediate that the service consumes en route to its decision. Suppose root $i$ has domain $[\ell_i, u_i]$, with width $D_i = u_i - \ell_i > 0$, discretized into $m$ candidates of width $\delta_i = D_i/(m-1)$. The partial derivative $\partial g / \partial x_i$ measures how sensitively $g$ responds to perturbations in root $i$. We evaluate the sensitivity at the population mean of each root:
\[
    h_i = \left|\frac{\partial g}{\partial x_i}(\mu)\right|,
\]
where the population means $\mu = (\mu_1, \dots, \mu_k)$ are a single publicly available vector (in our experiments, we use the NHANES data held out from the test set; see Sec.~\ref{sec:experiments-rq1}). The computed weights at the population means reflect how $g$ responds to perturbations around a typical patient, rather than at the pathological worst case. Since $\mu$ is computed from a public dataset independent of the private data being sanitized, the allocation remains fixed across all patients and no information flows from patient to mechanism. 

A first-order expansion of the error $\Delta = g(x') - g(x)$ in the per-root noise $\eta_i = x'_i - x_i$ gives $\Delta \approx \sum_i h_i \eta_i$. Minimizing the upper bound $\mathbb{E}[|\Delta|] \leq \sum_i |h_i|\,\bar{\eta}_i(\epsilon_i)$ subject to $\sum_i \epsilon_i = B$ and $\epsilon_i \geq \epsilon_{\min}$, where $\bar{\eta}_i$ is the worst-case expected noise over grid positions, yields the closed-form Laplace solution $\epsilon_i \propto \sqrt{|h_i|\,\delta_i}$ (full derivation, including saturation handling and the convexity argument, in App.~\ref{app:budget-derivation}). The objective is separable and solved numerically for the discrete Exponential, Bounded Laplace, and Staircase mechanisms; RQ3 (Sec.~\ref{sec:experiments-rq3}) verifies the $\sqrt{|h_i|\,\delta_i}$ scaling holds empirically for all three.

%% file: sections/4_PrivacyAnalysis.tex
\section{Privacy Analysis}

\noindent We define a privacy game to capture an adversary's ability to distinguish between two sanitized root vectors.

%%%%%%%%%%%%%%%%%%%%%%%%%%%%%%%%%%%%%%%%%%%%%%%%%%%%%
%%%%%%%%%%%%%%%%%%%%%%%%%%%%%%%%%%%%%%%%%%%%%%%%%%%%%
%%%%%%%%%%%%%%%%%%%%%%%%%%%%%%%%%%%%%%%%%%%%%%%%%%%%%
\subsection{DAG Privacy Game $G_{\mathrm{DAG}}$}

The adversary $\mathcal{A}$ chooses two root vectors $X^0$ and $X^1$ with the same leakage and any arbitrary set of functions $F=\{f_1, \dots, f_l\}$. The user sanitizes a randomly selected root vector ($\hat{X}^b$) and returns it to the adversary. Since the adversary controls the interaction with the user agent, it can compute any $f_i$ on the released root vector $\hat{X}^b$ to improve its advantage, and constructs the conversation DAG $\hat{G}_b$. Given $X^0$ and $X^1$, the adversary must guess which root vector was selected. 

\begin{algorithmic}[1]
\Statex \underline{Initialize}:
\State $K \leftarrow S()$
\State $b \overset{{\scriptscriptstyle\$}}{\leftarrow} \{0,1\}$ \textcolor{blue}{$\rhd$} Select a random bit
\Statex \underline{$\mathrm{Challenge}(X^0, X^1,F)$}: \textcolor{blue}{$\rhd$} Adversary selects two root vectors and an arbitrary set of functions
\State $L_0 \gets L(X^0, F, C)$ ; $L_1 \gets L(X^1, F, C)$
\Statex \textcolor{blue}{$\rhd$} $L$ is the leakage function; $F$ is the chosen function set
\State \textbf{if} $L_0 \neq L_1$ \textbf{then return} $\bot$
\Statex \textcolor{blue}{$\rhd$} Only root pairs with the same leakage are valid
\Statex \underline{Sanitize and Construct}:
\State \textbf{for} $i \in \mathcal{I}$: $\hat{x}^b_i \gets E_F(K, N_{\tau_i}, x^b_i)$ \textcolor{blue}{$\rhd$} FPE for category-I roots
\State \textbf{for} $i \in \mathcal{II}$: $\hat{x}^b_i \gets \mathcal{M}_{\epsilon_i}(x^b_i)$ \textcolor{blue}{$\rhd$} mDP for category-II roots; sampled once, reused
\State \textbf{return} $\hat{X}^b$ \hfill \textcolor{blue}{$\rhd$} User returns the sanitized roots. Adversary may compute the full graph $\hat{G}_b \gets C(\hat{X}^b, F)$ from sanitized roots and use any strategy to guess $b$.
% \Statex \textcolor{blue}{$\rhd$} Return the sanitized graph to the adversary
\Statex \underline{$\mathrm{Finalize}(b')$}:
\State \textbf{return} $[b' = b]$ \hfill \textcolor{blue}{$\rhd$} $b'$ is the adversary's guess for $b$
\end{algorithmic}

\medskip
The adversary's advantage is defined as: 
\begin{align}
\label{alg:full-dag-adv}
    \mathrm{Adv}_{\mathrm{DAG}}(\mathcal{A}) = |2\Pr[G_{\mathrm{DAG}}(\mathcal{A}) = 1] - 1|
\end{align}

%%%%%%%%%%%%%%%%%%%%%%%%%%%%%%%%%%%%%%%%%%%%%%%%%%%%%
%%%%%%%%%%%%%%%%%%%%%%%%%%%%%%%%%%%%%%%%%%%%%%%%%%%%%
%%%%%%%%%%%%%%%%%%%%%%%%%%%%%%%%%%%%%%%%%%%%%%%%%%%%%

%%%%%%%%%%%%%%%%%%%%%%%%%%%%%%%%%%%%%%%%%%%%%%%%%%%%%
%%%%%%%%%%%%%%%%%%%%%%%%%%%%%%%%%%%%%%%%%%%%%%%%%%%%%
%%%%%%%%%%%%%%%%%%%%%%%%%%%%%%%%%%%%%%%%%%%%%%%%%%%%%
\subsection{Privacy Guarantee}

\begin{theorem}[DAG Privacy Guarantee] Consider an adversary $\mathcal{A}$ in the setting of the full-graph privacy game $G_{\mathrm{DAG}}$. Let $(X^0, X^1, F)$ be the adversary's chosen root vectors and function set, with roots of type $\tau_I$ and $\tau_{II}$. For the DAG privacy game with the privacy budget $\sum_{i\in[\tau_{II}]} \epsilon_i = B$ and the security parameter $\kappa$, 

\begin{align}
    \mathrm{Adv}_{\mathrm{DAG}}(\mathcal{A}) \leq \sum_{i\in[\tau_{II}]} e^{\frac{\epsilon_i}{\delta_i}|x^0_i - x^1_i|} + \big|[\tau_{I}]\big|\cdot\mathrm{negl}(\kappa).
\end{align}
\label{thm:privacy-guarantee}
The bound depends only on the root sanitization parameters $\epsilon_i$ and the root distances $|x^0_i - x^1_i|$. It is independent of the adversary's choice of functions $F$ as it gains no additional distinguishing power from computing arbitrary functions on the sanitized roots.
\end{theorem}

\paragraph{Proof Sketch} Since the adversary's view is entirely determined by $\hat{X}^b$, we bound the advantage directly via a hybrid argument over the roots. Define hybrids $X^{(0)} = X^0, \ldots, X^{(k)} = X^1$, where each consecutive pair differs in a single root. By the triangle inequality, the total advantage is bounded by the sum of single-root advantages. Each single-root advantage is bounded by FPE security (for $\tau_I$ roots) or the $\epsilon_j$-mDP guarantee (for $\tau_{II}$ roots). Since the bound involves only root parameters and distances, it holds for every choice of $F$, including the adversary's optimal choice. The full proof is given in App.~\ref{app:thm:privacy-guarantee}.
 
\begin{remark}[Contrast with independent noising]
\label{rem:m-all-contrast}
Under $\mathcal{M}$-All, the user independently noises each requested value and returns all noised values to the adversary. The adversary therefore receives independently noised versions of both $x_i$ and $f(x_i)$, obtaining multiple independent observations about the same root. The privacy cost of these additional releases is not captured by the root-only bound above and is analyzed in Sec.~\ref{sec:privacy-leakage}.
\end{remark}

%%%%%%%%%%%%%%%%%%%%%%%%%%%%%%%%%%%%%%%%%%%%%%%%%%%%%
%%%%%%%%%%%%%%%%%%%%%%%%%%%%%%%%%%%%%%%%%%%%%%%%%%%%%
%%%%%%%%%%%%%%%%%%%%%%%%%%%%%%%%%%%%%%%%%%%%%%%%%%%%%
\subsection{Privacy Leakage due to Independent Noising}
\label{sec:privacy-leakage}

Under the index-space mechanisms described in Sec.~\ref{sec:sanitization}, every node is mapped to the same index space $\{0, \ldots, m{-}1\}$ before noising. Independently releasing $\mathcal{M}_1(x)$ and $\mathcal{M}_2(f(x))$, each with parameter $\epsilon$, yields $2\epsilon$-mDP with respect to the index-space metric $d(t,t') = |t - t'|$. Since every node's index distance is bounded by $(m{-}1)$, this guarantee appears uniform across nodes.

However, independent noising creates two distinct sources of additional leakage that this guarantee does not capture: amplification through nonlinear functions (Sec.~\ref{sec:nonlinear-amplification}) and accumulation through repeated queries (Sec.~\ref{sec:redundancy-noising}).
 
%%%%%%%%%%%%%%%%%%%%%%%%%%%%%%%%%%%%%%%%%%%%%%%%%%%%%
\subsubsection{\textbf{Amplification through nonlinear functions.}}
\label{sec:nonlinear-amplification}
 
Consider two root values $x_0, x_1$ separated by $d$ index steps: $|t_{x_0} - t_{x_1}| = d$. In domain space, this corresponds to $|x_0 - x_1| = d \cdot D_x / (m{-}1)$. The derived values' index-space distance is:
\[
    |t_{f(x_0)} - t_{f(x_1)}| = \frac{m-1}{D_y}|f(x_0) - f(x_1)| \leq \frac{L \cdot D_x}{D_y} \cdot d,
\]
where $L$ is the Lipschitz constant of $f$ and $D_y$ is the derived domain width. The effective \emph{index-space amplification factor} is $L \cdot D_x / D_y$.
 
For linear functions this factor is exactly $1$ (normalization absorbs linear scaling). For nonlinear functions it can far exceed $1$: e.g., $f(x) = 1/x$ on $[1, 100]$ yields ${\sim}100\times$ amplification. The functions underlying the medical templates used in our experiments (Sec.~\ref{sec:exp-setup}) exhibit this behavior. We formalize the amplification in the following theorem:

\begin{theorem}[Distinguishability Amplification]
\label{thm:privacy-leakage}
Let $(\mathcal{X}, d_{\mathcal{X}})$ and 
$(\mathcal{Y}, d_{\mathcal{Y}})$ be metric spaces. Let 
$f\colon \mathcal{X} \to \mathcal{Y}$ be a Lipschitz function 
satisfying, for all $x, x' \in \mathcal{X}$,
\[
    \alpha \cdot d_{\mathcal{X}}(x, x') 
    \;\leq\; 
    d_{\mathcal{Y}}(f(x), f(x')) 
    \;\leq\; L \cdot d_{\mathcal{X}}(x, x'),
\]
for constants $\alpha, L > 0$. Let $\mathcal{M}\colon \mathcal{Y} \to \mathcal{Y}$ be an $\epsilon$-mDP mechanism on $(\mathcal{Y}, d_{\mathcal{Y}})$. Define $\psi = \mathcal{M} \circ f \colon \mathcal{X} \to \mathcal{Y}$. Then for all $x, x' \in \mathcal{X}$ and all measurable $T \subseteq \mathcal{Y}$:
\[
    \Pr[\psi(x) \in T] 
    \;\leq\; 
    e^{\epsilon\, L\, d_{\mathcal{X}}(x, x')} 
    \Pr[\psi(x') \in T].
\]
That is, the distinguishability level~\cite{Chatzikokolakis2013BroadeningTS} between any two inputs $x, x'$ is amplified by a factor of $L$ relative to the nominal parameter $\epsilon$. As $L$ is a supremum, the bound is tight.
\end{theorem}
 
The proof is available in App.~\ref{app:thm:privacy-leakage}. In index space, the amplification factor becomes $L_{\mathrm{idx}} = L \cdot D_x / D_y$.

\subsubsection{Accumulation through repeated releases.}
\label{sec:redundancy-noising}
 
When $\mathcal{M}$-All independently sanitizes both $x$ and $f(x)$, the joint distinguishability is additive: $\epsilon(1 + L)\,d_\mathcal{X}(x,x')$. More generally, with $d$ independently noised derived nodes $f_i(x)$ (each $(m_i, L_i)$-bi-Lipschitz, parameter $\epsilon$), the joint level is $\epsilon(1 + \sum_{i=1}^{d} L_i)\,d_\mathcal{X}(x,x')$. Under the worst-case adversary, even repeated root queries (where $L_{\mathrm{idx}} = 1$) produce fresh noise draws, enabling MAP estimation that reduces reconstruction error as $\sqrt{t}$.
 
\paragraph{MAP reconstruction.} The accumulation of independent observations has a concrete operational consequence. An adversary who observes $t-1$ independent releases of a root $x_r$ can compute the Maximum A Posteriori estimate:
\[
    \hat{x}_{\mathrm{MAP}} = \arg\max_{x_r} \prod_{i=1}^{t-1} p_{\mathcal{M}}(\tilde{y}_i \mid x_r, \epsilon) \cdot \pi(x_r),
\]
where $p_{\mathcal{M}}$ is the mechanism's PMF and $\pi$ is any prior. Under \sysname{}, all queries return the same deterministic output $\hat{x}_r$, so the likelihood reduces to a single factor $p_{\mathcal{M}}(\hat{x}_r \mid x_r, \epsilon_r)$ regardless of $t$. We evaluate this attack empirically in Sec.~\ref{sec:experiments-rq2}.

%%%%%%%%%%%%%%%%%%%%%%%%%%%%%%%%%%%%%%%%%%%%%%%%%%%%%
%%%%%%%%%%%%%%%%%%%%%%%%%%%%%%%%%%%%%%%%%%%%%%%%%%%%%
%%%%%%%%%%%%%%%%%%%%%%%%%%%%%%%%%%%%%%%%%%%%%%%%%%%%%
 
%%%%%%%%%%%%%%%%%%%%%%%%%%%%%%%%%%%%%%%%%%%%%%%%%%%%%
%%%%%%%%%%%%%%%%%%%%%%%%%%%%%%%%%%%%%%%%%%%%%%%%%%%%%
%%%%%%%%%%%%%%%%%%%%%%%%%%%%%%%%%%%%%%%%%%%%%%%%%%%%%

\subsection{Implied Privacy} The preceding section quantifies the cost of independent noising. We now show the converse, a coordinated mechanism that noises only the roots automatically extends privacy to all derived nodes, with no additional expenditure of budget.

\begin{theorem}[Implied Privacy]
\label{thm:implied-privacy}
Let $(\mathcal{X}, d_{\mathcal{X}})$ and $(\mathcal{Y}, d_{\mathcal{Y}})$ be metric spaces. 
Let $\mathcal{M}: \mathcal{X} \to \mathcal{X}$ be an $\epsilon$-metric differentially private ($\epsilon$-mDP) mechanism on $(\mathcal{X}, d_{\mathcal{X}})$ that adds noise to its input.

Let $f: \mathcal{X} \to \mathcal{Y}$ be a bi-Lipschitz function satisfying, for all $x, x' \in \mathcal{X}$,
\[
    \alpha \cdot d_{\mathcal{X}}(x, x') 
    \;\leq\; d_{\mathcal{Y}}(f(x), f(x')) 
    \;\leq\; L \cdot d_{\mathcal{X}}(x, x') ,
\]
for some constants $\alpha, L > 0$. Define a post-processing mechanism $\phi = f \circ \mathcal{M} : \mathcal{X} \to \mathcal{Y}.$ Then, for all $x, x' \in \mathcal{X}$ and all measurable sets $T \subseteq \mathcal{Y}$,
\[
    \Pr[\phi(x) \in T]
    \;\leq\;
    e^{\frac{\epsilon}{\alpha} \, d_{\mathcal{Y}}(f(x), f(x'))}
    \Pr[\phi(x') \in T].
\]
where the probability is taken over the randomness of $\mathcal{M}$. Equivalently, $\phi$ is $\frac{\epsilon}{\alpha}$-mDP with respect to the metric $d_{\mathcal{Y}}(f(x), f(x'))$.

Hence, adding mDP noise directly to $y = f(x)$ in $(\mathcal{Y}, d_{\mathcal{Y}})$ with privacy parameter $\frac{\epsilon}{\alpha}$ provides at least as strong privacy as noising $x$ and then applying $f$.
\end{theorem}

The proof is available in App.~\ref{app:thm:implied-privacy}. In index space, the inherited parameter is $\epsilon \cdot D_y / (\alpha \cdot D_x)$; when $D_y/D_x$ is small, strong protection is inherited for free. Combined with Theorem~\ref{thm:privacy-guarantee}, this guarantees that under \sysname{}, repeated queries and derived functions return deterministic post-processing of the initial noised roots, and the adversary's reconstruction power does not grow with $t$.

%% file: sections/5_experiments_worst_case_adversary.tex
\section{Experiments}

We conduct experiments to answer the following research questions:\\
 
\noindent\mygraybox{%
{\renewcommand{\labelenumi}{\textbf{RQ\arabic{enumi}:}}
\begin{enumerate}
\item \textbf{(Utility):} How much does root-only noising with dependency-aware budget allocation reduce target-node error compared to independent per-release noising?
\item \textbf{(Reconstruction):} Under a MAP adversary with repeated queries at matched per-root noise, how much does $\mathcal{M}$-All's privacy degrade compared to \sysname{}?
\item \textbf{(Structural Analysis):} What structural properties of the template and budget allocation determine the advantage provided by dependency-aware noising?
\item \textbf{(Deployment):} Do the mechanism-level findings transfer to a real LLM-driven tool-call deployment?
\end{enumerate}}%
}\\

\subsubsection*{Summary of Findings.}

\paragraph{\textbf{RQ1}} \sysname{} achieves $2.6\times$ lower aggregate wMAPE than $\mathcal{M}$-All at $\varepsilon = 0.1$ (7.8\% vs.\ 20.3\%, Exponential, $B = (2k{+}1)\varepsilon$), widening to $3.8\times$ at $(3k{+}1)\varepsilon$ and at stronger privacy. The wMAPE wins translate to clinical decisions: $2.2\times$ lower Risk Class Error at $(2k{+}1)\varepsilon$ (7.0\% vs.\ 15.7\%), widening to $3.3\times$ at $(3k{+}1)\varepsilon$. Additional adversarial turns improve \sysname{}'s utility while $\mathcal{M}$-All is static.

\paragraph{\textbf{RQ2}} $\mathcal{M}$-All's reconstruction wMAPE drops from $10.2\%$ to $3.2\%$ as $q$ increases from $1$ to $16$ ($\varepsilon_r = 0.1$, Strategy B, informed prior); \sysname{} is invariant. Each adversarial query produces an independent observation that tightens the MAP posterior; \sysname{}'s responses contribute no new information after the first release.

\paragraph{\textbf{RQ3}} \sysname{} wins on every template at every budget; the magnitude of the win tracks whether the target formula compresses or amplifies root noise. Compressive templates (AIP, HOMA, FIB4) yield the largest gains ($5.1\times$, $3.9\times$, $4.4\times$ at $B=(3k{+}1)\varepsilon$); amplifying templates (ANEMIA, CONICITY) yield smaller multiplicative gains on already-low absolute errors. The allocation follows $\varepsilon_i \propto (|h_i| D_i)^{1/2}$ (fitted slope 0.50), is mechanism-independent, and reclaims budget from zero-sensitivity roots.

\paragraph{\textbf{RQ4}} The RQ1--3 ordering and per-template structure transfer cell-for-cell to a real LLM-driven deployment (GPT-5.4 nano~\cite{openai2026gpt54nano} agent harness, 21{,}600 sessions, 0\% session-level compliance failures), confirming that the privacy-utility tradeoff established at the mechanism level holds end-to-end.
 
% \noindent\textbf{Summary of Findings.} Utility is measured via wMAPE and Risk Class Error; reconstruction vulnerability via a MAP adversary with uniform and informed priors at $q \in \{1, 4, 8\}$ repeated queries (Sec.~\ref{sec:metrics}).
 
% \textbf{RQ1:} \sysname{} achieves $2.3\times$ lower aggregate wMAPE than $\mathcal{M}$-All at $\varepsilon = 0.1$ (7.6\% vs.\ 17.1\%, Exponential, $B = (2k{+}1)\varepsilon$), widening at stronger privacy. A double asymmetry emerges: additional adversarial turns improve \sysname{}'s utility while $\mathcal{M}$-All's stays flat. The win rate increases from 27\% to 65\% as budget grows.
 
% \textbf{RQ2:} $\mathcal{M}$-All's reconstruction wMAPE drops from 6.3\% to 3.2\% as $q$ increases from 1 to 8 ($\varepsilon_r = 0.1$, informed prior); \sysname{} is invariant. Two effects compromise $\mathcal{M}$-All: the target release leaks a cross-root constraint at $q = 1$, and repeated queries enable noise averaging as $q$ grows.
 
% \textbf{RQ3:} The advantage is governed by whether the target formula compresses or amplifies root noise: compressive templates yield large gains (FIB4: $5.1\times$, HOMA: $4.7\times$); amplifying templates yield small losses (ANEMIA: $-0.3$pp). The allocation follows $\varepsilon_i \propto (|h_i| D_i)^{1/2}$ (fitted slope 0.50), is mechanism-independent, and reclaims budget from zero-sensitivity roots.

\subsection{Setup}
\label{sec:exp-setup}
\paragraph{Data.} We use CDC's publicly available National Health and Nutrition
Examination Survey (NHANES) 2017--2018~\cite{nhanes2017exam}, which covers   body measures, blood pressure, hematology, cholesterol, triglycerides, glucose, insulin etc. We consider $8$ clinical templates for diagnosis. Each template defines $k$ root attributes (a subset of the NHANES medical measurements) and a deterministic target function $g$ that computes a clinical index used by the service's threshold logic to complete the diagnosis. For example, the FIB4 template uses four roots (age, AST, PLT, ALT) and computes $g(\text{age, AST, PLT, ALT}) = \text{age} \cdot \text{AST} / (\text{PLT} \cdot \sqrt{\text{ALT}})$. This then is thersholded to determine with the risk as ``\texttt{Low}'', ``\texttt{Indeterminate}'' or ``\texttt{High}''. We draw $200$ adult male patients per template from the NHANES population (the per-template subpopulation ranges from $n = 1{,}195$ to $2{,}840$) for RQ1--RQ3, and the first $100$ of these for the deployment evaluation in RQ4. All template details are in App.~\ref{app:med-profiles}.

% \paragraph{Interaction model.} Following the threat model (Sec.~\ref{sec:threat-model}), we evaluate a multi-turn interaction over $t \in \{k{+}1,\; 2k{+}1,\; 3k{+}1\}$ turns. In each of the first $t-1$ turns, the adversary requests a root or an arbitrary function of roots; in the final turn, it requests the roots required to compute the target and applies $g$ on the sanitized values it receives. With $k \in \{2, 3, 4\}$, $t$ ranges from 3 to 13, giving $1$, $2$, or $3$ adversarial queries per root. For RQ1--RQ3, we simulate this interaction directly through mechanism implementations to isolate the privacy-utility tradeoff from NER errors and prompt formatting artifacts; RQ4 (Sec.~\ref{sec:experiments-rq4}) re-runs the same threat model end-to-end through an LLM tool-call pipeline.

\paragraph{Interaction model.} Following the threat model (Sec.~\ref{sec:threat-model}), we evaluate a multi-turn interaction over $t$ turns with total privacy budget $B = t \cdot \varepsilon$. In each of the first $t-1$ turns, the adversary requests a root or an arbitrary function of roots; in the final turn, it requests the roots required to compute $g$ and applies it on the sanitized values it receives. RQ1 sweeps three turn budgets $t \in \{k{+}1,\; 2k{+}1,\; 3k{+}1\}$ across 10 privacy levels $\varepsilon \in [0.005, 5.0]$; with $k \in \{2, 3, 4\}$, this gives $1$, $2$, or $3$ adversarial queries per root before the final turn. RQ2 fixes per-root noise $\varepsilon_r \in \{0.05, 0.1, 0.5, 1.0\}$ and varies the adversarial query count $q \in \{1, 4, 8, 16\}$ directly. For RQ1--RQ3, we simulate this interaction through mechanism implementations to isolate the privacy-utility tradeoff from NER errors and prompt formatting artifacts; RQ4 (Sec.~\ref{sec:experiments-rq4}) re-runs the same threat model end-to-end through an LLM tool-call pipeline.

\paragraph{Mechanisms and configurations.} We evaluate three noise mechanisms (Discrete Exponential~\cite{mcsherry2007mechanism}, Bounded Laplace, Staircase~\cite{geng2015optimal}) under the three configurations of Sec.~\ref{sec:exploiting-dependencies}: $\mathcal{M}$-All (every release noised independently with $\varepsilon$ budget), $\mathcal{M}$-Roots (only $k$ roots noised, uniform allocation $\varepsilon_i = B/k$), and $\mathcal{M}$-Opt ($k$ roots noised, sensitivity-weighted allocation $\varepsilon_i \propto \sqrt{|h_i| \cdot \delta_i}$); only the latter two are dependency-aware. All operate in normalized index space with $m = 1{,}000$ candidates; mechanism details are in App.~\ref{app:baselines}.

% \paragraph{Budgets and allocation.} All methods share total budget $B = t \cdot \varepsilon$. RQ1 sweeps 10 privacy levels $\varepsilon \in [0.005, 5.0]$ across three adversary turn budgets $t \in \{k{+}1, 2k{+}1, 3k{+}1\}$; RQ2 fixes per-root noise $\varepsilon_r \in \{0.05, 0.1, 0.5, 1.0\}$ and varies adversarial queries $q \in \{1, 4, 8, 16\}$. The three configurations distribute $B$ differently: $\mathcal{M}$-All assigns $\varepsilon$ to each of the $t$ releases; $\mathcal{M}$-Roots distributes $B$ uniformly across $k$ roots ($\varepsilon_i = t\varepsilon/k$); $\mathcal{M}$-Opt allocates per the sensitivity-weighted solution of Sec.~\ref{sec:opt-budget} ($\varepsilon_i \propto \sqrt{|h_i| \cdot \delta_i}$). Roots saturating at $\varepsilon_{\min}$ have their budget redistributed to active roots.

\paragraph{Budget allocation.} The three configurations distribute $B$ differently: $\mathcal{M}$-All assigns $\varepsilon$ to each of the $t$ releases; $\mathcal{M}$-Roots distributes $B$ uniformly across $k$ roots ($\varepsilon_i = t\varepsilon/k$); $\mathcal{M}$-Opt allocates per the sensitivity-weighted solution of Sec.~\ref{sec:opt-budget}. Roots saturating at $\varepsilon_{\min}$ have their budget redistributed to active roots.

\paragraph{Domain knowledge.} \sysname{} uses two forms of public knowledge, neither dependent on the private patient data: \textit{domain bounds} $[\ell_i, u_i]$ from the per-template NHANES adult male subpopulation, and \textit{population-mean root values} $\mu$ from the same population (excluding test samples) for evaluating sensitivities $h_i = |\partial g / \partial x_i(\mu)|$. The $B$-mDP guarantee holds regardless, as domain knowledge affects only utility.

\paragraph{Text-DP reference.} We additionally include CAPE-F~\cite{wu2025capecontextawarepromptperturbation} as a text-DP reference, adapted to numeric values by constraining decoding to numerical characters. CAPE-F operates under a different DP definition and is not directly comparable on $\varepsilon$; we include it to demonstrate that text-based DP is fundamentally unsuitable for structured numeric data.

\subsection{Metrics}
\label{sec:metrics}
We evaluate each sanitization method along two axes: \emph{numeric accuracy} of the target clinical index, and \emph{clinical decision preservation}. Both are measured at the target node of each template---the final derived value from which the clinical risk class is determined. Let $y$ denote the ground truth target value and $\hat{y}$ the target value obtained from (possibly noised) inputs.

\paragraph{Weighted Mean Absolute Percentage Error (wMAPE)}
For each template with $N$ samples, we compute:
\begin{equation}
\text{wMAPE} = \frac{\sum_{i=1}^{N} |\hat{y}_i - y_i|}{\sum_{i=1}^{N} |y_i|} \times 100\%
\end{equation}
the mean absolute error normalized by the mean ground-truth magnitude~\cite{hyndman2006another,kolassa2007advantages}. A wMAPE of 0\% indicates perfect numeric preservation; values above 100\% indicate that average distortion exceeds the typical value of the index. We use wMAPE rather than per-sample MAPE because MAPE is unstable when ground-truth values approach zero (e.g., FIB4 typical value ${\sim}1$ within $[0.12, 37.95]$; AIP typical value ${\sim}0.4$ within $[-0.64, 1.97]$).

\paragraph{Risk Class Error (RCE)}
Each template defines clinically meaningful risk categories based on published threshold values (App.~\ref{app:med-profiles}). For example, the FIB-4 template classifies patients into Low ($< 1.30$), Indeterminate ($1.30$--$2.67$), or High ($> 2.67$) fibrosis risk. We map both the ground truth and sanitized target values to their respective risk classes $c(y)$ and $c(\hat{y})$, then compute:
\begin{equation}
\text{Risk Class Error} = \frac{1}{N} \sum_{i=1}^{N} \frac{|c(\hat{y}_i) - c(y_i)|}{C - 1} \times 100\%
\end{equation}
where $C$ is the number of risk classes for the profile ($C = 3$ for ANEMIA, AIP, FIB4, HOMA; $C = 2$ for the remaining profiles). This metric directly measures \emph{clinical harm}: a Risk Class Error of 0\% means that sanitization never changes the clinical decision, regardless of numeric distortion. The two metrics are complementary: a method may achieve low wMAPE but non-zero RCE if perturbations cross decision boundaries.

% \noindent The two metrics capture complementary aspects of utility. A method can achieve low wMAPE (small numeric perturbation) while still incurring non-zero Risk Class Error if the perturbation crosses a decision boundary. Conversely, a method with higher wMAPE may preserve risk classes perfectly if the noise stays within the same classification band. \sysname{}'s budget allocation is formulated to minimize absolute error at the target, but the resulting improvement in numeric accuracy typically translates to improved risk class preservation as well.

\subsection{RQ1: Target Utility Under Adversarial Queries}
\label{sec:experiments-rq1}

\providecommand{\pmstd}[1]{{\scriptsize $\pm$#1}}

\begin{table*}[t]
\centering
\setlength{\tabcolsep}{2.8pt}
\small
\caption{Aggregate wMAPE (\%) at mid-to-high privacy levels across three adversary budget levels (Exponential mechanism). $\bar{B}$: average total budget across templates. $\mathcal{M}$-All is invariant to the budget level; $\mathcal{M}$-Opt improves with each additional adversarial turn. Aggregate values are the mean across 8 templates; bootstrap SEs are propagated from per-template SEs as $\sqrt{\sum_t \mathrm{SE}_t^2}/k$. At $\varepsilon \geq 1.0$, all methods achieve $\leq 1\%$ wMAPE and differences are negligible (see App.~\ref{app:rq1-tables}).}
\label{tab:double_asymmetry}
\begin{tabular}{c|cccc|cccc|cccc}
\toprule
& \multicolumn{4}{c|}{$B = (k{+}1)\varepsilon$} & \multicolumn{4}{c|}{$B = (2k{+}1)\varepsilon$} & \multicolumn{4}{c}{$B = (3k{+}1)\varepsilon$} \\
$\varepsilon$ & $\bar{B}$ & All & Roots & Opt & $\bar{B}$ & All & Roots & Opt & $\bar{B}$ & All & Roots & Opt \\
\midrule
0.5 & 1.81 & 5.6\pmstd{0.30} & 4.0\pmstd{0.25} & \textbf{2.7\pmstd{0.13}} & 3.12 & 5.7\pmstd{0.30} & 2.4\pmstd{0.13} & \textbf{1.7\pmstd{0.09}} & 4.44 & 5.9\pmstd{0.33} & 1.8\pmstd{0.10} & \textbf{1.2\pmstd{0.05}} \\
0.2 & 0.73 & 12.6\pmstd{0.64} & 10.1\pmstd{0.47} & \textbf{6.6\pmstd{0.29}} & 1.25 & 12.3\pmstd{0.62} & 5.9\pmstd{0.37} & \textbf{4.5\pmstd{0.22}} & 1.78 & 12.2\pmstd{0.55} & 4.6\pmstd{0.29} & \textbf{3.3\pmstd{0.20}} \\
0.1 & 0.36 & 20.5\pmstd{0.88} & 14.9\pmstd{0.67} & \textbf{13.6\pmstd{0.65}} & 0.62 & 20.3\pmstd{0.96} & 11.4\pmstd{0.57} & \textbf{7.8\pmstd{0.32}} & 0.89 & 21.5\pmstd{0.96} & 8.7\pmstd{0.47} & \textbf{5.6\pmstd{0.26}} \\
0.05 & 0.18 & 35.0\pmstd{1.45} & 26.0\pmstd{1.11} & \textbf{21.6\pmstd{0.94}} & 0.31 & 33.8\pmstd{1.69} & 18.0\pmstd{0.80} & \textbf{13.3\pmstd{0.58}} & 0.44 & 33.7\pmstd{1.75} & 14.3\pmstd{0.68} & \textbf{11.2\pmstd{0.57}} \\
0.02 & 0.07 & 70.3\pmstd{4.35} & 55.5\pmstd{6.81} & \textbf{46.0\pmstd{2.80}} & 0.12 & 75.7\pmstd{10.89} & 34.1\pmstd{1.55} & \textbf{29.2\pmstd{1.23}} & 0.18 & 72.3\pmstd{5.10} & 27.9\pmstd{1.29} & \textbf{21.9\pmstd{0.96}} \\
0.01 & 0.04 & 122.5\pmstd{10.10} & 101.9\pmstd{9.58} & \textbf{86.6\pmstd{6.21}} & 0.06 & 112.6\pmstd{8.41} & 63.0\pmstd{5.39} & \textbf{48.3\pmstd{2.62}} & 0.09 & 129.5\pmstd{9.85} & 47.3\pmstd{2.59} & \textbf{37.3\pmstd{1.58}} \\
\bottomrule
\end{tabular}
\end{table*}

We evaluate target-node utility under the adversary budgets defined in Sec.~\ref{sec:exp-setup}. $\mathcal{M}$-All noises each root release independently at $\varepsilon$, computing the target deterministically from the latest noised values at the target turn. $\mathcal{M}$-Roots and $\mathcal{M}$-Opt distribute the full budget $B$ across roots once at initialization, then post-process all subsequent releases.

\paragraph{Data.} We use the same 8 clinical templates and 200-patient NHANES 2017--2018 test split as in the experimental setup (Sec.~\ref{sec:exp-setup}), across 10 privacy budgets $\varepsilon \in [0.005, 5.0]$.

% ────────────────────────────────────────────────────────────────────────
% \subsubsection{What governs the magnitude of \sysname{}'s advantage?}
% Under the root-only configurations, the target perturbation depends on how the formula $g$ amplifies or compresses each root's noise (captured by $|h_i| = |\partial g / \partial x_i|$). Templates whose target formulas \emph{compress} root noise (e.g., divisions, products with large constants) yield the largest gaps between $\mathcal{M}$-All and \sysname{} --- root-domain noise gets attenuated en route to the target, while $\mathcal{M}$-All's per-release noise compounds across roots in a Lipschitz-amplified way (Sec.~\ref{sec:nonlinear-amplification}). Templates whose formulas \emph{amplify} root noise into a narrow target domain show smaller absolute errors for both configurations, with \sysname{} retaining a multiplicative advantage on a smaller base. Two structural factors further modulate the gap: (i) \textit{budget concentration} --- each root receives $t/k$ times $\mathcal{M}$-All's per-release budget, so higher $t/k$ favors root-only noising; (ii) \textit{zero-sensitivity roots} --- some roots (e.g., RBC for MCHC, TC for AIP) waste budget under $\mathcal{M}$-Roots; $\mathcal{M}$-Opt clamps them at $\varepsilon_{\min}$ and redistributes the freed budget to active roots.

% ────────────────────────────────────────────────────────────────────────

\subsubsection{Aggregate results.}
Table~\ref{tab:double_asymmetry} reports aggregate wMAPE for the Exponential mechanism. The ordering $\mathcal{M}$-All $>$ $\mathcal{M}$-Roots $>$ $\mathcal{M}$-Opt holds at every $\varepsilon$ and every budget level, confirming that root-only noising and dependency-aware allocation provide additive gains; BLap and Stair follow the same trend (App.~\ref{app:rq1-tables}). At $\varepsilon = 0.1$ and $B = (2k{+}1)\varepsilon$, Exp-Opt achieves 7.8\% vs.\ 20.3\% for Exp-All --- a $2.6\times$ improvement, widening to $3.8\times$ at $(3k{+}1)\varepsilon$ ($5.6\%$ vs.\ $21.5\%$). The gap also widens at lower $\varepsilon$: at $\varepsilon = 0.01$, $B = (2k{+}1)\varepsilon$, the gap is $48.3\%$ vs.\ $112.6\%$.

\subsubsection{Clinical decision preservation (RCE)}
The wMAPE wins translate into clinical-decision wins. At the focal ($\varepsilon = 0.1$, Exp, $B = (2k{+}1)\varepsilon$), $\mathcal{M}$-Opt achieves Risk Class Error of $7.0\% \pm 0.6\%$ vs.\ $\mathcal{M}$-All's $15.7\% \pm 0.8\%$, a $2.2\times$ reduction; the advantage widens to $3.3\times$ at $B = (3k{+}1)\varepsilon$ ($5.3\%$ vs.\ $17.5\%$), tracking the wMAPE pattern. Table~\ref{tab:rce_per_template_rq1} reports the per-template breakdown. The Bounded Laplace and Staircase mechanisms exhibit the same ordering at every cell (App.~\ref{app:rq1-tables}, Tables~\ref{tab:rq1_kplus1_rce}--\ref{tab:rq1_3kplus1_rce}).

\begin{table}[h]
\centering
\setlength{\tabcolsep}{2pt}
% \footnotesize
\caption{Per-template Risk Class Error (\%) at $\varepsilon = 0.1$ for
the discrete Exponential mechanism across three adversarial budgets,
under the controlled harness (RQ1). Bootstrap SE over patients
($n = 200$). Aggregate cross-mechanism results are in
App.~\ref{app:rq1-tables}, Tables~\ref{tab:rq1_kplus1_rce}--\ref{tab:rq1_3kplus1_rce}.
Lower is better.}
\label{tab:rce_per_template_rq1}
\begin{tabular}{l|ccc|ccc}
\toprule
& \multicolumn{3}{c|}{\textbf{$\mathcal{M}$-All}} & \multicolumn{3}{c}{\textbf{$\mathcal{M}$-Opt}} \\
Template & $k{+}1$ & $2k{+}1$ & $3k{+}1$ & $k{+}1$ & $2k{+}1$ & $3k{+}1$ \\
\midrule
HOMA     & 36.0\pmstd{3.4} & 31.0\pmstd{3.3} & 42.0\pmstd{3.5} & 21.0\pmstd{2.9} & 19.0\pmstd{2.8} & \textbf{15.5\pmstd{2.6}} \\
AIP      & 18.2\pmstd{2.7} & 20.8\pmstd{2.9} & 26.2\pmstd{3.1} & 12.0\pmstd{2.3} & 5.8\pmstd{1.7}  & \textbf{4.5\pmstd{1.5}}  \\
FIB4     & 7.2\pmstd{1.8}  & 10.0\pmstd{2.1} & 10.2\pmstd{2.1} & 5.5\pmstd{1.6}  & 3.0\pmstd{1.2}  & \textbf{1.8\pmstd{0.9}}  \\
NLR      & 21.5\pmstd{2.9} & 24.0\pmstd{3.0} & 19.5\pmstd{2.8} & 19.0\pmstd{2.8} & 12.0\pmstd{2.3} & \textbf{9.5\pmstd{2.1}}  \\
TYG      & 22.0\pmstd{2.9} & 22.0\pmstd{2.9} & 25.5\pmstd{3.1} & 18.0\pmstd{2.7} & 11.0\pmstd{2.2} & \textbf{9.0\pmstd{2.0}}  \\
ANEMIA   & 8.0\pmstd{1.9}  & 6.0\pmstd{1.7}  & 8.0\pmstd{1.9}  & 2.2\pmstd{1.0}  & 1.0\pmstd{0.7}  & \textbf{1.0\pmstd{0.7}}  \\
CONICITY & 7.0\pmstd{1.8}  & 10.0\pmstd{2.1} & 7.5\pmstd{1.9}  & 4.5\pmstd{1.5}  & 3.0\pmstd{1.2}  & \textbf{1.0\pmstd{0.7}}  \\
VASCULAR & 3.0\pmstd{1.2}  & 1.5\pmstd{0.9}  & 1.0\pmstd{0.7}  & 0.0\pmstd{0.0}  & 1.0\pmstd{0.7}  & \textbf{0.5\pmstd{0.5}}  \\
\bottomrule
\end{tabular}
\end{table}

\subsubsection{Budget scaling: the double asymmetry.}
\label{sec:double-asymmetry}
 
$\mathcal{M}$-All's target utility depends only on $\varepsilon$, not $B$ --- the target is computed from a single fresh draw per root at the target turn, so the adversary's additional turns do not improve target accuracy. $\mathcal{M}$-Roots and $\mathcal{M}$-Opt benefit from every unit of additional budget. Table~\ref{tab:double_asymmetry} illustrates this \emph{double asymmetry}: $\mathcal{M}$-All columns are nearly flat across budget levels, while $\mathcal{M}$-Opt columns decrease as budget concentrates on roots. At $\varepsilon = 0.1$, Exp-Opt improves from 13.6\% ($k{+}1$) to 7.8\% ($2k{+}1$) to 5.6\% ($3k{+}1$), while Exp-All stays at ${\sim}20$--$21\%$. \sysname{} wins on every (template, mechanism) cell at every budget level (Sec.~\ref{sec:experiments-rq3}); the multiplicative advantage \emph{grows} with $B$, driven by templates with compressive target formulas (FIB4, HOMA, AIP).

\subsubsection{Text-DP reference (CAPE-F)} CAPE-F achieves $263$--$403\%$ wMAPE across all $\varepsilon$ ($65$--$125\times$ worse than BLap-Opt) and is essentially flat across the privacy axis, confirming that text-DP is unsuitable for structured numeric data regardless of $\varepsilon$.

% \subsubsection{CAPE-F: text-DP reference.}
% It achieves $263$--$403$\% wMAPE across all $\varepsilon$, $65$--$125\times$ worse than BLap-Opt. Its error is essentially flat across $\varepsilon$ (the $1.5\times$ ratio from $\varepsilon = 5.0$ to $0.005$ compares to BLap-Opt's $225\times$), confirming the error is dominated by linguistic-vs-numeric mismatch rather than privacy noise. The embedding model cannot capture relationships between numerical tokens and context: ``$0$'' might be replaced by ``$00$'', immediately increasing error by an order of magnitude.

\subsubsection{Privacy cost of extra turns.} The $t-1$ adversarial turns preceding the target have no utility cost for either method, but they impose a severe privacy cost on $\mathcal{M}$-All --- each query produces an independent observation a reconstruction adversary can exploit, while \sysname{}'s post-processed responses reveal no additional information. We quantify this in Sec.~\ref{sec:experiments-rq2}.

% \paragraph{Why not concentrate $\mathcal{M}$-All's budget on the target?}
% A natural alternative is to give $\mathcal{M}$-All the full $B = t\varepsilon$ on the target release, ignoring the $t-1$ adversarial turns. Our threat model rules this out: the adversary controls the interaction, and the user does not know in advance which release will be the target. A user that allocates $B$ to a single release exposes itself to an adversary that simply asks for that release first and exits, leaving $\mathcal{M}$-All with no budget left for legitimate queries. Equivalently, in any strategy that allocates a budget across $t$ unknown future requests, $\mathcal{M}$-All must spend per-release; \sysname{} avoids this dilemma by treating all releases as post-processing of the initial root noising, so allocation is decoupled from request order. Across the $t$ release sequence we evaluate, $\mathcal{M}$-All's per-release budget $\varepsilon$ is the unique strategy that yields the same $\varepsilon$-mDP guarantee on every individual release.

\subsection{RQ2: Reconstruction Attacks}
\label{sec:experiments-rq2}

RQ1 showed $\mathcal{M}$-All's adversarial turns yield no utility benefit. We now ask: do they impose a \emph{privacy cost}? The accumulation analysis (Sec.~\ref{sec:redundancy-noising}) predicts that each independent release gives the adversary a fresh observation, enabling MAP reconstruction that improves with query count. We test this by varying the number of queries under matched per-root noise.

% RQ1 showed that $\mathcal{M}$-All's adversarial turns have no utility benefit --- the target is a single independent release at $\varepsilon$. We now ask: do these turns have a \emph{privacy cost}? The accumulation analysis (Sec.~\ref{sec:redundancy-noising}) predicts that each independent release gives the adversary a fresh observation, enabling MAP reconstruction that improves with query count. We test this prediction by comparing reconstruction under matched per-root noise, varying the number of queries.

\subsubsection{Setup.}
We fix per-root noise $\varepsilon_r \in \{0.05, 0.1, 0.5, 1.0\}$ and vary the adversarial query count $q \in \{1, 4, 8, 16\}$. $\mathcal{M}$-All produces a fresh independent draw per query; $\mathcal{M}$-Roots noises each root once at $\varepsilon_r$ and caches; $\mathcal{M}$-Opt noises each root once under the sensitivity-weighted allocation of total budget $B = k\varepsilon_r$. Method-agnostic RNG seeding makes $\mathcal{M}$-All's first draw and the $\mathcal{M}$-Roots cached value bit-identical at any $(\text{template}, \text{root}, \text{sample})$ triple, so any $q \geq 1$ difference is purely independent-draws vs.\ caching.

% \paragraph{Adversary strategies.}
% Each strategy uses $q$ \emph{root queries} (each requesting one or more sanitized root values) followed by one final \emph{target query} that forces a release of every root needed for the target. For the sake of narration, the last turn requests all roots for computing the clinical diagnosis instead of requesting them individually over $k$ extra turns. The total $q+1$ turns therefore yield $q+1$ observations on at least one root. 

\paragraph{Adversary strategies.}
Each strategy uses $q$ \emph{root queries} (each requesting one or more sanitized root values) followed by one final \emph{target query} that bundles all $k$ roots needed for the diagnosis into a single turn. We bundle the final-turn roots rather than spreading them across $k$ separate turns purely for narrative simplicity --- every strategy has the same $q+1$-turn structure, and the privacy and utility accounting is unchanged either way. The total $q+1$ turns yield $q+1$ observations on at least one root.

\paragraph{Strategy A (single-root, worst-case):} the $q$ root queries  all target $r^* = \arg\max_r \varepsilon_r^{\mathrm{Opt}}$ --- the root that receives the largest budget under $\mathcal{M}$-Opt's sensitivity-weighted allocation, and therefore the least-noised one; combined with the target turn, the adversary observes $q{+}1$ values on $r^*$: under $\mathcal{M}$-All, $q{+}1$ fresh independent draws; under cached methods, $q{+}1$ identical replays. We report wMAPE on $r^*$. 

\paragraph{Strategy B (all-roots, average-case):} the $q$ queries are distributed across all $k$ roots; combined with the target turn, each root accumulates $q/k + 1$ observations. We report mean wMAPE across roots.

\paragraph{MAP adversary and priors.} The adversary performs a per-root 1-D argmax with full knowledge of the sanitization pipeline: $\hat{x}_r = \arg\max_x \prod_i p_{\mathcal{M}}(\tilde{y}_i \mid x, \varepsilon_r) \cdot \pi(x)$. The factorization across roots is exact: under both $\mathcal{M}$-All and the cached methods, per-root observations are independent given the true value, and the user agent never releases a value that couples roots through the target. Unobserved roots return the prior mean (Strategy A only). For cached methods, the per-root likelihood collapses to a single factor regardless of $q$. We evaluate two priors --- \emph{uniform} ($\pi(x_r) \propto 1$, MAP reduces to MLE) and \emph{informed} (independent Gaussian per root, $\pi(x_r) = \mathcal{N}(\mu_r, \sigma_r^2)$, calibrated from NHANES holdout statistics).

\subsubsection{Results.}

Fig.~\ref{fig:recon_main} and Table~\ref{tab:recon_main} present the main results (Exponential, informed prior, Strategy B). Full results across all privacy levels, mechanisms, and per-template breakdowns are in App.~\ref{app:rq2-tables}.

\begin{figure*}[t]
    \centering
    \begin{subfigure}[t]{0.48\textwidth}
        \centering
        \includegraphics[width=\columnwidth]{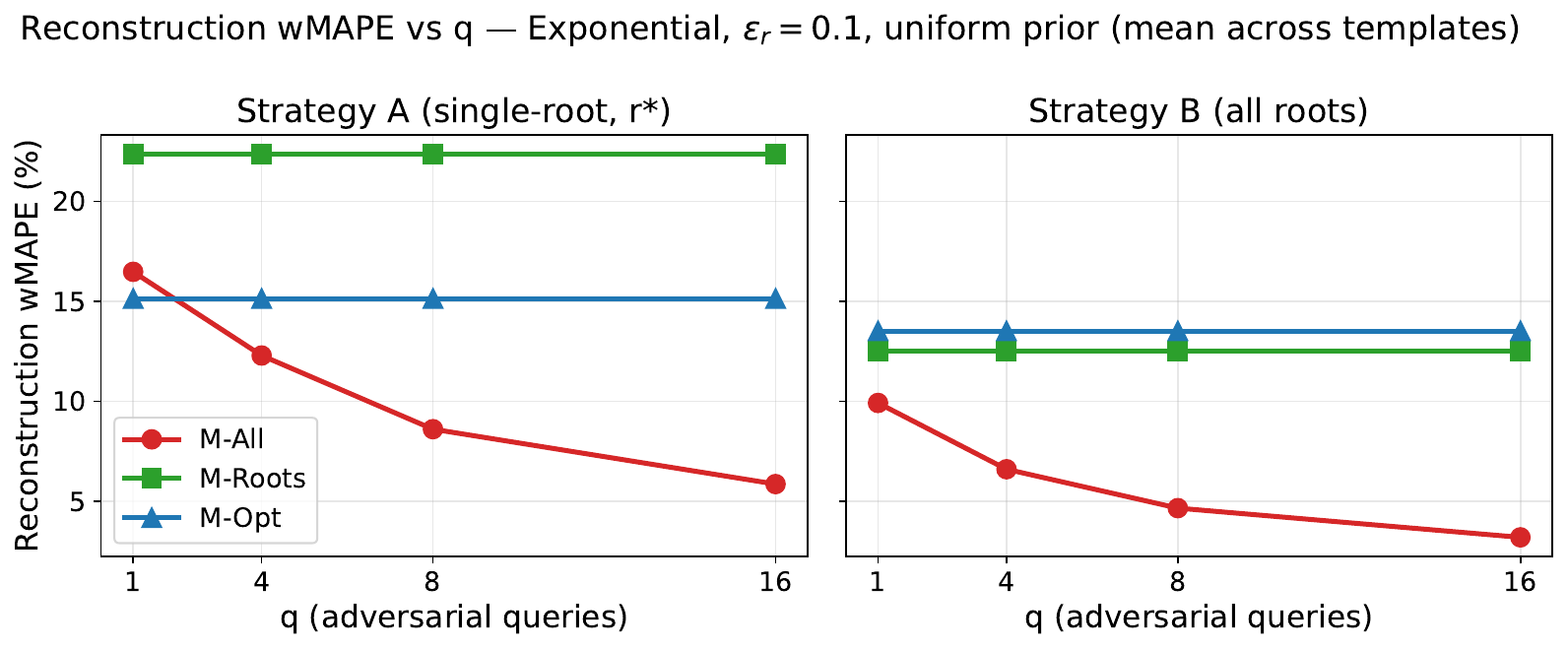}
        \caption{Uniform prior. At $q = 1$, $\mathcal{M}$-All has $q+1 = 2$ observations on the targeted root, while $\mathcal{M}$-Roots has the cached value replayed twice. The gap at $q = 1$ reflects the second independent draw under $\mathcal{M}$-All; the divergence at $q > 1$ is the noise-averaging effect on accumulating draws.}
        \label{fig:recon_uniform}
    \end{subfigure}
    \hfill
    \begin{subfigure}[t]{0.48\textwidth}
        \centering
        \includegraphics[width=\columnwidth]{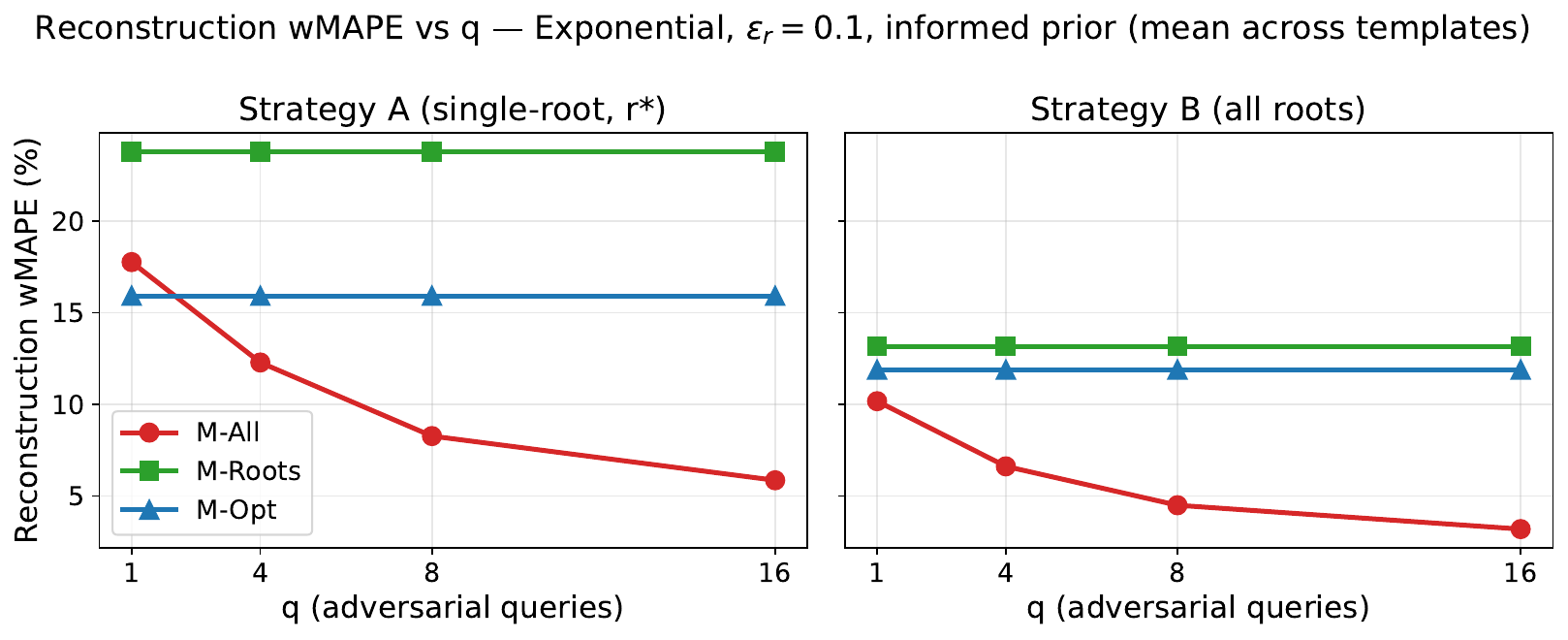}
        \caption{Informed prior. With Gaussian priors over unqueried roots, the qualitative behavior is unchanged from the uniform case: $\mathcal{M}$-All decreases monotonically with $q$; cached methods are flat.}
        \label{fig:recon_informed}
    \end{subfigure}
    \caption{Reconstruction wMAPE (\%) vs.\ adversarial query count $q$ at $\varepsilon_r = 0.1$ (Exponential, mean across 8 templates). Each row shows two strategies (A: single-root, B: all-roots). $\mathcal{M}$-All (red) degrades with $q$ in all settings; $\mathcal{M}$-Roots (green) and $\mathcal{M}$-Opt (blue) are flat.}
    \label{fig:recon_main}
\end{figure*}

\begin{table}[t]
\centering
\setlength{\tabcolsep}{2.5pt}
\small
\caption{Aggregate reconstruction wMAPE (\%) across privacy levels (Exponential, Strategy B, informed prior, mean across 8 templates). $\mathcal{M}$-All decreases with $q$ at every $\varepsilon_r$; $\mathcal{M}$-Roots and $\mathcal{M}$-Opt are invariant. Lower values mean the adversary reconstructs more accurately (worse for privacy).}
\label{tab:recon_main}
\begin{tabular}{c|ccc|ccc|ccc|ccc}
\toprule
& \multicolumn{3}{c|}{$q=1$} & \multicolumn{3}{c|}{$q=4$} & \multicolumn{3}{c|}{$q=8$} & \multicolumn{3}{c}{$q=16$} \\
$\varepsilon_r$ & All & Rts & Opt & All & Rts & Opt & All & Rts & Opt & All & Rts & Opt \\
\midrule
0.05 & 16.4 & 22.8 & 21.0 & 10.9 & 22.8 & 21.0 & 8.6  & 22.8 & 21.0 & 6.4 & 22.8 & 21.0 \\
0.1  & 10.2 & 13.2 & 11.9 & 6.6  & 13.2 & 11.9 & 4.5  & 13.2 & 11.9 & 3.2 & 13.2 & 11.9 \\
0.5  & 2.8  & 2.9  & 3.4  & 1.3  & 2.9  & 3.4  & 0.9  & 2.9  & 3.4  & 0.6 & 2.9  & 3.4 \\
1.0  & 1.4  & 1.5  & 2.3  & 0.7  & 1.5  & 2.3  & 0.5  & 1.5  & 2.3  & 0.3 & 1.5  & 2.3 \\
\bottomrule
\end{tabular}
\end{table}

\paragraph{Source of leakage.}
The result reveals a single effect compromising $\mathcal{M}$-All's privacy: \emph{repeated query accumulation}. Each adversarial query produces a fresh independent observation on the queried root, and the target turn contributes one additional fresh draw under $\mathcal{M}$-All (replayed cached value under \sysname{}). As $q$ grows, the adversary accumulates $q+1$ independent observations on each affected root and combines them via per-root MAP, tightening the posterior around the true value. The reduction in wMAPE follows the standard $1/\sqrt{q+1}$ scaling for averaged estimators of a fixed parameter under additive noise. \sysname{}'s cached methods provide no additional information after the first observation: every subsequent query returns the same value, leaving the per-root MAP unchanged.

\paragraph{Strategy A}
Fig.~\ref{fig:recon_uniform} (left panel) shows the targeted-root reconstruction. At $q = 1$, $\mathcal{M}$-All has 2 observations on $r^*$ (one query plus the target turn) while $\mathcal{M}$-Roots has the cached value replayed twice; the gap (16.5\% vs.\ 22.4\% aggregate) is purely the gain from a second independent draw under $\mathcal{M}$-All. As $q$ grows, $\mathcal{M}$-All's wMAPE drops to 12.3\% ($q = 4$), 8.6\% ($q = 8$), and 5.9\% ($q = 16$), tracking the expected $1/\sqrt{q+1}$ noise-averaging rate. $\mathcal{M}$-Roots and $\mathcal{M}$-Opt remain flat: at every $q$, $\mathcal{M}$-Roots reconstructs $r^*$ with 22.4\% wMAPE (the per-root error of one $\varepsilon_r$-noised release), and $\mathcal{M}$-Opt achieves 15.1\% because it concentrates budget on $r^*$ ($\varepsilon_{r^*} > \varepsilon_r$). The $2.8\times$ reduction ($16.5\% \to 5.9\%$) under $\mathcal{M}$-All is entirely attributable to independent noise averaging.

The informed prior (Fig.~\ref{fig:recon_informed}, left panel) gives nearly identical results because the target release does not enter the MAP --- the only effect of the informed prior on the targeted root is a soft regularization toward the population mean. At $\varepsilon_r = 0.1$: $\mathcal{M}$-All starts at 17.8\% (vs.\ 16.5\% uniform), drops to 5.9\% at $q = 16$. $\mathcal{M}$-Roots remains at 23.8\%. The qualitative result is unchanged.

\paragraph{Strategy B}
Table~\ref{tab:recon_main} shows that the pattern is consistent across all privacy levels. At $\varepsilon_r = 0.05$ (strong privacy), $\mathcal{M}$-All drops from $16.4\%$ to $6.4\%$ as $q$ goes from $1$ to $16$, while $\mathcal{M}$-Roots and $\mathcal{M}$-Opt remain at $22.8\%$ and $21.0\%$ respectively. At $\varepsilon_r = 0.1$, $\mathcal{M}$-All drops from $10.2\%$ to $3.2\%$ ($3.2\times$ reduction), against flat cached methods at $13.2\%$ ($\mathcal{M}$-Roots) and $11.9\%$ ($\mathcal{M}$-Opt). At $\varepsilon_r = 0.5$ (moderate privacy), $\mathcal{M}$-All drops from $2.8\%$ to $0.6\%$, while $\mathcal{M}$-Roots stays at $2.9\%$. At $q = 16$, $\mathcal{M}$-All reconstructs $4$--$5\times$ more accurately than $\mathcal{M}$-Roots across every privacy level. The per-template breakdown (Table~\ref{tab:rq2_per_template}, App.~\ref{app:rq2-tables}) confirms $\mathcal{M}$-All degrades on every template, with the largest absolute degradation on templates with wide root domains (HOMA: $19.1\% \to 5.1\%$, NLR: $20.9\% \to 7.6\%$ at $\varepsilon_r = 0.1$, informed prior).

% \begin{figure*}[t]
%     \centering
%     \includegraphics[width=\linewidth]{figures/grid_exp_B_eps0.1_informed.pdf}
%     \caption{Per-template reconstruction wMAPE (\%) vs.\ $q$ at $\varepsilon_r = 0.1$ (Exponential, Strategy B, informed prior). $\mathcal{M}$-All degrades on every template; $\mathcal{M}$-Roots and $\mathcal{M}$-Opt are flat.}
%     \label{fig:recon_per_template}
% \end{figure*}

\paragraph{$\mathcal{M}$-Opt vs.\ $\mathcal{M}$-Roots.}
The two cached methods show complementary behavior across strategies. In Strategy A, $\mathcal{M}$-Opt achieves lower wMAPE than $\mathcal{M}$-Roots on the targeted root ($15.1\%$ vs.\ $22.4\%$ aggregate, uniform prior) because it concentrates budget on $r^*$. In Strategy B, $\mathcal{M}$-Roots sometimes achieves lower \emph{average} wMAPE than $\mathcal{M}$-Opt --- for example, ANEMIA: M-Roots $1.9\%$ vs.\ M-Opt $3.6\%$ --- because $\mathcal{M}$-Opt starves low-sensitivity roots (e.g., RBC at $\varepsilon_{\min}$), making them harder to reconstruct, which increases the average. This illustrates a tradeoff: $\mathcal{M}$-Opt protects the most sensitive root but at the cost of higher average reconstruction error across all roots.

\paragraph{Key result}
$\mathcal{M}$-All's reconstruction wMAPE decreases monotonically with $q$ at every privacy level, every template, and under both priors. $\mathcal{M}$-Roots and $\mathcal{M}$-Opt are invariant to $q$. Under a worst-case adversary, $\mathcal{M}$-All's privacy degrades with each additional query; \sysname{} maintains a constant guarantee determined solely by the initial root noising. The Bounded Laplace and Staircase mechanisms exhibit the same pattern (App.~\ref{app:rq2-tables}).

Our reconstruction experiments use direct root queries, which is the simplest adversary strategy. Under $\mathcal{M}$-All, nonlinear derived-function queries could further amplify the adversary's advantage (Sec.~\ref{sec:nonlinear-amplification}); exploring optimal adversary function choice is left to future work. Under \sysname{}, function choice is irrelevant by Theorem~\ref{thm:privacy-guarantee}. Our results are therefore conservative estimates of $\mathcal{M}$-All's vulnerability.

\subsection{RQ3: Structural Analysis}
\label{sec:experiments-rq3}
\input{figures/rq4_aggregate_wmape}

RQ1 showed that $\mathcal{M}$-Opt consistently outperforms $\mathcal{M}$-Roots in aggregate, demonstrating that \emph{how} the budget is distributed matters beyond eliminating redundant noising. We now examine the structural conditions that determine the magnitude of \sysname{}'s advantage and the allocation properties that drive it. \sysname{} wins on every (template, mechanism, $\varepsilon$) cell at $B = (2k{+}1)\varepsilon$ and above on both wMAPE and RCE; the size of the gap, however, varies substantially across templates, tracking two structural factors: (1) whether the target formula compresses or amplifies root noise, and (2) the geometry of the sensitivity-weighted budget allocation across roots. Per-template win counts across all mechanisms and privacy levels are reported in App.~\ref{app:rq3-winrate} (Table~\ref{tab:rq1_win_rate}).

% ────────────────────────────────────────────────────────────────────────
\subsubsection{Per-template analysis.}
\label{sec:per-template}

All per-template numbers in this section are reported at $B = (2k{+}1)\varepsilon$ for the Exponential mechanism unless otherwise noted. Full per-template wMAPE curves across $\varepsilon$ and all three mechanisms are in App.~\ref{app:rq1-tables-figs} (Fig.~\ref{fig:per_template_app}).

The per-template breakdown sorts templates into two regimes (Table~\ref{tab:per_template_summary}). \emph{Compressing} target formulas (logarithms, divisions by large constants) attenuate root noise with wide target domains, allowing \sysname{} to recover utility (AIP, HOMA, FIB4, NLR; absolute gaps $14$--$29$\,pp). \emph{Amplifying} formulas project root noise into narrow target domains, producing small absolute errors for both methods but preserving a $2$--$3\times$ multiplicative advantage for \sysname{} (ANEMIA, CONICITY, TYG, VASCULAR; gaps $\leq 4$\,pp). Across all 8 templates, no template falls below $2.1\times$ at $B = (2k{+}1)\varepsilon$; the pattern holds across all three mechanisms (App.~\ref{app:rq1-tables}) and grows more pronounced at $(3k{+}1)\varepsilon$, where ratios reach $3.0$--$5.4\times$.

\begin{table}[h]
\centering
\setlength{\tabcolsep}{4pt}
\small
\caption{Per-template wMAPE (\%) at $\varepsilon = 0.1$, $B = (2k{+}1)\varepsilon$ (Exponential). \emph{Compressing} formulas show large absolute gaps; \emph{amplifying} formulas show small absolute gaps but comparable multiplicative ratios. Full per-template tables in App.~\ref{app:rq1-tables}.}
\label{tab:per_template_summary}
\begin{tabular}{l l c c c}
\toprule
\textbf{Template} & \textbf{Formula} & \textbf{All} & \textbf{Opt} & \textbf{Ratio} \\
\midrule
\multicolumn{5}{l}{\emph{Compressing — large absolute gaps}} \\
AIP   & $\log_{10}(\text{TG}/\text{HDL})$           & 41.9 & 12.7 & $3.3\times$ \\
HOMA  & $\text{Glu}\cdot\text{Ins}/405$              & 31.1 &  9.9 & $3.1\times$ \\
NLR   & $\text{NEU}/\text{LYM}$                      & 49.2 & 23.1 & $2.1\times$ \\
FIB4  & $\text{age}\cdot\text{AST}/(\text{PLT}\sqrt{\text{ALT}})$ & 23.9 & 10.0 & $2.4\times$ \\
\midrule
\multicolumn{5}{l}{\emph{Amplifying — small absolute gaps}} \\
VASCULAR & $100(\text{SBP}{-}\text{DBP})/\text{SBP}$ &  6.3 &  2.6 & $2.4\times$ \\
TYG      & $\ln(\text{TG}\cdot\text{Glu}/2)$         &  4.6 &  2.0 & $2.3\times$ \\
CONICITY & $\text{waist}/(0.109\sqrt{\text{wt}/\text{ht}})$ &  2.8 &  1.0 & $2.8\times$ \\
ANEMIA   & $100\cdot\text{Hb}/\text{Hct}$            &  2.4 &  0.7 & $3.3\times$ \\
\bottomrule
\end{tabular}
\end{table}

\subsubsection{Power-law allocation and mechanism independence.}

The closed-form approximation (Sec.~\ref{sec:opt-budget}) predicts $\varepsilon_i \propto (|h_i| \cdot D_i)^{1/2}$. Empirically, normalized per-root allocations across all 8 templates collapse onto a single log-log line with fitted slope 0.50 (Exponential, $\varepsilon = 0.1$, $B = (2k{+}1)\varepsilon$; full plot in App.~\ref{app:per-root-budget}, Fig.~\ref{fig:allocation_powerlaw_01}), confirming the prediction. Roots with higher sensitivity--domain products receive proportionally more budget.

% The closed-form approximation (Sec.~\ref{sec:opt-budget}) predicts $\varepsilon_i \propto (|h_i| \cdot D_i)^{1/2}$. Fig.~\ref{fig:allocation_powerlaw} verifies this empirically: on a log-log scale, the normalized per-root allocations across all 8 templates collapse onto a single line with fitted slope 0.50 (Exponential, $B = (2k{+}1)\varepsilon$, $\varepsilon = 0.1$). Roots with higher sensitivity--domain products, which contribute more to the target value, receive proportionally more budget.

The allocation is mechanism-independent: at $\varepsilon = 0.1$, all three mechanisms produce fitted slopes of exactly 0.50; at $\varepsilon = 0.5$, minor deviations appear (Exp 0.49, BLap 0.52, Stair 0.51). Full three-mechanism plots at all privacy levels are in App.~\ref{app:per-root-budget}. This confirms that the budget distribution is governed by the DAG's sensitivity geometry, not the noise distribution. A deployment can switch mechanisms without recomputing the allocation.

% \begin{figure}[t]
%     \centering
%     \includegraphics[width=\columnwidth]{figures/rq2_allocation_powerlaw_2kplus1_Exp_eps0p1.pdf}
%     \caption{Log-log scatter of normalized per-root budget $\varepsilon_i$ vs.\ sensitivity--domain product $|h_i| \cdot D_i$ (Exponential, $\varepsilon = 0.1$, $B = (2k{+}1)\varepsilon$). All 8 templates, active roots only. Fitted slope = 0.50. Roots are labeled by medical attribute.}
%     \label{fig:allocation_powerlaw}
% \end{figure}

\subsubsection{Zero-sensitivity roots.}

Some roots have zero sensitivity to the target: RBC in ANEMIA ($\partial \text{MCHC}/\partial \text{RBC} = 0$) and TC in AIP ($\partial \text{AIP}/\partial \text{TC} = 0$). Under $\mathcal{M}$-Roots, these consume $1/k$ of the budget without reducing target error. Under $\mathcal{M}$-Opt, they are clamped at $\varepsilon_{\min}$ and their share is redistributed, effectively giving each active root $50\%$ more budget on these 3-root templates.

This creates complementary effects in the two RQs. In RQ1, the budget reclamation amplifies \sysname{}'s advantage on these 3-root templates: each active root receives effectively $50\%$ more budget than under uniform allocation, contributing to the strong multiplicative wins on ANEMIA ($3.3\times$) and AIP ($3.3\times$) at $\varepsilon = 0.1$, $B = (2k{+}1)\varepsilon$. In RQ2, the clamped roots are very hard to reconstruct (high noise), increasing $\mathcal{M}$-Opt's average reconstruction wMAPE under Strategy B (e.g., ANEMIA: M-Opt $3.6\%$ vs.\ M-Roots $1.9\%$) --- while simultaneously making the high-budget roots harder for an adversary to improve upon under Strategy A.

%%%%%%%%%%%%%%%%%%%%%%%%%%%%%%%%%%%%%%%%%%%%%%%%%%%%%
%%%%%%%%%%%%%%%%%%%%%%%%%%%%%%%%%%%%%%%%%%%%%%%%%%%%%
%%%%%%%%%%%%%%%%%%%%%%%%%%%%%%%%%%%%%%%%%%%%%%%%%%%%%
\subsection{RQ4: End-to-End Deployment in an LLM Tool-Call Pipeline}
\label{sec:experiments-rq4}

RQ1--3 establish the privacy-utility tradeoff at the mechanism level, with the user agent and adversary simulated directly through mechanism implementations. We now ask: do these findings transfer to a real LLM-driven tool-call deployment, where the user agent is an LLM that parses natural-language queries, sanitizes via the mechanism, and emits structured replies through a tool-call interface?

\subsubsection{Setup.}
We re-run RQ1's threat model in a deployment harness: a GPT-5.4 nano
user agent holds the patient's private root values and answers per-turn
queries from an adversary by invoking sanitization tools and emitting
natural-language replies, which the adversary parses through an
alias-resolving regex parser frozen before the sweep. We sweep $8$
templates, $100$ patients/template, three mechanisms ($\mathcal{M}$-All,
$\mathcal{M}$-Roots, $\mathcal{M}$-Opt; Exponential), three budgets
($t \in \{k{+}1, 2k{+}1, 3k{+}1\}$), and three privacy levels
($\varepsilon \in \{0.01, 0.05, 0.1\}$) for $21{,}600$ sessions and
$162{,}000$ observations total.

\subsubsection{Aggregate ordering reproduces.}
Table~\ref{tab:agent_aggregate_wmape} reports aggregate wMAPE across the agent-eval sweep. The ordering $\mathcal{M}$-All $>$ $\mathcal{M}$-Roots $>$ $\mathcal{M}$-Opt holds at every cell, matching RQ1. At $\varepsilon = 0.1$, $B = (2k{+}1)\varepsilon$, agent-eval $\mathcal{M}$-Opt achieves $7.6\% \pm 0.5\%$ vs.\ $\mathcal{M}$-All's $22.2\% \pm 1.2\%$ --- a $2.9\times$ improvement, agreeing with RQ1's $2.6\times$ ($7.8\%$ vs.\ $20.3\%$) within $0.12$pp at the focal cell and within $\pm 1.9$pp across all nine $(\varepsilon, B)$ aggregate cells.

\subsubsection{Clinical decision preservation under deployment.}
Beyond numeric utility, the deployment evaluation lets us measure
whether the sanitization layer preserves the clinical decision the
service would make (Risk Class Error, Sec.~\ref{sec:metrics}).
Aggregated across templates at $\varepsilon = 0.1$, $\mathcal{M}$-Opt
achieves $7.3\%$ RCE at $B = (2k{+}1)\varepsilon$ vs.\ $\mathcal{M}$-All's
$19.2\%$ ($2.6\times$ reduction); the gap widens to $3.6\times$ at
$(3k{+}1)\varepsilon$ ($5.4\%$ vs.\ $19.7\%$). The deployment numbers
agree with RQ1's harness within bootstrap noise ($7.0\%$ harness vs.\
$7.3\%$ deployment for $\mathcal{M}$-Opt at the focal cell), confirming
the LLM-mediated tool call layer does not introduce systematic
clinical-decision drift. Table~\ref{tab:rce_per_template} reports
per-template RCE; the largest absolute gaps appear on compressing
formulas (HOMA, AIP, FIB4), where $\mathcal{M}$-All's high numeric
error spills into class-boundary crossings.

\begin{table}[h]
\centering
\setlength{\tabcolsep}{2pt}
% \footnotesize
\caption{Agent-eval Risk Class Error (\%) per template at
$\varepsilon = 0.1$ under deployment (RQ4), with bootstrap SE over
patients ($n = 100$). Showing the two larger adversarial budgets;
the $B = (k{+}1)\varepsilon$ column and the $\mathcal{M}$-Roots
configuration are in App.~\ref{app:agent-eval-rce}
(Table~\ref{tab:agent_rce_eps0.1_app}). Lower is better; deployment
ordering matches the controlled harness in
Table~\ref{tab:rce_per_template_rq1} within bootstrap noise.}
\label{tab:rce_per_template}
\begin{tabular}{l|cc|cc}
\toprule
& \multicolumn{2}{c|}{\textbf{$\mathcal{M}$-All}} & \multicolumn{2}{c}{\textbf{$\mathcal{M}$-Opt}} \\
Template & $2k{+}1$ & $3k{+}1$ & $2k{+}1$ & $3k{+}1$ \\
\midrule
HOMA     & 30.5\pmstd{3.2} & 33.0\pmstd{3.1} & 15.0\pmstd{2.5} & \textbf{10.5\pmstd{2.1}} \\
AIP      & 30.5\pmstd{3.2} & 29.0\pmstd{3.2} & 11.0\pmstd{2.3} & \textbf{8.0\pmstd{1.9}}   \\
FIB4     & 18.0\pmstd{2.6} & 18.0\pmstd{2.7} & 4.0\pmstd{1.4}  & \textbf{3.0\pmstd{1.2}}  \\
NLR      & 20.5\pmstd{2.8} & 24.5\pmstd{2.9} & 10.0\pmstd{2.1} & \textbf{7.5\pmstd{1.9}}   \\
TYG      & 30.0\pmstd{3.3} & 24.0\pmstd{3.1} & 12.5\pmstd{2.4} & \textbf{10.0\pmstd{2.2}} \\
ANEMIA   & 12.0\pmstd{2.3} & 17.5\pmstd{2.7} & 2.5\pmstd{1.1}  & \textbf{1.0\pmstd{0.7}}   \\
CONICITY & 10.5\pmstd{2.2} & 9.5\pmstd{2.0}  & 2.0\pmstd{1.0}  & \textbf{2.0\pmstd{1.0}}            \\
VASCULAR & 2.0\pmstd{1.0}  & 2.0\pmstd{1.0}  & 1.5\pmstd{0.9}  & \textbf{1.5\pmstd{0.9}}            \\
\bottomrule
\end{tabular}
\end{table}

\subsubsection{Per-template structure reproduces.}
The per-template ordering matches RQ1's structural analysis (Sec.~\ref{sec:experiments-rq3}). Templates with compressing target formulas (AIP, HOMA, FIB4) yield the largest absolute gaps under deployment as in the controlled harness; templates with amplifying formulas (ANEMIA, CONICITY) yield smaller absolute gaps. $\mathcal{M}$-Opt strictly improves on $\mathcal{M}$-All in $9/9$ (B, $\varepsilon$) cells for 7 of 8 templates; the eighth (NLR) shows $\mathcal{M}$-Opt $<$ $\mathcal{M}$-All in $7/9$ cells, reflecting NLR's narrow target domain and high baseline noise variance, where bootstrap intervals on M-All and M-Opt overlap at $B = (k{+}1)\varepsilon$.

\subsubsection{Pipeline robustness.}
The deployment-layer integration is end-to-end clean. Across 21{,}600 sessions, we recorded 162{,}000 sanitized observations. We observed zero sessions with parse failures (the LLM never emitted a malformed response that broke the adversary's parser), no observation-level parse failures and no rounding gaps in the LLM's text reply.

% \subsubsection{Cross-validation against RQ1.}
% At $\varepsilon = 0.1$, the agent-eval and RQ1 wMAPE values agree within $\pm 1.9$pp absolute across all nine $(\varepsilon, B)$ aggregate cells, with the headline cell ($\varepsilon = 0.1$, $B = (2k{+}1)\varepsilon$, $\mathcal{M}$-Opt) agreeing within $0.12$pp ($7.64\%$ deployment vs.\ $7.76\%$ harness). At lower $\varepsilon$ ($0.01$), where wMAPE is already saturated above $80\%$, absolute pp gaps grow but the qualitative ordering and per-template structure hold. This level of agreement is expected --- the two evaluations share patient data, mechanism, allocation, and threat model, differing only in the user-agent layer (LLM tool-call vs.\ direct mechanism call). The cross-validation confirms that LLM-side parsing, formatting, and prompt-handling do not introduce systematic biases at the operating points where the privacy-utility tradeoff is most informative.

\paragraph{Scope of evaluation.} Our deployment harness covers the Exponential mechanism, three $\varepsilon$ levels, and one LLM backbone (GPT-5.4 nano); transfer to other mechanisms, backbones, or fully autonomous agents that decide what to disclose is out of scope.

%% file: figures/rq4_aggregate_wmape.tex
\providecommand{\pmstd}[1]{{\scriptsize $\pm$#1}}

\begin{table*}[t]
\centering
\small
\caption{Agent-eval aggregate wMAPE (\%) averaged across 8 templates, per (config, $B$, $\varepsilon$). 21{,}600 sessions total, GPT-5.4 nano user agent, Exponential mechanism. Bootstrap SEs (composed across templates) shown. The mechanism-level ordering $\mathcal{M}$-All $>$ $\mathcal{M}$-Roots $>$ $\mathcal{M}$-Opt holds at every cell, matching RQ1 within bootstrap noise.}
\label{tab:agent_aggregate_wmape}
\resizebox{\textwidth}{!}{
\begin{tabular}{c|ccc|ccc|ccc}
\toprule
& \multicolumn{3}{c|}{\textbf{M-All}} & \multicolumn{3}{c|}{\textbf{M-Roots}} & \multicolumn{3}{c}{\textbf{M-Opt}} \\
$\varepsilon$ & $B=k+1$ & $B=2k+1$ & $B=3k+1$ & $B=k+1$ & $B=2k+1$ & $B=3k+1$ & $B=k+1$ & $B=2k+1$ & $B=3k+1$ \\
\midrule
0.01 & 98.8\pmstd{7.3} & 108\pmstd{7.9} & 115\pmstd{10} & 82.6\pmstd{5.0} & 55.5\pmstd{3.0} & 43.2\pmstd{2.2} & 71.6\pmstd{4.6} & 47.2\pmstd{2.6} & 35.5\pmstd{1.8} \\
0.05 & 32.8\pmstd{1.9} & 34.4\pmstd{1.8} & 34.6\pmstd{2.0} & 26.0\pmstd{1.4} & 17.8\pmstd{1.1} & 13.9\pmstd{0.9} & 20.9\pmstd{1.1} & 14.0\pmstd{0.8} & 10.5\pmstd{0.7} \\
0.1 & 21.6\pmstd{1.3} & 22.2\pmstd{1.2} & 21.7\pmstd{1.3} & 15.9\pmstd{1.0} & 10.6\pmstd{0.7} & 7.9\pmstd{0.6} & 12.3\pmstd{0.8} & 7.6\pmstd{0.5} & 5.6\pmstd{0.4} \\
\bottomrule
\end{tabular}
}
\end{table*}

%% file: sections/6_RelatedWork.tex
\section{Related Work}
\label{sec:related-work}

\paragraph{Inference-time prompt privacy.}
Prior sanitizers perturb prompts before submission via metric local DP over embeddings~\cite{yue2021santext,chen2022custext}, MLM sampling~\cite{dp-mlm}, context-aware bucket perturbation~\cite{wu2025capecontextawarepromptperturbation}, format-preserving numeric protection~\cite{chowdhury2025}, or token-level budget allocation~\cite{SYBW}; see~\cite{edemacu2025privacy} for a survey. All treat releases independently in single-turn settings. \sysname{} tracks dependencies across multi-turn interactions, noising only roots so that additional releases incur no additional cost.

\paragraph{DP under correlated data.}
Pufferfish~\cite{kifer2014pufferfish}, dependent DP~\cite{liu2016dependence}, Bayesian DP~\cite{zhao2019bayesian}, Wasserstein/Markov-Quilt mechanisms~\cite{song2017pufferfish}, R\'enyi composition~\cite{pierquin2024renyi}, and recent temporal/cross-sequence work~\cite{correlated2025csdp} all calibrate noise to account for statistical correlations \emph{between records}. We eliminate redundancy via post-processing rather than recalibrating noise.

\paragraph{LLM agent privacy and security.}
Agentic systems face prompt injection~\cite{greshake2023indirect,zhan2024injecagent,debenedetti2024agentdojo,hou2025mcp} and multi-turn extraction attacks~\cite{zhang2025searchingprivacyrisks} on the MCP protocol~\cite{anthropic2024mcp}. Architectural defenses include AirGapAgent~\cite{bagdasarian2024airgapagent} (dual-LLM minimization), PAPILLON~\cite{siyan2025papillon} (local-remote proxy redaction), and CaMeL~\cite{debenedetti2025camel} (capability tracking). These enforce minimization empirically; \sysname{} provides formal guarantees that hold \emph{even when} unintended computations occur, since they are post-processing of noised roots. Broader LLM privacy risks --- training-data extraction~\cite{carlini2021extracting}, attribute inference~\cite{staab2024beyond}, contextual-norm violations~\cite{mireshghallah2024can,shao2024privacylens} --- motivate client-side sanitization in the first place.

% \paragraph{Secure computation alternatives.} An orthogonal line of work protects medical data via cryptography or hardware isolation --- multiparty homomorphic encryption for clinical analytics~\cite{froelicher2021famhe} and TEE-based pipelines~\cite{segarra2019tee} --- offering full-computation confidentiality but at high compute and trust cost. \sysname{} takes a lighter approach: noise the roots once and let the post-processing theorem cover everything downstream, including the derived turns that would otherwise consume budget under adaptive composition~\cite{haney2023concurrent}.

%% file: sections/7_Conclusion.tex
\section{Conclusion}
\label{sec:conclusion}
We presented \sysname{}, a dependency-aware privacy mechanism for multi-turn agentic interactions: noise the roots once, compute everything else deterministically. Privacy then depends only on the initial root sanitization regardless of adversary function choices, query count, or turn order, while independent noising degrades with every additional release. The resulting double asymmetry --- additional turns simultaneously improve \sysname{}'s utility and degrade independent noising's privacy --- holds across 8 medical templates ($2.6$--$3.8\times$ lower wMAPE at $\varepsilon = 0.1$), across MAP reconstruction attacks ($\mathcal{M}$-All degrades from $16.5\%$ to $5.9\%$ as $q{:}1{\to}16$; \sysname{} invariant), and end-to-end through a 21{,}600-session LLM tool-call deployment (GPT-5.4 nano, $0\%$ compliance failures). 
% The structural cause is simple: when one value flows through multiple service interactions, independent noising re-spends budget on each release while a dependency-aware mechanism spends it once. We believe this perspective should inform privacy design as agentic systems reach more sensitive domains.

\paragraph{Limitations and future work.}
We assume perfect NER for root identification, following prior sanitizers~\cite{chowdhury2025,wu2025capecontextawarepromptperturbation}; integrating noisy NER into the privacy analysis is open. Our evaluation covers structured numeric values in single-user sessions; extending to free-form text, multi-user composition, additional backbones, and nonlinear derived-function attacks (which would tighten the $\mathcal{M}$-All vulnerability bound via Theorem~\ref{thm:privacy-leakage}) are natural next steps. Two further directions are promising: learning approximate target functions from public population data when the exact formula is unknown, and composing \sysname{} with capability-based agent defenses~\cite{bagdasarian2024airgapagent,debenedetti2025camel}.

%% file: sections/Appendix.tex
\section{Appendix}

\begin{figure*}[b]
\centering
\includegraphics[width=\textwidth]{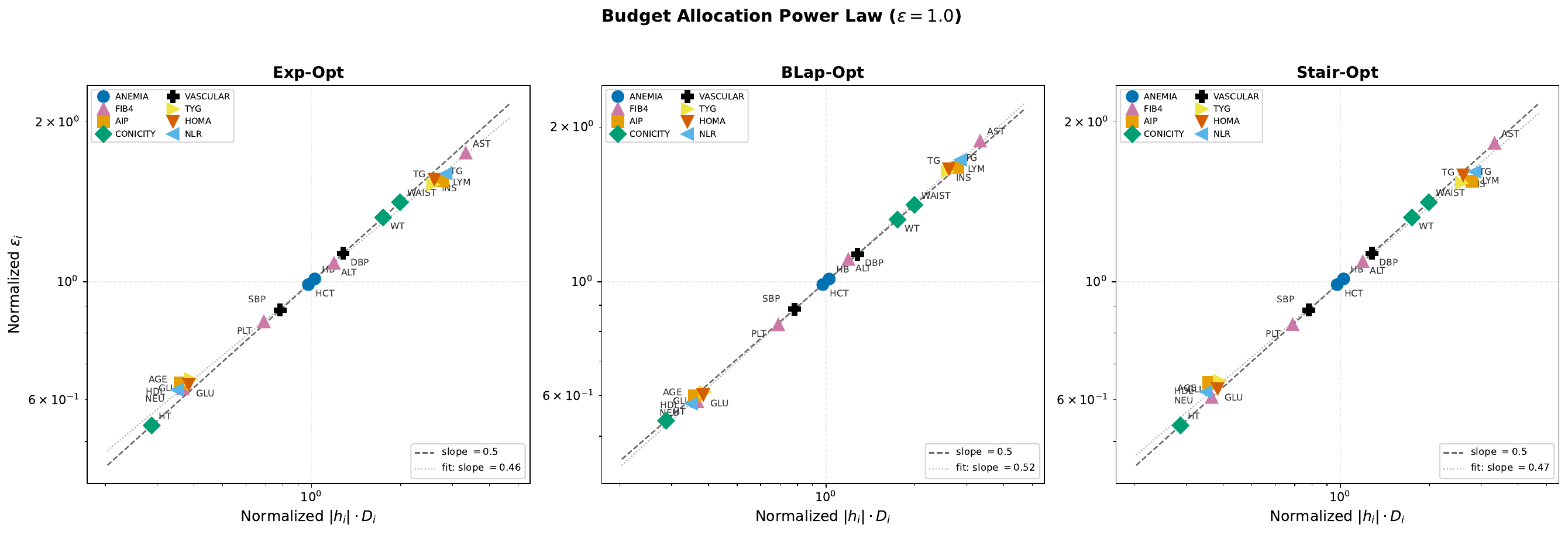}
\caption{Log-log scatter of normalized per-root budget $\varepsilon_i$ vs.\ sensitivity--domain product $|h_i| \cdot D_i$ at $\varepsilon = 0.1$.  Each panel uses one mechanism's optimizer.  All three produce the same power law $\varepsilon_i \propto (|h_i| \cdot D_i)^{0.5}$ (dashed line; fitted slopes 0.46, 0.52, 0.47).  Hollow markers indicate floor-clamped roots with zero target sensitivity.}
\label{fig:allocation_powerlaw_10}
\end{figure*}

\begin{figure*}[t]
\centering
\includegraphics[width=\textwidth]{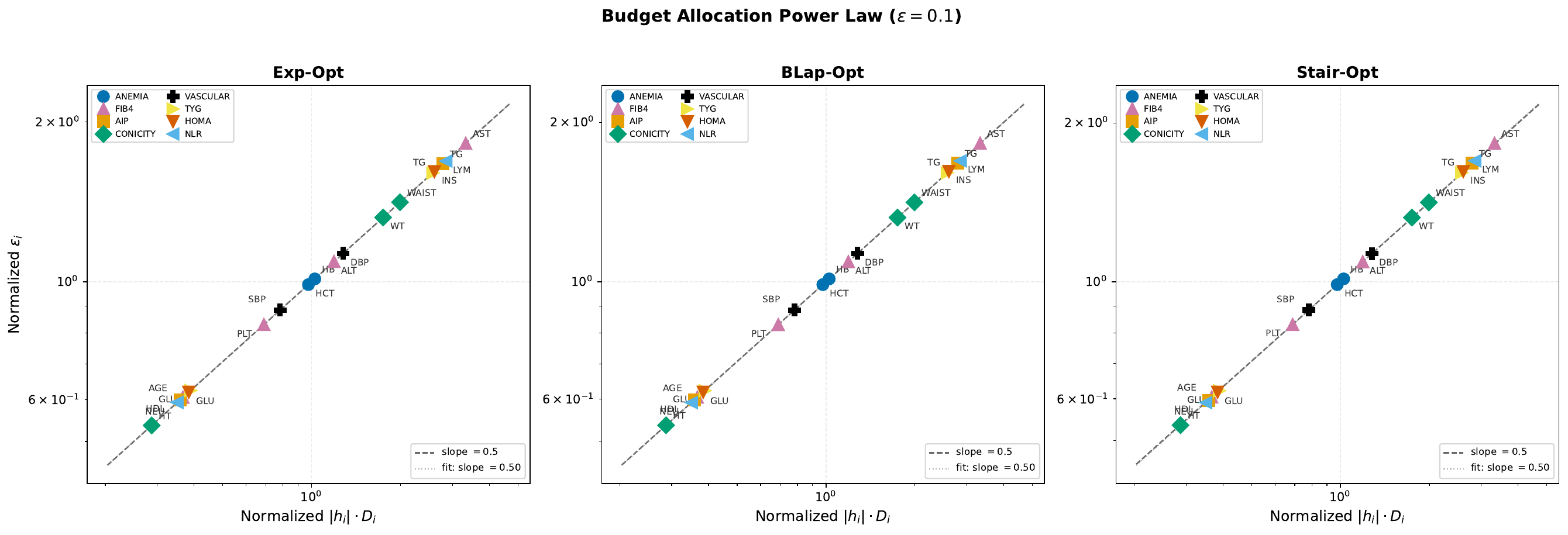}
\caption{Log-log scatter of normalized per-root budget $\varepsilon_i$ vs.\ sensitivity--domain product $|h_i| \cdot D_i$ at $\varepsilon = 0.1$.  Each panel uses one mechanism's optimizer.  All three produce the same power law $\varepsilon_i \propto (|h_i| \cdot D_i)^{0.5}$ (dashed line; fitted slopes 0.50, 0.50, 0.50).  Hollow markers indicate floor-clamped roots with zero target sensitivity.}
\label{fig:allocation_powerlaw_01}
\end{figure*}

\subsection{Medical Profile Details}
\label{app:med-profiles}
\subsubsection*{1. Anemia Classification (MCHC)}
Red blood cell indices are derived from hemoglobin (Hb), hematocrit (Hct), and red blood cell count (RBC):
\begin{align}
\text{MCV} &= \frac{\text{Hct}}{\text{RBC}} \times 10, \quad
\text{MCH} = \frac{\text{Hb}}{\text{RBC}} \times 10, \quad
\text{MCHC} = \frac{\text{Hb}}{\text{Hct}} \times 100
\end{align}
These are the standard red cell indices defined in clinical hematology~\cite{hoffbrand2019essential}. The target value MCHC (mean corpuscular hemoglobin concentration, in g/dL) classifies chromicity: hypochromic ($\text{MCHC} < 32$), normochromic ($32 \leq \text{MCHC} \leq 36$), or hyperchromic ($\text{MCHC} > 36$)~\cite{buttarello2008automated}.

\subsubsection*{2. Liver Fibrosis (FIB-4)}
The FIB-4 index~\cite{sterling2006development} is a validated non-invasive biomarker for hepatic fibrosis staging:
\begin{equation}
\text{FIB-4} = \frac{\text{age} \times \text{AST}}{\text{PLT} \times \sqrt{\text{ALT}}}
\end{equation}
where AST and ALT are in U/L and PLT is in $10^9$/L. Risk thresholds: low risk/rule out ($\text{FIB-4} < 1.30$), indeterminate ($1.30 \leq \text{FIB-4} \leq 2.67$), high risk/rule in ($\text{FIB-4} > 2.67$). These cutoffs are from Sterling et al.~\cite{sterling2006development} and validated in subsequent NAFLD studies~\cite{mcpherson2017age}.

\subsubsection*{3. Atherogenic Index of Plasma (AIP)}
Lipid-derived values include non-HDL cholesterol ($\text{TC} - \text{HDL}$) and Friedewald LDL ($\text{TC} - \text{HDL} - \text{TG}/5$)~\cite{friedewald1972estimation}. The atherogenic index of plasma~\cite{dobiasova2001atherogenic} is:
\begin{equation}
\text{AIP} = \log_{10}\!\left(\frac{\text{TG}}{\text{HDL}}\right)
\end{equation}
where TG and HDL are in mg/dL. Risk classification: low cardiovascular risk ($\text{AIP} < 0.11$), intermediate ($0.11 \leq \text{AIP} \leq 0.21$), high ($\text{AIP} > 0.21$)~\cite{dobiasova2004atherogenic}.

\subsubsection*{4. Conicity Index (Obesity)}
Body composition is assessed via BMI ($\text{wt}/\text{ht}^2$), waist-to-height ratio ($\text{waist}/\text{ht}$), and the conicity index~\cite{valdez1991simple}:
\begin{equation}
\text{CI} = \frac{\text{waist}}{0.109 \times \sqrt{\text{wt}/\text{ht}}}
\end{equation}
where waist and height are in meters, and weight is in kilograms. The denominator represents the circumference of a cylinder with the same height and mass. Central obesity is indicated when $\text{CI} > 1.25$~\cite{pitanga2005sensitivity}.

\subsubsection*{5. Vascular Stiffness (Pulse Pressure Index)}
Blood pressure derivatives include pulse pressure ($\text{PP} = \text{SBP} - \text{DBP}$), mean arterial pressure ($\text{MAP} = \text{DBP} + \text{PP}/3$), and mid-blood pressure ($\text{MBP} = (\text{SBP} + \text{DBP})/2$). The pulse pressure index~\cite{kovaite2008arterial}:
\begin{equation}
\text{PPI} = \frac{\text{PP}}{\text{SBP}} = \frac{\text{SBP} - \text{DBP}}{\text{SBP}}
\end{equation}
High arterial stiffness is indicated when $\text{PPI} > 0.60$~\cite{franklin2009predictive}.

\subsubsection*{6. Triglyceride-Glucose Index (TyG)}
The TyG index~\cite{simental2008triglyceride} is a surrogate marker for insulin resistance:
\begin{equation}
\text{TyG} = \ln\!\left(\frac{\text{TG} \times \text{Glu}}{2}\right)
\end{equation}
where TG is in mg/dL and Glu is fasting glucose in mg/dL. Insulin resistance is indicated when $\text{TyG} > 8.5$~\cite{guerrero2010product}.

\subsubsection*{7. HOMA-IR (Insulin Resistance)}
The Homeostatic Model Assessment for Insulin Resistance~\cite{matthews1985homeostasis}:
\begin{equation}
\text{HOMA-IR} = \frac{\text{Glu} \times \text{Ins}}{405}
\end{equation}
where Glu is fasting glucose (mg/dL) and Ins is fasting insulin ($\mu$U/mL). The constant 405 normalizes to conventional units. Classification: insulin sensitive ($\text{HOMA} < 1.0$), normal ($1.0 \leq \text{HOMA} < 2.5$), insulin resistant ($\text{HOMA} \geq 2.5$)~\cite{gayoso2016homa}.

\subsubsection*{8. Neutrophil-to-Lymphocyte Ratio (NLR)}
The NLR is a marker of systemic inflammation~\cite{zahorec2001ratio}:
\begin{equation}
\text{NLR} = \frac{\text{Neutrophils}}{\text{Lymphocytes}}
\end{equation}
with auxiliary values $\text{NLR}_{\text{sum}} = \text{Neu} + \text{Lym}$ and $\text{NLR}_{\text{diff}} = |\text{Neu} - \text{Lym}|$ included as intermediate nodes. Systemic inflammation or physiologic stress is indicated when $\text{NLR} \geq 3.0$~\cite{forget2017normal}.

\subsection{Additional Baseline Details}
\label{app:baselines}
\subsubsection{\textbf{Discrete Exponential Mechanism}}
\label{app:exp-noise}

The mechanism samples index $s \in \{0, \dots, m{-}1\}$ with probability
\[
    P_{t_i}(s) = \frac{\alpha_i^{|t_i - s|}}{\sum_{r=0}^{m-1} \alpha_i^{|t_i - r|}}, \qquad \alpha_i = \exp(-\epsilon_i/2).
\]
The expected absolute noise in index units at position $t_i$ is:
\begin{equation}
\label{eq:bar-eta-discrete}
    \mathbb{E}[|s - t_i|] = \frac{\sum_{s=0}^{m-1} |t_i - s|\,\alpha_i^{|t_i - s|}}{\sum_{s=0}^{m-1} \alpha_i^{|t_i - s|}}.
\end{equation}
In domain units: $\bar{\eta}_i^{\,\mathrm{Exp}}(\epsilon_i) = \delta_i \cdot \mathbb{E}[|s - t_i|]$. Evaluated at the grid center $t_i = (m{-}1)/2$, where both tails have full support.
%% -----------------------------------------------------------------

\subsubsection{\textbf{Bounded Laplace}}
\label{app:blap-alloc}
The Bounded Laplace mechanism sanitizes root $i$ as follows. The true value $x_i$ is mapped to index space via $t_i = (x_i - \ell_i)/\delta_i$. Continuous noise is drawn from a Laplace distribution:
\[
    \tilde{t}_i = t_i + Z, \qquad Z \sim \mathrm{Lap}(0, 1/\epsilon_i),
\]
with density $p(z) = (\epsilon_i / 2)\exp(-\epsilon_i |z|)$. The noisy index is clamped and rounded:
\[
    s = \mathrm{round} \big(\mathrm{clip}(\tilde{t}_i,\; 0,\; m{-}1)\big).
\]
The output is $x'_i = s\,\delta_i + \ell_i$.
%% -----------------------------------------------------------------

\subsubsection{\textbf{Staircase Mechanism}}
\label{app:stair-alloc}

The Staircase mechanism~\cite{geng2015optimal} replaces the Laplace's smooth exponential density with a piecewise-constant density. For sensitivity $\Delta = 1$ in index space and privacy parameter $\epsilon_i$, the noise $\hat{\eta}_i$ is sampled as follows:
\begin{enumerate}
    \item Draw a sign $S \in \{-1, +1\}$ uniformly at random.
    \item Draw $G \sim \mathrm{Geom}(1 - e^{-\epsilon_i})$, with $\Pr[G = g] = (1 - e^{-\epsilon_i})\,e^{-\epsilon_i g}$ for $g = 0, 1, 2, \ldots$\; The integer $G$ selects which ``step'' the noise lands on: the noise magnitude will fall in the interval $[G,\, G{+}1)$ in index units.
    \item Draw $U \sim \mathrm{Uniform}(0, 1)$.
    \item With probability $\gamma = 1/(1 + e^{\epsilon_i/2})$: set $\hat{\eta}_i = S \cdot (G + U)$.\\
          With probability $1 - \gamma$: set $\hat{\eta}_i = S \cdot (G + 1 - U)$.
\end{enumerate}

\noindent The geometric variable $G$ determines the step, with each successive step having probability $e^{-\epsilon_i}$ times the previous --- the same decay rate as the Laplace. The uniform variable $U$ places the noise within the step. The $\gamma$-branch orients $U$ toward the inner edge (closer to zero); the $(1{-}\gamma)$-branch orients it toward the outer edge. The parameter $\gamma = 1/(1 + e^{\epsilon_i/2})$ balances these orientations.

The continuous noise is clamped to $[0, m{-}1]$ and rounded to the nearest grid point. The output in domain units is $\eta_i = (s - t_i)\,\delta_i$.

\subsection{Per-template win rates}
\label{app:rq3-winrate}
 
Table~\ref{tab:rq1_win_rate} reports the per-template win rate of $\mathcal{M}$-Opt over $\mathcal{M}$-All across all mechanisms and privacy levels. Under the unified single-draw root-based protocol (Sec.~\ref{sec:experiments-rq1}), $\mathcal{M}$-Opt wins on essentially every (template, mechanism, $\varepsilon$) cell at every budget level: $98\%$ at $B = (k{+}1)\varepsilon$ and $> 99\%$ at the larger budgets. The remaining few losses concentrate on NLR at very low $\varepsilon$, where wMAPE is already saturated above $80\%$ and bootstrap intervals overlap.
 
\begin{table}[t]
\centering
\small
\caption{$\mathcal{M}$-Opt vs.\ $\mathcal{M}$-All win rate per template (across 3 mechanisms $\times$ 10 $\varepsilon$ values = 30 pairs). A ``win'' means $\mathcal{M}$-Opt achieves lower wMAPE.}
\label{tab:rq1_win_rate}
\begin{tabular}{lcccc}
\toprule
\textbf{Template} & $k$ & $k{+}1$ & $2k{+}1$ & $3k{+}1$ \\
\midrule
ANEMIA   & 3 & 30/30 & 30/30 & 30/30 \\
FIB4     & 4 & 30/30 & 30/30 & 30/30 \\
AIP      & 3 & 30/30 & 30/30 & 30/30 \\
CONICITY & 3 & 30/30 & 30/30 & 30/30 \\
VASCULAR & 2 & 30/30 & 30/30 & 30/30 \\
TYG      & 2 & 30/30 & 30/30 & 30/30 \\
HOMA     & 2 & 30/30 & 30/30 & 30/30 \\
NLR      & 2 & 26/30 & 29/30 & 29/30 \\
\midrule
\textbf{Total} & & 236/240 (98\%) & 239/240 (>99\%) & 239/240 (>99\%) \\
\bottomrule
\end{tabular}
\end{table}

\subsection{Reconstruction Attack Results (RQ2)}
\label{app:rq2-tables}
\label{app:reconstruction}

This appendix supports the claims of Sec.~\ref{sec:experiments-rq2} with three pieces of evidence: (i) a per-template breakdown showing that $\mathcal{M}$-All's reconstruction wMAPE degrades with query count on every template under the focal condition (Exponential, Strategy B, informed prior, $\varepsilon_r = 0.1$); (ii) Exponential aggregates under Strategy A and Strategy B for both priors, across all five privacy levels; (iii) aggregate reconstruction wMAPE for the Bounded Laplace and Staircase mechanisms under Strategy B (informed prior), confirming the same pattern observed for the Exponential mechanism (Table~\ref{tab:recon_main} in the main text). Bootstrap details and the full appendix (per-mechanism, per-strategy, per-prior, per-template at every $\varepsilon_r$) are in the supplementary material.

\paragraph{Per-template results (Exponential, $\varepsilon_r = 0.1$, Strategy B, informed prior).}
Table~\ref{tab:rq2_per_template} reports per-template reconstruction wMAPE under Strategy B with the informed prior at $\varepsilon_r = 0.1$. $\mathcal{M}$-All degrades monotonically with $q$ on every template; $\mathcal{M}$-Roots and $\mathcal{M}$-Opt are invariant. The largest absolute degradations occur on templates with wide root domains (HOMA, NLR, TYG).

\begin{table}[h]
\centering
\setlength{\tabcolsep}{3pt}
\small
\caption{Per-template reconstruction wMAPE (\%) at $\varepsilon_r = 0.1$ (Exponential, Strategy B, informed prior). $\mathcal{M}$-All degrades monotonically with $q$ on every template; $\mathcal{M}$-Roots and $\mathcal{M}$-Opt are invariant to $q$ (single value shown).}
\label{tab:rq2_per_template}
\begin{tabular}{l|cccc|cc}
\toprule
& \multicolumn{4}{c|}{$\mathcal{M}$-All} & $\mathcal{M}$-Roots & $\mathcal{M}$-Opt \\
\textbf{Template} & $q{=}1$ & $q{=}4$ & $q{=}8$ & $q{=}16$ & (any $q$) & (any $q$) \\
\midrule
ANEMIA   & 1.6  & 0.8  & 0.6  & 0.4  & 1.9  & 3.6 \\
FIB4     & 8.5  & 4.9  & 3.4  & 2.4  & 10.2 & 9.3 \\
AIP      & 11.2 & 7.2  & 5.6  & 3.9  & 15.0 & 14.3 \\
CONICITY & 2.1  & 0.9  & 0.7  & 0.5  & 2.2  & 1.9 \\
VASCULAR & 2.5  & 1.3  & 0.9  & 0.6  & 2.7  & 2.7 \\
TYG      & 15.5 & 11.0 & 7.3  & 5.0  & 20.4 & 17.3 \\
HOMA     & 19.1 & 9.7  & 6.9  & 5.1  & 24.3 & 20.5 \\
NLR      & 20.9 & 17.2 & 10.5 & 7.6  & 28.8 & 25.6 \\
\bottomrule
\end{tabular}
\end{table}

\paragraph{Exponential: Strategy A and Strategy B aggregates (both priors).}
Tables~\ref{tab:rq2_exp_A_uniform}--\ref{tab:rq2_exp_B_informed} report aggregate reconstruction wMAPE for the Exponential mechanism across all five privacy levels under both strategies and both priors. Under Strategy A (single targeted root $r^*$), $\mathcal{M}$-All has $q+1$ observations on $r^*$ at each $q$ and the gap to $\mathcal{M}$-Roots widens monotonically; under Strategy B (queries distributed across all roots), the same pattern holds for the average across roots. The informed prior gives qualitatively identical behavior to uniform — under v3 there is no target-release coupling, so the prior contributes only soft regularization.

\begin{table}[h]
\centering
\setlength{\tabcolsep}{2.5pt}
\small
\caption{Exponential, Strategy A, uniform prior: aggregate reconstruction wMAPE (\%) (mean across 8 templates). Reports the targeted root $r^*$ only.}
\label{tab:rq2_exp_A_uniform}
\begin{tabular}{c|ccc|ccc|ccc|ccc}
\toprule
& \multicolumn{3}{c|}{$q=1$} & \multicolumn{3}{c|}{$q=4$} & \multicolumn{3}{c|}{$q=8$} & \multicolumn{3}{c}{$q=16$} \\
$\varepsilon_r$ & All & Rts & Opt & All & Rts & Opt & All & Rts & Opt & All & Rts & Opt \\
\midrule
0.01 & 127  & 213  & 104  & 97.2 & 213  & 104  & 74.7 & 213  & 104  & 55.4 & 213  & 104 \\
0.05 & 26.1 & 42.4 & 27.9 & 23.3 & 42.4 & 27.9 & 16.9 & 42.4 & 27.9 & 12.3 & 42.4 & 27.9 \\
0.1  & 16.5 & 22.4 & 15.1 & 12.3 & 22.4 & 15.1 & 8.6  & 22.4 & 15.1 & 5.9  & 22.4 & 15.1 \\
0.5  & 4.9  & 5.3  & 3.4  & 2.3  & 5.3  & 3.4  & 1.7  & 5.3  & 3.4  & 1.1  & 5.3  & 3.4 \\
1.0  & 2.5  & 2.7  & 1.7  & 1.2  & 2.7  & 1.7  & 0.9  & 2.7  & 1.7  & 0.6  & 2.7  & 1.7 \\
\bottomrule
\end{tabular}
\end{table}

\begin{table}[h]
\centering
\setlength{\tabcolsep}{2.5pt}
\small
\caption{Exponential, Strategy A, informed prior: aggregate reconstruction wMAPE (\%) (mean across 8 templates). Targeted root only.}
\label{tab:rq2_exp_A_informed}
\begin{tabular}{c|ccc|ccc|ccc|ccc}
\toprule
& \multicolumn{3}{c|}{$q=1$} & \multicolumn{3}{c|}{$q=4$} & \multicolumn{3}{c|}{$q=8$} & \multicolumn{3}{c}{$q=16$} \\
$\varepsilon_r$ & All & Rts & Opt & All & Rts & Opt & All & Rts & Opt & All & Rts & Opt \\
\midrule
0.01 & 35.4 & 38.4 & 44.8 & 34.7 & 38.4 & 44.8 & 34.1 & 38.4 & 44.8 & 32.5 & 38.4 & 44.8 \\
0.05 & 28.3 & 40.3 & 29.1 & 19.7 & 40.3 & 29.1 & 15.6 & 40.3 & 29.1 & 11.7 & 40.3 & 29.1 \\
0.1  & 17.8 & 23.8 & 15.9 & 12.3 & 23.8 & 15.9 & 8.3  & 23.8 & 15.9 & 5.9  & 23.8 & 15.9 \\
0.5  & 5.0  & 5.3  & 3.4  & 2.3  & 5.3  & 3.4  & 1.7  & 5.3  & 3.4  & 1.1  & 5.3  & 3.4 \\
1.0  & 2.6  & 2.7  & 1.7  & 1.2  & 2.7  & 1.7  & 0.9  & 2.7  & 1.7  & 0.6  & 2.7  & 1.7 \\
\bottomrule
\end{tabular}
\end{table}

\begin{table}[h]
\centering
\setlength{\tabcolsep}{2.5pt}
\small
\caption{Exponential, Strategy B, uniform prior: aggregate reconstruction wMAPE (\%) (mean across 8 templates). Reports the average across all $k$ roots.}
\label{tab:rq2_exp_B_uniform}
\begin{tabular}{c|ccc|ccc|ccc|ccc}
\toprule
& \multicolumn{3}{c|}{$q=1$} & \multicolumn{3}{c|}{$q=4$} & \multicolumn{3}{c|}{$q=8$} & \multicolumn{3}{c}{$q=16$} \\
$\varepsilon_r$ & All & Rts & Opt & All & Rts & Opt & All & Rts & Opt & All & Rts & Opt \\
\midrule
0.01 & 73.8 & 115  & 84.9 & 54.5 & 115  & 84.9 & 41.7 & 115  & 84.9 & 30.7 & 115  & 84.9 \\
0.05 & 16.6 & 23.8 & 22.4 & 12.6 & 23.8 & 22.4 & 9.2  & 23.8 & 22.4 & 6.6  & 23.8 & 22.4 \\
0.1  & 9.9  & 12.5 & 13.5 & 6.6  & 12.5 & 13.5 & 4.7  & 12.5 & 13.5 & 3.2  & 12.5 & 13.5 \\
0.5  & 2.7  & 2.9  & 5.5  & 1.3  & 2.9  & 5.5  & 0.9  & 2.9  & 5.5  & 0.6  & 2.9  & 5.5 \\
1.0  & 1.3  & 1.5  & 4.3  & 0.7  & 1.5  & 4.3  & 0.5  & 1.5  & 4.3  & 0.3  & 1.5  & 4.3 \\
\bottomrule
\end{tabular}
\end{table}

\begin{table}[h]
\centering
\setlength{\tabcolsep}{2.5pt}
\small
\caption{Exponential, Strategy B, informed prior: aggregate reconstruction wMAPE (\%) (mean across 8 templates). This is the focal condition reported in main-text Table~\ref{tab:recon_main}.}
\label{tab:rq2_exp_B_informed}
\begin{tabular}{c|ccc|ccc|ccc|ccc}
\toprule
& \multicolumn{3}{c|}{$q=1$} & \multicolumn{3}{c|}{$q=4$} & \multicolumn{3}{c|}{$q=8$} & \multicolumn{3}{c}{$q=16$} \\
$\varepsilon_r$ & All & Rts & Opt & All & Rts & Opt & All & Rts & Opt & All & Rts & Opt \\
\midrule
0.01 & 24.7 & 28.8 & 30.3 & 23.3 & 28.8 & 30.3 & 21.6 & 28.8 & 30.3 & 19.6 & 28.8 & 30.3 \\
0.05 & 16.4 & 22.8 & 21.0 & 10.9 & 22.8 & 21.0 & 8.6  & 22.8 & 21.0 & 6.4  & 22.8 & 21.0 \\
0.1  & 10.2 & 13.2 & 11.9 & 6.6  & 13.2 & 11.9 & 4.5  & 13.2 & 11.9 & 3.2  & 13.2 & 11.9 \\
0.5  & 2.8  & 2.9  & 3.4  & 1.3  & 2.9  & 3.4  & 0.9  & 2.9  & 3.4  & 0.6  & 2.9  & 3.4 \\
1.0  & 1.4  & 1.5  & 2.3  & 0.7  & 1.5  & 2.3  & 0.5  & 1.5  & 2.3  & 0.3  & 1.5  & 2.3 \\
\bottomrule
\end{tabular}
\end{table}

\paragraph{Bounded Laplace and Staircase aggregates (Strategy B, informed prior).}
Tables~\ref{tab:rq2_blap_agg} and~\ref{tab:rq2_stair_agg} report aggregate reconstruction wMAPE for the Bounded Laplace and Staircase mechanisms under the focal condition (Strategy B, informed prior, mean across 8 templates). Both follow the same pattern as Exponential: $\mathcal{M}$-All decreases with $q$ at every $\varepsilon_r$, while $\mathcal{M}$-Roots and $\mathcal{M}$-Opt are invariant.

\begin{table}[h]
\centering
\setlength{\tabcolsep}{2.5pt}
\small
\caption{Bounded Laplace, Strategy B, informed prior: aggregate reconstruction wMAPE (\%) (mean across 8 templates).}
\label{tab:rq2_blap_agg}
\begin{tabular}{c|ccc|ccc|ccc|ccc}
\toprule
& \multicolumn{3}{c|}{$q=1$} & \multicolumn{3}{c|}{$q=4$} & \multicolumn{3}{c|}{$q=8$} & \multicolumn{3}{c}{$q=16$} \\
$\varepsilon_r$ & All & Rts & Opt & All & Rts & Opt & All & Rts & Opt & All & Rts & Opt \\
\midrule
0.01 & 25.9 & 31.1 & 32.4 & 23.4 & 31.1 & 32.4 & 24.0 & 31.1 & 32.4 & 23.1 & 31.1 & 32.4 \\
0.05 & 11.5 & 13.5 & 12.2 & 7.7  & 13.5 & 12.2 & 6.3  & 13.5 & 12.2 & 5.4  & 13.5 & 12.2 \\
0.1  & 6.6  & 7.2  & 6.8  & 3.5  & 7.2  & 6.8  & 2.7  & 7.2  & 6.8  & 2.0  & 7.2  & 6.8 \\
0.5  & 1.4  & 1.5  & 2.3  & 0.7  & 1.5  & 2.3  & 0.5  & 1.5  & 2.3  & 0.4  & 1.5  & 2.3 \\
1.0  & 0.7  & 0.8  & 1.7  & 0.4  & 0.8  & 1.7  & 0.3  & 0.8  & 1.7  & 0.2  & 0.8  & 1.7 \\
\bottomrule
\end{tabular}
\end{table}

\begin{table}[h]
\centering
\setlength{\tabcolsep}{2.5pt}
\small
\caption{Staircase, Strategy B, informed prior: aggregate reconstruction wMAPE (\%) (mean across 8 templates). $\varepsilon_r = 1.0$ omitted (mechanism degenerate at high budget).}
\label{tab:rq2_stair_agg}
\begin{tabular}{c|ccc|ccc|ccc|ccc}
\toprule
& \multicolumn{3}{c|}{$q=1$} & \multicolumn{3}{c|}{$q=4$} & \multicolumn{3}{c|}{$q=8$} & \multicolumn{3}{c}{$q=16$} \\
$\varepsilon_r$ & All & Rts & Opt & All & Rts & Opt & All & Rts & Opt & All & Rts & Opt \\
\midrule
0.01 & 25.8 & 30.4 & 31.4 & 23.7 & 30.4 & 31.4 & 23.7 & 30.4 & 31.4 & 23.2 & 30.4 & 31.4 \\
0.05 & 11.6 & 13.6 & 12.0 & 7.8  & 13.6 & 12.0 & 6.6  & 13.6 & 12.0 & 5.6  & 13.6 & 12.0 \\
0.1  & 6.6  & 7.1  & 6.6  & 3.5  & 7.1  & 6.6  & 2.7  & 7.1  & 6.6  & 2.1  & 7.1  & 6.6 \\
0.5  & 1.5  & 1.4  & 1.6  & 0.7  & 1.4  & 1.6  & 0.5  & 1.4  & 1.6  & 0.4  & 1.4  & 1.6 \\
\bottomrule
\end{tabular}
\end{table}

\subsection{RQ1: Per-template wMAPE Curves}
\label{app:rq1-tables-figs}

Fig.~\ref{fig:per_template_app} shows the per-template wMAPE for all 8 medical
templates across the full $\varepsilon$ sweep at $B = (2k{+}1)\varepsilon$.
Templates are arranged in two rows by absolute-gap regime (Sec.~\ref{sec:experiments-rq3}):
the top row contains compressing-formula templates with the largest absolute
$\mathcal{M}$-All vs.\ \sysname{} gaps; the bottom row contains amplifying-formula
templates with smaller absolute gaps but comparable multiplicative ratios. Each
panel overlays nine methods (3 mechanisms $\times$ 3 configurations).

\begin{figure*}[t]
    \centering
    \includegraphics[width=\linewidth]{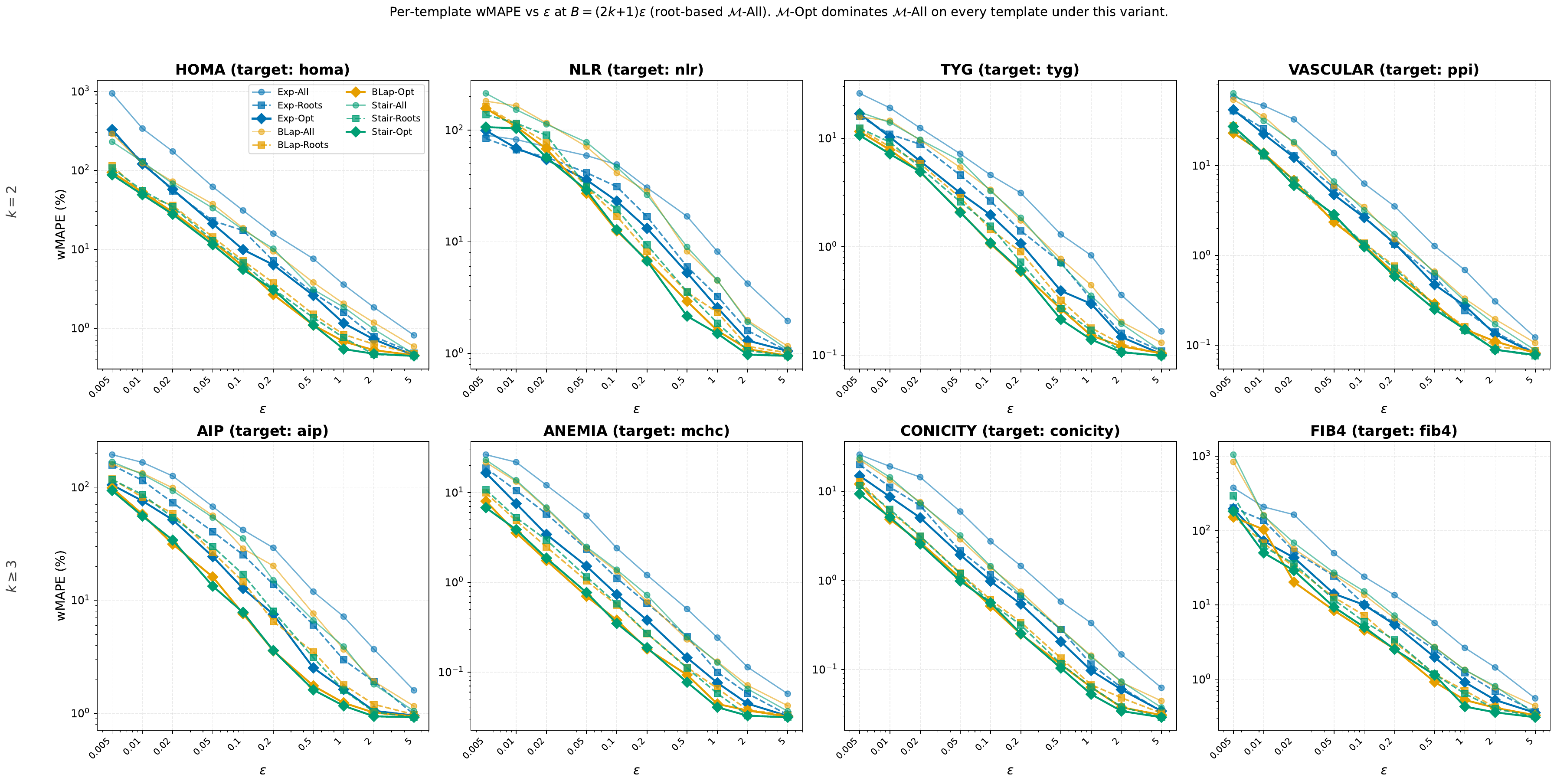}
    \caption{Per-template wMAPE (\%) vs.\ $\varepsilon$ at $B = (2k{+}1)\varepsilon$.
    Top row: compressing formulas with high-error baselines (FIB4, HOMA, AIP,
    NLR). Bottom row: amplifying formulas with already-low base errors (ANEMIA,
    CONICITY, TYG, VASCULAR). Nine methods per panel
    (3 mechanisms $\times$ 3 variants). The headline cell ($\varepsilon = 0.1$,
    Exponential) is summarized in main-text Table~\ref{tab:per_template_summary};
    full numerical tables follow in App.~\ref{app:rq1-tables}.}
    \label{fig:per_template_app}
\end{figure*}

\subsection{Per-root Budget Allocation}
\label{app:per-root-budget}
Figure~\ref{fig:allocation_bars} shows the per-root budget shares across templates, mechanisms, and privacy levels: $\mathcal{M}$-Opt allocates more budget to influential roots in every condition. Figures~\ref{fig:allocation_powerlaw_10} and~\ref{fig:allocation_powerlaw_01} confirm the closed-form prediction $\varepsilon_i \propto \sqrt{|h_i|\,\delta_i}$ (Sec.~\ref{sec:opt-budget}) on log-log axes for all three mechanisms at $\varepsilon = 0.1$ and $\varepsilon = 1.0$, with fitted slopes within $\pm 0.02$ of the predicted $0.5$.
% We show the per root budget allocation across different templates, mechanisms and privacy parameters in Figure~\ref{fig:allocation_bars}. All $\mathcal{M}$-Opt methods exploit the budget to minimize target node error, by allocating more budget to influential nodes. Figure~\ref{fig:allocation_powerlaw_01}---~\ref{fig:allocation_powerlaw_10} show the allocation following the power law, for $\varepsilon = 0.1$ and $\varepsilon = 1.0$ respectively.

\begin{figure*}[t]
\centering
\includegraphics[width=\textwidth]{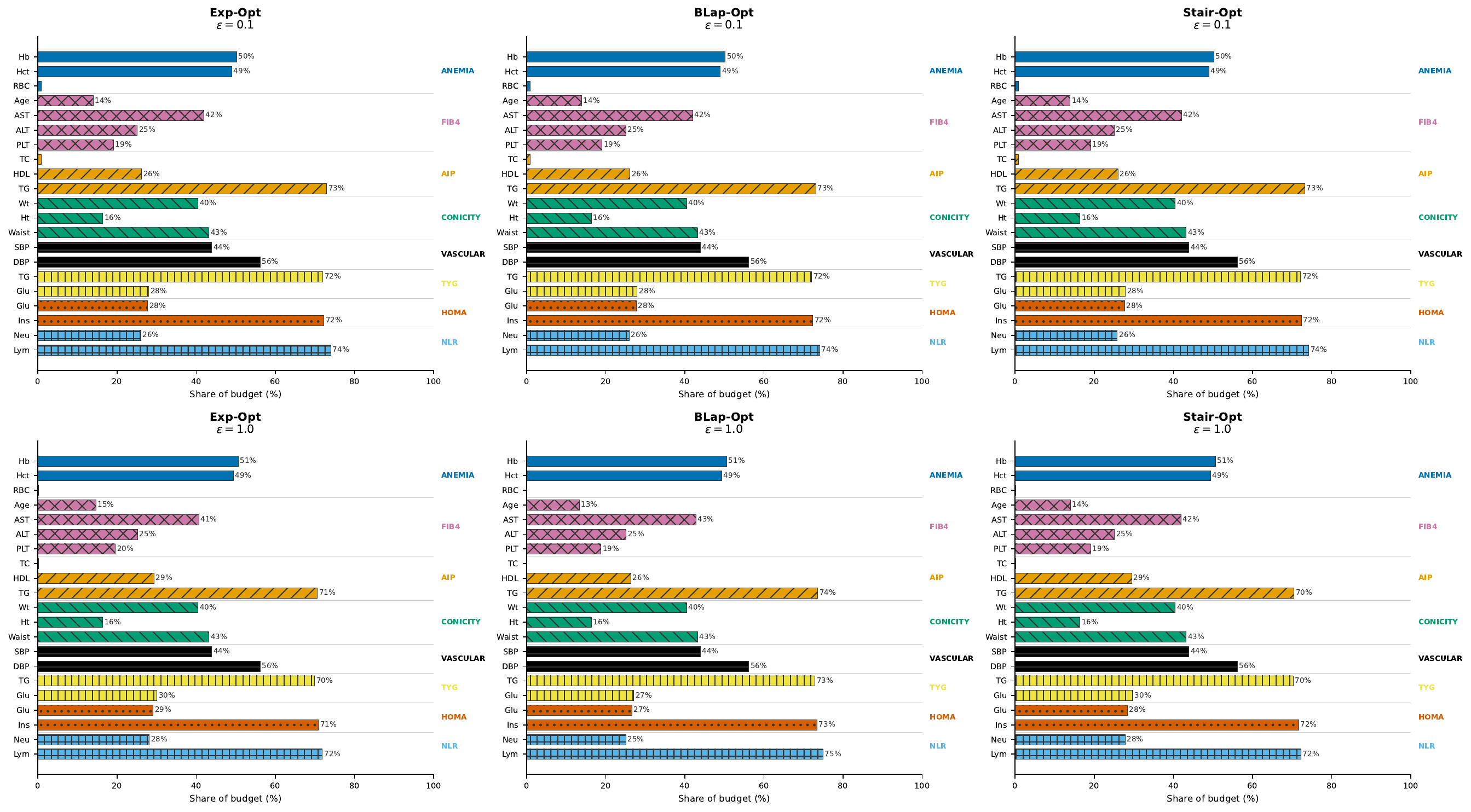}
\caption{Per-root budget shares (\%) for three mechanisms at $\varepsilon = 0.1$ (left three panels) and $\varepsilon = 1.0$ (right three).  Allocations are nearly identical across mechanisms at low $\varepsilon$; at $\varepsilon = 1.0$ minor mechanism-specific differences emerge (e.g., BLap assigns 43\% to AST in FIB4 vs.\ Exp's 41\%).}
\label{fig:allocation_bars}
\end{figure*}

\subsection{Closed-Form Budget Allocation in the Laplace Limit}
\label{app:budget-derivation}
 
Sec.~\ref{sec:opt-budget} establishes the optimization
\[
    \min_{\epsilon_1, \dots, \epsilon_k} \sum_{i=1}^k |h_i| \cdot \bar{\eta}_i(\epsilon_i)
    \qquad \text{s.t.} \qquad
    \sum_{i=1}^k \epsilon_i = B, \quad \epsilon_i \geq \epsilon_{\min}.
\]
The objective is solved numerically for each mechanism. Here we derive the closed-form solution in the Laplace limit, where the discrete index-space distribution of Sec.~\ref{sec:sanitization} converges to a continuous Laplace distribution as the grid becomes fine ($m \to \infty$).

\paragraph{Setup and bound.} Each root is independently sanitized via $\mathcal{M}_{\epsilon_i}$, producing output index $s$ and domain-unit noise $\eta_i = (s - t_i)\,\delta_i$. Let $x' = (x_1', \dots, x_k')$ be the sanitized vector with $x_i' = s\,\delta_i + \ell_i$. The error in the intermediate $g$ is $\Delta = g(x') - g(x)$; a first-order expansion gives $\Delta \approx \sum_{i=1}^{k} h_i\,\eta_i$. Applying the triangle inequality inside the expectation,
\begin{align*}
    \mathbb{E}\big[|\Delta|\big]
    \;\approx\; \mathbb{E}\Big[\Big|\sum_{i=1}^k h_i\,\eta_i\Big|\Big]
    \;\leq\; \sum_{i=1}^k |h_i|\,\mathbb{E}\big[|\eta_i|\big].
\end{align*}
Since the $\eta_i$ are independent (each mechanism uses independent randomness), this upper bound is tight up to cancellation effects. The expected absolute noise $\mathbb{E}[|\eta_i|]$ depends on $t_i$, the true value's position within the grid; to keep the optimization data-independent, we take the worst case over all positions:
\begin{equation}
\label{eq:bar-eta}
    \bar{\eta}_i(\epsilon_i) = \sup_{t_i \in [0,\, m-1]} \mathbb{E}\big[|\eta_i|\big].
\end{equation}
This yields the optimization restated above.
 
\paragraph{Continuum approximation.} For the discrete Exponential mechanism, the sampling weight at index $s$ is $P(s) \propto \exp(-\epsilon_i |t_i - s|/2)$. As $m \to \infty$, this converges to a Laplace distribution centered at $t_i$ with scale $b_i = 2/\epsilon_i$ in index units. The expected absolute deviation of a Laplace with scale $b_i$ is $b_i$. In domain units, $\bar{\eta}_i(\epsilon_i) \approx 2\delta_i / \epsilon_i$. The Bounded Laplace mechanism is asymptotically identical (truncation effects are exponentially small in $\epsilon_i (m{-}1)$).
 
\paragraph{Lagrangian.} Substituting into the objective, the Lagrangian for the active constraint set ($\epsilon_i > \epsilon_{\min}$) becomes
\begin{align*}
    \mathcal{L}(\epsilon_1, \dots, \epsilon_k, \lambda) = \sum_{i=1}^{k} |h_i|\,\frac{2\delta_i}{\epsilon_i} + \lambda\left(\sum_{i=1}^{k} \epsilon_i - B\right),
\end{align*}
where $\lambda \in \mathbb{R}$ is the multiplier for the budget constraint. Taking $\partial \mathcal{L}/\partial \epsilon_i = 0$,
\[
    -\frac{2|h_i|\,\delta_i}{\epsilon_i^2} + \lambda = 0
    \quad \Longrightarrow \quad
    \epsilon_i^* = \sqrt{\frac{2|h_i|\,\delta_i}{\lambda}}
    \;\propto\; \sqrt{|h_i|\,\delta_i}.
\]
Let $w_i = \sqrt{|h_i|\,\delta_i}$. Substituting $\epsilon_i^* = w_i \sqrt{2/\lambda}$ into the budget constraint $\sum_{i \in \mathcal{R}_A} \epsilon_i = B - |\mathcal{R}_S|\,\epsilon_{\min}$ (where $\mathcal{R}_A$ is the set of active roots and $\mathcal{R}_S$ is the set of saturated roots) and solving for $\lambda$ gives
\begin{equation}
\label{eq:closed-form}
    \tilde{\epsilon}_i = \frac{w_i}{\sum_{j \in \mathcal{R}_A} w_j} \left(B - |\mathcal{R}_S|\,\epsilon_{\min}\right).
\end{equation}
The second derivative of the objective, $4|h_i|\delta_i / \epsilon_i^3 > 0$ for all $\epsilon_i > 0$, confirms strict convexity, so this stationary point is the unique global minimum. This matches the closed-form claim made in Sec.~\ref{sec:opt-budget}.
 
\paragraph{Saturation handling.} We solve~\eqref{eq:closed-form} iteratively: roots whose allocation falls below $\epsilon_{\min}$ are clamped, removed from $\mathcal{R}_A$, added to $\mathcal{R}_S$, and the remaining budget is redistributed among the still-active roots. This terminates after at most $k$ iterations.
 
\paragraph{Discrete mechanisms.} For the discrete Exponential, Bounded Laplace, and Staircase mechanisms, we solve the optimization directly over the discrete $\bar{\eta}_i(\epsilon_i)$ via standard numerical optimization. RQ3 (Sec.~\ref{sec:experiments-rq3}) verifies that the scaling $\epsilon_i \propto \sqrt{|h_i| \delta_i}$ predicted by the Laplace closed-form holds empirically across all three discrete mechanisms (fitted log-log slope $0.50$ at $\varepsilon = 0.1$).
 
\paragraph{Sensitivity computation.} Sec.~\ref{sec:opt-budget} requires the sensitivity weights $h_i = |\partial g / \partial x_i(\mu)|$, evaluated at the population mean $\mu = (\mu_1, \dots, \mu_k)$. \sysname{} computes these via a single forward-mode automatic differentiation pass through the conversation DAG:
\begin{enumerate}
    \item Set all root values to their population means: $x = \mu$.
    \item Propagate values forward through the DAG in topological order, computing all intermediate node values.
    \item Propagate gradient seeds forward using the chain rule with exact analytical local partial derivatives at each node.
    \item Take absolute values: $h_i = |\partial g / \partial x_i(\mu)|$.
\end{enumerate}
Since target formulas may be nonlinear, $\partial g / \partial x_i$ can depend on the values of other roots; evaluating at $\mu$ captures these cross-dependencies at a representative operating point. The population means are domain knowledge --- in our experiments, computed from the NHANES~\cite{nhanes2017exam} reference population (excluding test data). If the deployment population differs substantially from the reference, the allocation may be suboptimal for that population, but the privacy guarantee remains $B$-mDP regardless: the weights $h_i$ enter the utility objective only and do not affect the privacy proof.

% ============================================================================
% RCE APPENDIX SUBSECTION (RQ4)
%
% Place this in sections/Appendix.tex, after the existing app:rq2-tables
% subsection (or wherever you want it - the appendix is unbounded).
%
% Purpose: provides per-template agent-eval RCE numbers with bootstrap SEs
% to back the compressed §5.6 RQ4 RCE subsubsection in the body.
%
% The body subsubsection references this only implicitly through
% Fig.~\ref{fig:rce_grid} (which is in the body). This appendix subsection
% provides the numerical detail behind the figure for reviewers who want
% exact cell values.
% ============================================================================

\subsection{Agent-Eval RCE: Per-Template Numbers (RQ4)}
\label{app:agent-eval-rce}

Table~\ref{tab:agent_rce_eps0.1_app} reports the per-template
deployment-RCE numbers, with
bootstrap standard errors over patients. Aggregated across templates at
the focal cell ($\varepsilon = 0.1$, $B = (2k{+}1)\varepsilon$),
$\mathcal{M}$-Opt achieves $7.3\%$ RCE vs.\ $\mathcal{M}$-All's $19.2\%$;
at $B = (3k{+}1)\varepsilon$, $5.4\%$ vs.\ $19.7\%$. The cell-by-cell
ordering matches the harness RCE results
(App.~\ref{app:rq1-tables}, Tables~\ref{tab:rq1_kplus1_rce}--\ref{tab:rq1_3kplus1_rce}).

\begin{table*}[t]
\centering
\small
\setlength{\tabcolsep}{3pt}
\caption{Per-template agent-eval Risk Class Error (\%) at $\varepsilon = 0.1$,
with bootstrap SE over patients ($n = 100$ per template-cell). Best per
row in \textbf{bold}.}
\label{tab:agent_rce_eps0.1_app}
\begin{tabular}{l|ccc|ccc|ccc}
\toprule
& \multicolumn{3}{c|}{\textbf{$\mathcal{M}$-All}} & \multicolumn{3}{c|}{\textbf{$\mathcal{M}$-Roots}} & \multicolumn{3}{c}{\textbf{$\mathcal{M}$-Opt}} \\
Template & $B{=}k{+}1$ & $B{=}2k{+}1$ & $B{=}3k{+}1$ & $B{=}k{+}1$ & $B{=}2k{+}1$ & $B{=}3k{+}1$ & $B{=}k{+}1$ & $B{=}2k{+}1$ & $B{=}3k{+}1$ \\
\midrule
HOMA     & 28.0\pmstd{3.2} & 30.5\pmstd{3.2} & 33.0\pmstd{3.1} & 27.0\pmstd{3.0} & 18.5\pmstd{2.7} & 15.5\pmstd{2.5} & 21.0\pmstd{2.8} & 15.0\pmstd{2.5} & \textbf{10.5\pmstd{2.1}} \\
ANEMIA   & 10.0\pmstd{2.2} & 12.0\pmstd{2.3} & 17.5\pmstd{2.7} & 11.0\pmstd{2.2} & 4.0\pmstd{1.4}  & 3.0\pmstd{1.2}  & 5.5\pmstd{1.6}  & 2.5\pmstd{1.1}  & \textbf{1.0\pmstd{0.7}}  \\
FIB4     & 13.5\pmstd{2.4} & 18.0\pmstd{2.6} & 18.0\pmstd{2.7} & 12.5\pmstd{2.3} & 5.0\pmstd{1.6}  & 3.5\pmstd{1.3}  & 9.5\pmstd{2.1}  & 4.0\pmstd{1.4}  & \textbf{3.0\pmstd{1.2}}  \\
AIP      & 34.5\pmstd{3.3} & 30.5\pmstd{3.2} & 29.0\pmstd{3.2} & 30.5\pmstd{3.2} & 22.5\pmstd{2.9} & 17.0\pmstd{2.7} & 18.5\pmstd{2.7} & 11.0\pmstd{2.3} & \textbf{8.0\pmstd{1.9}}   \\
CONICITY & 9.0\pmstd{1.9}  & 10.5\pmstd{2.2} & 9.5\pmstd{2.0}  & 5.5\pmstd{1.6}  & 2.5\pmstd{1.1}  & 2.5\pmstd{1.1}  & 3.0\pmstd{1.2}  & \textbf{2.0\pmstd{1.0}} & 2.0\pmstd{1.0}           \\
VASCULAR & \textbf{1.5\pmstd{0.8}} & 2.0\pmstd{1.0}  & 2.0\pmstd{1.0}  & 1.5\pmstd{0.9}  & 1.5\pmstd{0.9}  & 1.5\pmstd{0.9}  & 1.5\pmstd{0.9}  & 1.5\pmstd{0.9}  & 1.5\pmstd{0.9}           \\
TYG      & 27.0\pmstd{3.0} & 30.0\pmstd{3.3} & 24.0\pmstd{3.1} & 24.0\pmstd{3.2} & 15.5\pmstd{2.7} & 13.0\pmstd{2.5} & 18.0\pmstd{2.8} & 12.5\pmstd{2.4} & \textbf{10.0\pmstd{2.2}} \\
NLR      & 18.5\pmstd{2.7} & 20.5\pmstd{2.8} & 24.5\pmstd{2.9} & 15.0\pmstd{2.6} & 13.0\pmstd{2.4} & 10.5\pmstd{2.2} & 15.0\pmstd{2.6} & 10.0\pmstd{2.1} & \textbf{7.5\pmstd{1.9}}   \\
\bottomrule
\end{tabular}
\end{table*}

\subsection{Proof of Theorem~\ref{thm:privacy-guarantee}}
\label{app:thm:privacy-guarantee}

\begin{proof}[Proof of Theorem~\ref{thm:privacy-guarantee}]
By the privacy game definition,
\begin{align*}
\mathrm{Adv}_{\mathrm{DAG}}(\mathcal{A})
  &= |2\Pr[G_{\mathrm{DAG}}(\mathcal{A}) = 1] - 1| \\
  &= |\Pr[b'{=}0 \mid b{=}0] - \Pr[b'{=}0 \mid b{=}1]|.    
\end{align*}

The adversary's view in the game is fully determined by $\hat{X}^b$: $\mathcal{A}$ chooses $F$ at challenge time and computes $\hat{G}_b = C(\hat{X}^b, F)$ deterministically. Hence $\Pr[b'{=}0 \mid b]$ depends only on the distribution of $\hat{X}^b$ (and on $\mathcal{A}$'s fixed strategy on its view).

\medskip
\noindent\textbf{Hybrid argument.} Without loss of generality, assume $X^0$ and $X^1$ differ in exactly the roots $\{1, \ldots, k\}$. Define hybrids
\[
  X^{(0)} = X^0,\qquad X^{(j)} = (x^1_1, \ldots, x^1_j, x^0_{j+1}, \ldots, x^0_k),\qquad X^{(k)} = X^1.
\]
Each consecutive pair $(X^{(j-1)}, X^{(j)})$ differs in exactly root $j$. Let
\[
  p_j := \Pr[\mathcal{A}\text{ outputs } 0 \mid \text{input } X^{(j)}\text{ is sanitized}].
\]
Then $\Pr[b'{=}0 \mid b{=}0] = p_0$ and $\Pr[b'{=}0 \mid b{=}1] = p_k$, so
\[
  \mathrm{Adv}_{\mathrm{DAG}}(\mathcal{A}) = |p_0 - p_k|
  = \Big| \sum_{j=1}^{k} (p_{j-1} - p_j) \Big|
  \leq \sum_{j=1}^{k} |p_{j-1} - p_j|.
\]

\medskip
\noindent\textbf{Single-root step.} For each $j$, $X^{(j-1)}$ and $X^{(j)}$ agree everywhere except at root $j$ (where they take values $x^0_j$ and $x^1_j$ respectively), and the rest of the noised vector is sampled from the same distribution. Marginalizing over the shared roots, $|p_{j-1} - p_j|$ depends only on the distribution of the $j$-th sanitized component. Conditioning on the realized values of all other (shared) roots, denote
\[
  q_j(z) := \Pr[\mathcal{A}\text{ outputs } 0 \mid \text{root $j$ is sanitized to } z,\ \text{rest fixed to } X'^{(j)}].
\]
Then $|p_{j-1} - p_j| \leq \mathbb{E}_{X'^{(j)}}\big| \mathbb{E}_z[q_j(z) \mid x^0_j] - \mathbb{E}_z[q_j(z) \mid x^1_j] \big|$ where the inner expectations are over the noise of root $j$. We bound the inner difference for each fixed shared context.

\medskip
\noindent\textbf{Case 1: root $j$ is type $\tau_I$ (FPE).} By PRP security of FPE under security parameter $\kappa$, no polynomial-time adversary can distinguish $E_F(K, N_{\tau_j}, x^0_j)$ from $E_F(K, N_{\tau_j}, x^1_j)$ except with probability $\mathrm{negl}(\kappa)$. Hence
\[
  |p_{j-1} - p_j| \leq \mathrm{negl}(\kappa).
\]

\medskip
\noindent\textbf{Case 2: root $j$ is type $\tau_{II}$ (mDP with budget $\epsilon_j$).} The mechanism $M_{\epsilon_j}$ is $\epsilon_j$-mDP under the index-space metric $d(t, t') = |t - t'|$, with $|t_j^0 - t_j^1| = |x_j^0 - x_j^1| / \delta_j$. Hence for every measurable output set $S$,
\[
  \Pr[M_{\epsilon_j}(x_j^0) \in S] \leq e^{\epsilon_j |x_j^0 - x_j^1| / \delta_j} \Pr[M_{\epsilon_j}(x_j^1) \in S].
\]
Marginalizing $\mathcal{A}$'s output (a function of $\hat{X}^b$) over the noise distribution of root $j$ and applying this pointwise,
\begin{align*}
  \mathbb{E}_z[q_j(z) \mid x^0_j]
  &= \sum_z \Pr[M_{\epsilon_j}(x_j^0) = z]\, q_j(z) \\
  &\leq e^{\epsilon_j |x_j^0 - x_j^1| / \delta_j} \sum_z \Pr[M_{\epsilon_j}(x_j^1) = z]\, q_j(z) \\
  &= e^{\epsilon_j |x_j^0 - x_j^1| / \delta_j}\, \mathbb{E}_z[q_j(z) \mid x^1_j].
\end{align*}
Since each expectation is in $[0, 1]$, this gives
\[
  \big| \mathbb{E}_z[q_j(z) \mid x^0_j] - \mathbb{E}_z[q_j(z) \mid x^1_j] \big|
  \leq e^{\epsilon_j |x_j^0 - x_j^1| / \delta_j} - 1
  \leq e^{\epsilon_j |x_j^0 - x_j^1| / \delta_j}.
\]
Taking expectation over $X'^{(j)}$,
\[
  |p_{j-1} - p_j| \leq e^{\epsilon_j |x_j^0 - x_j^1| / \delta_j}.
\]

\medskip
\noindent\textbf{Combining.} Summing over $j$,
\[
  \mathrm{Adv}_{\mathrm{DAG}}(\mathcal{A})
  \leq \sum_{j \in [\tau_{II}]} e^{\epsilon_j |x_j^0 - x_j^1| / \delta_j} + |[\tau_I]| \cdot \mathrm{negl}(\kappa).
\]
\end{proof}

\subsection{Proof of Theorem~\ref{thm:privacy-leakage}}
\label{app:thm:privacy-leakage}
\begin{proof}
Since $\mathcal{M}$ is $\epsilon$-mDP on 
$(\mathcal{Y}, d_{\mathcal{Y}})$, for all 
$y, y' \in \mathcal{Y}$ and measurable 
$T \subseteq \mathcal{Y}$:
\[
    \Pr[\mathcal{M}(y) \in T] 
    \leq e^{\epsilon\, d_{\mathcal{Y}}(y, y')} 
    \Pr[\mathcal{M}(y') \in T].
\]
Let $y = f(x)$ and $y' = f(x')$ for some $x, x' \in \mathcal{X}$. Then:
\[
    \Pr[\mathcal{M}(f(x)) \in T] 
    \leq e^{\epsilon\, d_{\mathcal{Y}}(f(x), f(x'))} 
    \Pr[\mathcal{M}(f(x')) \in T].
\]
Apply the upper Lipschitz bound 
$d_{\mathcal{Y}}(f(x), f(x')) \leq L\, d_{\mathcal{X}}(x, x')$ and substitute $\psi$:
\[
    \Pr[\psi(x) \in T] 
    \leq e^{\epsilon\, L\, d_{\mathcal{X}}(x, x')} 
    \Pr[\psi(x') \in T].
\]
\end{proof}

% \subsection{Proof of Theorem~\ref{thm:total-leakage}}
% \label{app:thm:total-leakage}
% \begin{proof}
% Since all mechanisms use independent randomness, the joint density factors:
% \begin{align*}
%     &\prod_{r=1}^{k} \frac{\Pr[\mathcal{M}_r(x_r) \in S_r]}{\Pr[\mathcal{M}_r(x_r') \in S_r]} \cdot \prod_{j=1}^{d} \frac{\Pr[\mathcal{N}_j(f_j(x)) \in T_j]}{\Pr[\mathcal{N}_j(f_j(x')) \in T_j]}.
% \end{align*}
% Each root ratio is bounded by the $\epsilon_r$-mDP guarantee of $\mathcal{M}_r$:
% \[
%     \frac{\Pr[\mathcal{M}_r(x_r) \in S_r]}{\Pr[\mathcal{M}_r(x_r') \in S_r]} \leq e^{\epsilon_r\, |x_r - x_r'|}.`
% \]
% Each derived ratio is bounded by the $\hat{\epsilon}_j$-mDP guarantee of $\mathcal{N}_j$:
% \[
%     \frac{\Pr[\mathcal{N}_j(f_j(x)) \in T_j]}{\Pr[\mathcal{N}_j(f_j(x')) \in T_j]} \leq e^{\hat{\epsilon}_j\, |f_j(x) - f_j(x')|}.
% \]
% Taking the product yields the stated bound. For tightness: for each mechanism, there exists a choice of $S_r^*$ or $T_j^*$ that achieves the corresponding ratio. Since the mechanisms are independent, the product set $S_1^* \times \cdots \times S_k^* \times T_1^* \times \cdots \times T_d^*$ achieves all worst cases simultaneously.
% \end{proof}

\subsection{Proof of Theorem~\ref{thm:implied-privacy}}
\label{app:thm:implied-privacy}
\begin{proof}
Since $\mathcal{M}$ is an $\epsilon$-mDP mechanism on $(\mathcal{X}, d_\mathcal{X})$, 
\[
    \Pr[\mathcal{M}(x) \in S]
    \;\leq\;
    e^{\epsilon \, d_{\mathcal{X}}(x, x')}
    \Pr[\mathcal{M}(x') \in S],
\]
where the probability is taken over the randomness of $\mathcal{M}$. \\

\noindent Fix $x, x' \in \mathcal{X}$ and let $T \subseteq \mathcal{Y}$ be any measurable set.  
Define $S = f^{-1}(T) \subseteq \mathcal{X}$. \\

\noindent Then, by the definition of $\phi = f \circ \mathcal{M}$,
\[
    \Pr[\phi(x) \in T]
    \;=\;
    \Pr[f(\mathcal{M}(x)) \in T]
    \;=\;
    \Pr[\mathcal{M}(x) \in S].
\]
Applying the mDP property of $\mathcal{M}$ yields
\begin{align*}
    \Pr[\mathcal{M}(x) \in S]
    &\leq
    e^{\epsilon \, d_{\mathcal{X}}(x, x')}
    \Pr[\mathcal{M}(x') \in S] \\
    \Rightarrow \quad
    \Pr[\phi(x) \in T]
    &\leq
    e^{\epsilon \, d_{\mathcal{X}}(x, x')}
    \Pr[\phi(x') \in T].
\end{align*}

\noindent Since $f$ is bi-Lipschitz, we have
\[
    d_{\mathcal{X}}(x, x')
    \;\leq\;
    \frac{1}{\alpha} \, d_{\mathcal{Y}}(f(x), f(x')).
\]

\noindent Substituting this bound into the inequality above gives
\[
    \Pr[\phi(x) \in T]
    \;\leq\;
    e^{\frac{\epsilon}{\alpha} \, d_{\mathcal{Y}}(f(x), f(x'))}
    \Pr[\phi(x') \in T].
\]

\noindent Thus, $\phi$ is $\frac{\epsilon}{\alpha}$-mDP with respect to the metric $d_{\mathcal{Y}}(f(x), f(x'))$.  
Consequently, adding mDP noise directly to $y = f(x)$ in $(\mathcal{Y}, d_{\mathcal{Y}})$ with privacy parameter $\frac{\epsilon}{\alpha}$ provides at least as strong a privacy guarantee as first noising $x$ and then applying $f$.
\end{proof}

% \subsection{Per-Template Privacy-Utility  Analysis}
\label{app:fair-utility-details}
\subsection{RQ1: Full Adversarial Utility Tables}
\label{app:rq1-tables}

This appendix provides the complete aggregate wMAPE and Risk Class Error (RCE) tables for all three mechanisms (Exponential, Bounded Laplace, Staircase) under each of the three adversary turn budgets ($t \in \{k{+}1, 2k{+}1, 3k{+}1\}$). Per-template wMAPE breakdowns at $\varepsilon = 0.1$ are also included. The main text (Sec.~\ref{sec:experiments-rq1}) reports the Exponential mechanism; the Bounded Laplace and Staircase results follow the same trends. Bootstrap standard errors (1000 resamples) are shown inline as {\scriptsize $\pm$SE}; aggregate SEs are propagated from per-template SEs as $\sqrt{\sum_i \mathrm{SE}_i^2}/k$.

\begin{table*}[t]
\centering
\setlength{\tabcolsep}{2.8pt}
\small
\caption{Aggregate wMAPE (\%) averaged across 8 templates with bootstrap SE. Budget $B = (k+1) \cdot \varepsilon$. Best mean per row in \textbf{bold}.}
\label{tab:rq1_kplus1_wmape}
\begin{tabular}{cc|ccc|ccc|ccc}
\toprule
& & \multicolumn{3}{c|}{\textbf{Exponential}} & \multicolumn{3}{c|}{\textbf{Bounded Laplace}} & \multicolumn{3}{c}{\textbf{Staircase}} \\
$\varepsilon$ & $\bar{B}$ & All & Roots & Opt & All & Roots & Opt & All & Roots & Opt \\
\midrule
5.0 & 18.12 & 0.6\pmstd{0.0} & 0.5\pmstd{0.0} & 0.4\pmstd{0.0} & 0.5\pmstd{0.0} & 0.4\pmstd{0.0} & 0.4\pmstd{0.0} & 0.4\pmstd{0.0} & 0.4\pmstd{0.0} & \textbf{0.4\pmstd{0.0}} \\
2.0 & 7.25 & 1.4\pmstd{0.1} & 1.0\pmstd{0.1} & 0.7\pmstd{0.0} & 0.8\pmstd{0.0} & 0.7\pmstd{0.0} & 0.5\pmstd{0.0} & 0.8\pmstd{0.0} & 0.5\pmstd{0.0} & \textbf{0.4\pmstd{0.0}} \\
1.0 & 3.62 & 3.0\pmstd{0.2} & 2.0\pmstd{0.1} & 1.3\pmstd{0.1} & 1.6\pmstd{0.1} & 1.1\pmstd{0.1} & 0.8\pmstd{0.0} & 1.5\pmstd{0.1} & 1.0\pmstd{0.1} & \textbf{0.7\pmstd{0.0}} \\
0.5 & 1.81 & 5.6\pmstd{0.3} & 4.0\pmstd{0.2} & 2.7\pmstd{0.1} & 3.1\pmstd{0.2} & 2.2\pmstd{0.1} & \textbf{1.4\pmstd{0.1}} & 3.1\pmstd{0.2} & 2.1\pmstd{0.1} & 1.5\pmstd{0.1} \\
0.2 & 0.73 & 12.6\pmstd{0.6} & 10.1\pmstd{0.5} & 6.6\pmstd{0.3} & 7.7\pmstd{0.5} & 5.6\pmstd{0.3} & 3.7\pmstd{0.2} & 7.4\pmstd{0.4} & 5.3\pmstd{0.3} & \textbf{3.4\pmstd{0.2}} \\
0.1 & 0.36 & 20.5\pmstd{0.9} & 14.9\pmstd{0.7} & 13.6\pmstd{0.7} & 14.3\pmstd{0.7} & 10.2\pmstd{0.5} & 6.9\pmstd{0.4} & 14.8\pmstd{0.8} & 10.5\pmstd{0.6} & \textbf{6.9\pmstd{0.3}} \\
0.05 & 0.18 & 35.0\pmstd{1.4} & 26.0\pmstd{1.1} & 21.6\pmstd{0.9} & 25.8\pmstd{1.2} & 19.0\pmstd{1.0} & \textbf{14.6\pmstd{0.7}} & 26.9\pmstd{1.2} & 20.7\pmstd{1.0} & 15.4\pmstd{0.8} \\
0.02 & 0.07 & 70.3\pmstd{4.3} & 55.5\pmstd{6.8} & 46.0\pmstd{2.8} & 51.5\pmstd{2.4} & 39.1\pmstd{1.7} & \textbf{29.8\pmstd{1.3}} & 50.6\pmstd{2.3} & 38.5\pmstd{1.6} & 31.3\pmstd{1.3} \\
0.01 & 0.04 & 122\pmstd{10.1} & 102\pmstd{9.6} & 86.6\pmstd{6.2} & 118\pmstd{33.3} & 63.7\pmstd{3.1} & \textbf{56.5\pmstd{2.9}} & 90.4\pmstd{6.4} & 61.1\pmstd{3.4} & 57.6\pmstd{5.1} \\
0.005 & 0.02 & 218\pmstd{18.4} & 154\pmstd{11.7} & \textbf{120\pmstd{10.1}} & 202\pmstd{29.9} & 138\pmstd{21.7} & 161\pmstd{19.5} & 158\pmstd{25.8} & 122\pmstd{12.1} & 136\pmstd{19.7} \\
\bottomrule
\end{tabular}
\end{table*}

\begin{table*}[t]
\centering
\setlength{\tabcolsep}{2.8pt}
\small
\caption{Aggregate Risk Class Error (\%) averaged across 8 templates with bootstrap SE. Budget $B = (k+1) \cdot \varepsilon$. Best mean per row in \textbf{bold}.}
\label{tab:rq1_kplus1_rce}
\begin{tabular}{cc|ccc|ccc|ccc}
\toprule
& & \multicolumn{3}{c|}{\textbf{Exponential}} & \multicolumn{3}{c|}{\textbf{Bounded Laplace}} & \multicolumn{3}{c}{\textbf{Staircase}} \\
$\varepsilon$ & $\bar{B}$ & All & Roots & Opt & All & Roots & Opt & All & Roots & Opt \\
\midrule
5.0 & 18.12 & 0.7\pmstd{0.2} & 0.6\pmstd{0.2} & \textbf{0.4\pmstd{0.2}} & 0.5\pmstd{0.2} & 0.5\pmstd{0.2} & 0.6\pmstd{0.2} & 0.6\pmstd{0.2} & 0.4\pmstd{0.2} & 0.4\pmstd{0.2} \\
2.0 & 7.25 & 1.6\pmstd{0.3} & 1.2\pmstd{0.3} & 0.6\pmstd{0.2} & 1.0\pmstd{0.2} & 0.9\pmstd{0.2} & 0.7\pmstd{0.2} & 1.0\pmstd{0.2} & 1.1\pmstd{0.3} & \textbf{0.5\pmstd{0.2}} \\
1.0 & 3.62 & 2.9\pmstd{0.4} & 2.1\pmstd{0.3} & 1.4\pmstd{0.3} & 1.7\pmstd{0.3} & 0.9\pmstd{0.2} & 0.8\pmstd{0.2} & 1.5\pmstd{0.3} & \textbf{0.7\pmstd{0.2}} & 0.9\pmstd{0.2} \\
0.5 & 1.81 & 5.1\pmstd{0.5} & 3.2\pmstd{0.4} & 2.8\pmstd{0.4} & 3.0\pmstd{0.4} & 2.2\pmstd{0.4} & 1.5\pmstd{0.3} & 3.3\pmstd{0.4} & 2.0\pmstd{0.3} & \textbf{1.4\pmstd{0.3}} \\
0.2 & 0.73 & 11.2\pmstd{0.7} & 8.6\pmstd{0.7} & 5.8\pmstd{0.5} & 5.9\pmstd{0.6} & 4.4\pmstd{0.5} & 3.5\pmstd{0.4} & 5.5\pmstd{0.5} & 4.6\pmstd{0.5} & \textbf{3.4\pmstd{0.4}} \\
0.1 & 0.36 & 15.4\pmstd{0.9} & 12.4\pmstd{0.7} & 10.3\pmstd{0.7} & 11.5\pmstd{0.7} & 8.5\pmstd{0.6} & 6.2\pmstd{0.6} & 10.3\pmstd{0.7} & 7.8\pmstd{0.6} & \textbf{5.2\pmstd{0.5}} \\
0.05 & 0.18 & 24.2\pmstd{1.0} & 19.1\pmstd{0.9} & 16.5\pmstd{0.9} & 19.3\pmstd{0.9} & 13.7\pmstd{0.8} & 11.5\pmstd{0.7} & 18.3\pmstd{0.9} & 14.2\pmstd{0.8} & \textbf{11.0\pmstd{0.7}} \\
0.02 & 0.07 & 32.3\pmstd{1.1} & 29.6\pmstd{1.1} & 26.0\pmstd{1.0} & 29.4\pmstd{1.1} & 26.3\pmstd{1.1} & \textbf{18.8\pmstd{0.9}} & 29.5\pmstd{1.1} & 24.9\pmstd{1.0} & 19.5\pmstd{0.9} \\
0.01 & 0.04 & 36.7\pmstd{1.1} & 34.4\pmstd{1.1} & 34.6\pmstd{1.1} & 37.4\pmstd{1.2} & 34.2\pmstd{1.1} & 30.0\pmstd{1.1} & 38.2\pmstd{1.2} & 33.8\pmstd{1.1} & \textbf{28.6\pmstd{1.1}} \\
0.005 & 0.02 & 38.3\pmstd{1.2} & 38.8\pmstd{1.1} & 37.9\pmstd{1.2} & 46.1\pmstd{1.2} & 41.5\pmstd{1.2} & 40.4\pmstd{1.2} & 44.3\pmstd{1.2} & 41.9\pmstd{1.2} & \textbf{37.6\pmstd{1.2}} \\
\bottomrule
\end{tabular}
\end{table*}

\begin{table*}[t]
\centering
\setlength{\tabcolsep}{2.8pt}
\small
\caption{Per-template wMAPE (\%) with bootstrap SE at $\varepsilon = 0.1$ ($\bar{B} = 0.36$). Budget $B = (k+1) \cdot \varepsilon$. Best per row in \textbf{bold}.}
\label{tab:rq1_kplus1_per_tmpl_wmape}
\begin{tabular}{l|ccc|ccc|ccc}
\toprule
& \multicolumn{3}{c|}{\textbf{Exponential}} & \multicolumn{3}{c|}{\textbf{Bounded Laplace}} & \multicolumn{3}{c}{\textbf{Staircase}} \\
Template & All & Roots & Opt & All & Roots & Opt & All & Roots & Opt \\
\midrule
ANEMIA & 2.6\pmstd{0.2} & 1.9\pmstd{0.1} & 1.2\pmstd{0.1} & 1.4\pmstd{0.1} & 0.9\pmstd{0.1} & 0.7\pmstd{0.1} & 1.2\pmstd{0.1} & 0.9\pmstd{0.1} & \textbf{0.6\pmstd{0.0}} \\
FIB4 & 25.3\pmstd{2.4} & 21.4\pmstd{2.1} & 17.8\pmstd{2.2} & 14.3\pmstd{1.4} & 11.0\pmstd{0.9} & \textbf{8.0\pmstd{0.7}} & 15.1\pmstd{2.0} & 9.8\pmstd{0.8} & 8.8\pmstd{0.8} \\
AIP & 43.8\pmstd{4.2} & 33.8\pmstd{3.3} & 25.3\pmstd{2.5} & 30.7\pmstd{3.3} & 23.5\pmstd{2.5} & 12.7\pmstd{1.3} & 29.1\pmstd{3.1} & 25.1\pmstd{2.7} & \textbf{12.0\pmstd{1.3}} \\
CONICITY & 2.6\pmstd{0.2} & 2.4\pmstd{0.1} & 1.8\pmstd{0.1} & 1.5\pmstd{0.1} & 0.9\pmstd{0.1} & 0.9\pmstd{0.1} & 1.5\pmstd{0.1} & 1.0\pmstd{0.1} & \textbf{0.9\pmstd{0.1}} \\
VASCULAR & 6.6\pmstd{0.5} & 4.4\pmstd{0.3} & 4.4\pmstd{0.3} & 3.3\pmstd{0.3} & 2.2\pmstd{0.2} & 2.4\pmstd{0.2} & 3.5\pmstd{0.3} & 2.1\pmstd{0.2} & \textbf{2.1\pmstd{0.1}} \\
TYG & 4.5\pmstd{0.3} & 3.4\pmstd{0.2} & 3.1\pmstd{0.2} & 3.5\pmstd{0.3} & 2.3\pmstd{0.2} & \textbf{1.7\pmstd{0.1}} & 3.2\pmstd{0.2} & 2.2\pmstd{0.2} & 1.8\pmstd{0.2} \\
HOMA & 28.5\pmstd{3.6} & 19.0\pmstd{2.8} & 18.1\pmstd{2.4} & 18.9\pmstd{2.2} & 12.3\pmstd{1.6} & \textbf{9.9\pmstd{1.3}} & 18.9\pmstd{2.3} & 12.2\pmstd{1.8} & 11.1\pmstd{1.0} \\
NLR & 49.8\pmstd{3.7} & 32.8\pmstd{2.4} & 37.5\pmstd{3.2} & 40.8\pmstd{3.8} & 28.2\pmstd{3.0} & 19.0\pmstd{2.0} & 46.1\pmstd{4.7} & 30.4\pmstd{3.0} & \textbf{18.1\pmstd{1.7}} \\
\bottomrule
\end{tabular}
\end{table*}

\begin{table*}[t]
\centering
\setlength{\tabcolsep}{2.8pt}
\small
\caption{Aggregate wMAPE (\%) averaged across 8 templates with bootstrap SE. Budget $B = (2k+1) \cdot \varepsilon$. Best mean per row in \textbf{bold}.}
\label{tab:rq1_2kplus1_wmape}
\begin{tabular}{cc|ccc|ccc|ccc}
\toprule
& & \multicolumn{3}{c|}{\textbf{Exponential}} & \multicolumn{3}{c|}{\textbf{Bounded Laplace}} & \multicolumn{3}{c}{\textbf{Staircase}} \\
$\varepsilon$ & $\bar{B}$ & All & Roots & Opt & All & Roots & Opt & All & Roots & Opt \\
\midrule
5.0 & 31.25 & 0.7\pmstd{0.0} & 0.4\pmstd{0.0} & 0.4\pmstd{0.0} & 0.5\pmstd{0.0} & 0.4\pmstd{0.0} & 0.4\pmstd{0.0} & 0.4\pmstd{0.0} & 0.4\pmstd{0.0} & \textbf{0.4\pmstd{0.0}} \\
2.0 & 12.50 & 1.5\pmstd{0.1} & 0.7\pmstd{0.0} & 0.5\pmstd{0.0} & 0.8\pmstd{0.0} & 0.5\pmstd{0.0} & 0.4\pmstd{0.0} & 0.8\pmstd{0.0} & 0.4\pmstd{0.0} & \textbf{0.4\pmstd{0.0}} \\
1.0 & 6.25 & 3.0\pmstd{0.2} & 1.2\pmstd{0.1} & 0.9\pmstd{0.1} & 1.6\pmstd{0.1} & 0.8\pmstd{0.0} & 0.6\pmstd{0.0} & 1.6\pmstd{0.1} & 0.7\pmstd{0.0} & \textbf{0.5\pmstd{0.0}} \\
0.5 & 3.12 & 5.7\pmstd{0.3} & 2.4\pmstd{0.1} & 1.7\pmstd{0.1} & 3.0\pmstd{0.2} & 1.3\pmstd{0.1} & 0.9\pmstd{0.0} & 2.9\pmstd{0.2} & 1.2\pmstd{0.1} & \textbf{0.8\pmstd{0.0}} \\
0.2 & 1.25 & 12.3\pmstd{0.6} & 5.9\pmstd{0.4} & 4.5\pmstd{0.2} & 8.6\pmstd{0.5} & 3.0\pmstd{0.2} & \textbf{2.2\pmstd{0.1}} & 8.0\pmstd{0.5} & 3.2\pmstd{0.2} & 2.2\pmstd{0.1} \\
0.1 & 0.62 & 20.3\pmstd{1.0} & 11.4\pmstd{0.6} & 7.8\pmstd{0.3} & 14.0\pmstd{0.8} & 6.2\pmstd{0.4} & \textbf{4.3\pmstd{0.3}} & 15.5\pmstd{0.8} & 6.6\pmstd{0.4} & 4.3\pmstd{0.2} \\
0.05 & 0.31 & 33.8\pmstd{1.7} & 18.0\pmstd{0.8} & 13.3\pmstd{0.6} & 25.8\pmstd{1.3} & 11.5\pmstd{0.6} & 8.8\pmstd{0.5} & 26.3\pmstd{1.2} & 11.7\pmstd{0.6} & \textbf{8.7\pmstd{0.4}} \\
0.02 & 0.12 & 75.7\pmstd{10.9} & 34.1\pmstd{1.5} & 29.2\pmstd{1.2} & 48.1\pmstd{2.0} & 27.7\pmstd{1.3} & 20.7\pmstd{1.0} & 48.0\pmstd{2.1} & 29.1\pmstd{1.3} & \textbf{20.4\pmstd{1.0}} \\
0.01 & 0.06 & 113\pmstd{8.4} & 63.0\pmstd{5.4} & 48.3\pmstd{2.6} & 82.4\pmstd{3.9} & 43.0\pmstd{1.8} & 43.9\pmstd{5.8} & 80.4\pmstd{5.4} & 43.9\pmstd{2.0} & \textbf{36.1\pmstd{1.7}} \\
0.005 & 0.03 & 218\pmstd{20.7} & 105\pmstd{6.9} & 102\pmstd{6.8} & 198\pmstd{27.5} & 76.5\pmstd{4.2} & 69.5\pmstd{5.2} & 223\pmstd{44.9} & 89.4\pmstd{13.8} & \textbf{65.5\pmstd{4.6}} \\
\bottomrule
\end{tabular}
\end{table*}

\begin{table*}[t]
\centering
\setlength{\tabcolsep}{2.8pt}
\small
\caption{Aggregate Risk Class Error (\%) averaged across 8 templates with bootstrap SE. Budget $B = (2k+1) \cdot \varepsilon$. Best mean per row in \textbf{bold}.}
\label{tab:rq1_2kplus1_rce}
\begin{tabular}{cc|ccc|ccc|ccc}
\toprule
& & \multicolumn{3}{c|}{\textbf{Exponential}} & \multicolumn{3}{c|}{\textbf{Bounded Laplace}} & \multicolumn{3}{c}{\textbf{Staircase}} \\
$\varepsilon$ & $\bar{B}$ & All & Roots & Opt & All & Roots & Opt & All & Roots & Opt \\
\midrule
5.0 & 31.25 & 0.6\pmstd{0.2} & 0.5\pmstd{0.2} & 0.6\pmstd{0.2} & 0.7\pmstd{0.2} & 0.5\pmstd{0.2} & \textbf{0.3\pmstd{0.1}} & 0.5\pmstd{0.2} & 0.4\pmstd{0.2} & 0.4\pmstd{0.2} \\
2.0 & 12.50 & 1.7\pmstd{0.3} & 0.9\pmstd{0.2} & 0.9\pmstd{0.2} & 1.1\pmstd{0.2} & 0.5\pmstd{0.2} & 0.4\pmstd{0.2} & 0.8\pmstd{0.2} & 0.5\pmstd{0.2} & \textbf{0.3\pmstd{0.1}} \\
1.0 & 6.25 & 2.7\pmstd{0.4} & 1.5\pmstd{0.3} & 0.8\pmstd{0.2} & 1.9\pmstd{0.3} & 0.6\pmstd{0.2} & 0.9\pmstd{0.2} & 1.4\pmstd{0.3} & 0.9\pmstd{0.2} & \textbf{0.4\pmstd{0.2}} \\
0.5 & 3.12 & 4.8\pmstd{0.5} & 2.8\pmstd{0.4} & 1.9\pmstd{0.3} & 2.8\pmstd{0.4} & 1.3\pmstd{0.3} & 1.2\pmstd{0.3} & 2.9\pmstd{0.4} & 1.0\pmstd{0.2} & \textbf{0.9\pmstd{0.2}} \\
0.2 & 1.25 & 10.2\pmstd{0.7} & 4.6\pmstd{0.5} & 4.5\pmstd{0.5} & 6.2\pmstd{0.6} & 3.0\pmstd{0.4} & \textbf{2.2\pmstd{0.4}} & 6.6\pmstd{0.6} & 2.8\pmstd{0.4} & 2.4\pmstd{0.4} \\
0.1 & 0.62 & 15.7\pmstd{0.8} & 10.1\pmstd{0.7} & 7.0\pmstd{0.6} & 10.6\pmstd{0.7} & 5.4\pmstd{0.5} & 4.4\pmstd{0.5} & 11.5\pmstd{0.7} & 5.4\pmstd{0.5} & \textbf{3.6\pmstd{0.4}} \\
0.05 & 0.31 & 23.2\pmstd{1.0} & 14.5\pmstd{0.8} & 11.7\pmstd{0.7} & 18.3\pmstd{0.9} & 9.2\pmstd{0.7} & 7.4\pmstd{0.6} & 19.5\pmstd{1.0} & 8.7\pmstd{0.7} & \textbf{6.4\pmstd{0.6}} \\
0.02 & 0.12 & 32.0\pmstd{1.1} & 25.4\pmstd{1.0} & 21.8\pmstd{1.0} & 30.2\pmstd{1.1} & 18.0\pmstd{0.9} & 14.8\pmstd{0.8} & 29.4\pmstd{1.1} & 20.4\pmstd{0.9} & \textbf{14.6\pmstd{0.8}} \\
0.01 & 0.06 & 35.8\pmstd{1.1} & 30.1\pmstd{1.1} & 27.2\pmstd{1.0} & 37.0\pmstd{1.2} & 27.1\pmstd{1.1} & 24.0\pmstd{1.0} & 37.8\pmstd{1.2} & 25.9\pmstd{1.1} & \textbf{23.1\pmstd{1.0}} \\
0.005 & 0.03 & 38.0\pmstd{1.1} & 34.6\pmstd{1.1} & 35.7\pmstd{1.1} & 42.5\pmstd{1.2} & 34.8\pmstd{1.2} & 33.6\pmstd{1.2} & 42.4\pmstd{1.2} & 35.1\pmstd{1.2} & \textbf{30.7\pmstd{1.1}} \\
\bottomrule
\end{tabular}
\end{table*}

\begin{table*}[t]
\centering
\setlength{\tabcolsep}{2.8pt}
\small
\caption{Per-template wMAPE (\%) with bootstrap SE at $\varepsilon = 0.1$ ($\bar{B} = 0.62$). Budget $B = (2k+1) \cdot \varepsilon$. Best per row in \textbf{bold}.}
\label{tab:rq1_2kplus1_per_tmpl_wmape}
\begin{tabular}{l|ccc|ccc|ccc}
\toprule
& \multicolumn{3}{c|}{\textbf{Exponential}} & \multicolumn{3}{c|}{\textbf{Bounded Laplace}} & \multicolumn{3}{c}{\textbf{Staircase}} \\
Template & All & Roots & Opt & All & Roots & Opt & All & Roots & Opt \\
\midrule
ANEMIA & 2.4\pmstd{0.1} & 1.1\pmstd{0.1} & 0.7\pmstd{0.0} & 1.3\pmstd{0.1} & 0.6\pmstd{0.0} & 0.4\pmstd{0.0} & 1.4\pmstd{0.1} & 0.6\pmstd{0.0} & \textbf{0.3\pmstd{0.0}} \\
FIB4 & 23.9\pmstd{1.8} & 10.1\pmstd{0.8} & 10.0\pmstd{0.9} & 13.7\pmstd{1.2} & 7.2\pmstd{0.7} & \textbf{4.6\pmstd{0.4}} & 15.2\pmstd{1.2} & 5.9\pmstd{0.5} & 5.0\pmstd{0.4} \\
AIP & 41.9\pmstd{4.1} & 25.3\pmstd{2.4} & 12.7\pmstd{1.4} & 28.7\pmstd{3.1} & 14.6\pmstd{1.6} & \textbf{7.6\pmstd{0.8}} & 35.3\pmstd{3.9} & 17.0\pmstd{2.0} & 7.8\pmstd{1.0} \\
CONICITY & 2.8\pmstd{0.2} & 1.2\pmstd{0.1} & 1.0\pmstd{0.1} & 1.4\pmstd{0.1} & 0.6\pmstd{0.0} & \textbf{0.5\pmstd{0.0}} & 1.5\pmstd{0.1} & 0.6\pmstd{0.0} & 0.6\pmstd{0.0} \\
VASCULAR & 6.3\pmstd{0.5} & 2.7\pmstd{0.2} & 2.6\pmstd{0.2} & 3.5\pmstd{0.2} & 1.4\pmstd{0.1} & 1.3\pmstd{0.1} & 3.2\pmstd{0.2} & 1.4\pmstd{0.1} & \textbf{1.2\pmstd{0.1}} \\
TYG & 4.6\pmstd{0.3} & 2.6\pmstd{0.2} & 2.0\pmstd{0.1} & 3.4\pmstd{0.3} & 1.4\pmstd{0.1} & \textbf{1.1\pmstd{0.1}} & 3.3\pmstd{0.3} & 1.6\pmstd{0.1} & 1.1\pmstd{0.1} \\
HOMA & 31.1\pmstd{4.4} & 17.3\pmstd{1.7} & 9.9\pmstd{1.2} & 18.5\pmstd{2.6} & 7.2\pmstd{1.0} & 6.1\pmstd{0.9} & 17.7\pmstd{2.3} & 6.7\pmstd{1.0} & \textbf{5.6\pmstd{0.7}} \\
NLR & 49.2\pmstd{4.3} & 31.2\pmstd{3.4} & 23.1\pmstd{1.5} & 41.4\pmstd{4.9} & 16.9\pmstd{2.0} & \textbf{12.5\pmstd{1.6}} & 46.6\pmstd{3.9} & 19.5\pmstd{2.6} & 12.8\pmstd{1.1} \\
\bottomrule
\end{tabular}
\end{table*}

\begin{table*}[t]
\centering
\setlength{\tabcolsep}{2.8pt}
\small
\caption{Aggregate wMAPE (\%) averaged across 8 templates with bootstrap SE. Budget $B = (3k+1) \cdot \varepsilon$. Best mean per row in \textbf{bold}.}
\label{tab:rq1_3kplus1_wmape}
\begin{tabular}{cc|ccc|ccc|ccc}
\toprule
& & \multicolumn{3}{c|}{\textbf{Exponential}} & \multicolumn{3}{c|}{\textbf{Bounded Laplace}} & \multicolumn{3}{c}{\textbf{Staircase}} \\
$\varepsilon$ & $\bar{B}$ & All & Roots & Opt & All & Roots & Opt & All & Roots & Opt \\
\midrule
5.0 & 44.38 & 0.6\pmstd{0.0} & 0.4\pmstd{0.0} & 0.4\pmstd{0.0} & 0.5\pmstd{0.0} & 0.4\pmstd{0.0} & 0.4\pmstd{0.0} & 0.4\pmstd{0.0} & \textbf{0.4\pmstd{0.0}} & 0.4\pmstd{0.0} \\
2.0 & 17.75 & 1.5\pmstd{0.1} & 0.5\pmstd{0.0} & 0.4\pmstd{0.0} & 0.8\pmstd{0.0} & 0.4\pmstd{0.0} & 0.4\pmstd{0.0} & 0.8\pmstd{0.1} & 0.4\pmstd{0.0} & \textbf{0.4\pmstd{0.0}} \\
1.0 & 8.88 & 3.0\pmstd{0.2} & 0.9\pmstd{0.1} & 0.6\pmstd{0.0} & 1.6\pmstd{0.1} & 0.6\pmstd{0.0} & 0.5\pmstd{0.0} & 1.5\pmstd{0.1} & 0.4\pmstd{0.0} & \textbf{0.4\pmstd{0.0}} \\
0.5 & 4.44 & 5.9\pmstd{0.3} & 1.8\pmstd{0.1} & 1.2\pmstd{0.1} & 3.0\pmstd{0.2} & 1.0\pmstd{0.1} & 0.7\pmstd{0.0} & 3.0\pmstd{0.2} & 0.9\pmstd{0.0} & \textbf{0.6\pmstd{0.0}} \\
0.2 & 1.78 & 12.2\pmstd{0.5} & 4.6\pmstd{0.3} & 3.3\pmstd{0.2} & 7.3\pmstd{0.4} & 2.2\pmstd{0.1} & \textbf{1.5\pmstd{0.1}} & 7.3\pmstd{0.4} & 2.3\pmstd{0.2} & 1.6\pmstd{0.1} \\
0.1 & 0.89 & 21.5\pmstd{1.0} & 8.7\pmstd{0.5} & 5.6\pmstd{0.3} & 15.9\pmstd{0.9} & 4.3\pmstd{0.2} & \textbf{2.8\pmstd{0.1}} & 13.8\pmstd{0.7} & 4.1\pmstd{0.2} & 3.1\pmstd{0.2} \\
0.05 & 0.44 & 33.7\pmstd{1.8} & 14.3\pmstd{0.7} & 11.2\pmstd{0.6} & 25.9\pmstd{1.1} & 8.7\pmstd{0.6} & \textbf{5.6\pmstd{0.3}} & 26.2\pmstd{1.2} & 8.4\pmstd{0.5} & 6.2\pmstd{0.3} \\
0.02 & 0.18 & 72.3\pmstd{5.1} & 27.9\pmstd{1.3} & 21.9\pmstd{1.0} & 50.0\pmstd{2.1} & 19.4\pmstd{1.0} & \textbf{14.5\pmstd{0.7}} & 48.1\pmstd{2.2} & 19.1\pmstd{0.9} & 14.8\pmstd{0.8} \\
0.01 & 0.09 & 129\pmstd{9.8} & 47.3\pmstd{2.6} & 37.3\pmstd{1.6} & 94.2\pmstd{9.2} & 35.7\pmstd{1.7} & \textbf{27.1\pmstd{1.2}} & 84.6\pmstd{6.1} & 34.9\pmstd{1.5} & 30.1\pmstd{2.9} \\
0.005 & 0.04 & 235\pmstd{26.4} & 81.0\pmstd{5.3} & 76.6\pmstd{5.1} & 193\pmstd{26.2} & 53.3\pmstd{2.4} & \textbf{45.4\pmstd{2.0}} & 176\pmstd{20.1} & 54.3\pmstd{2.4} & 50.6\pmstd{3.8} \\
\bottomrule
\end{tabular}
\end{table*}

\begin{table*}[t]
\centering
\setlength{\tabcolsep}{2.8pt}
\small
\caption{Aggregate Risk Class Error (\%) averaged across 8 templates with bootstrap SE. Budget $B = (3k+1) \cdot \varepsilon$. Best mean per row in \textbf{bold}.}
\label{tab:rq1_3kplus1_rce}
\begin{tabular}{cc|ccc|ccc|ccc}
\toprule
& & \multicolumn{3}{c|}{\textbf{Exponential}} & \multicolumn{3}{c|}{\textbf{Bounded Laplace}} & \multicolumn{3}{c}{\textbf{Staircase}} \\
$\varepsilon$ & $\bar{B}$ & All & Roots & Opt & All & Roots & Opt & All & Roots & Opt \\
\midrule
5.0 & 44.38 & 1.0\pmstd{0.2} & 0.4\pmstd{0.1} & 0.5\pmstd{0.2} & 0.5\pmstd{0.2} & \textbf{0.3\pmstd{0.1}} & 0.4\pmstd{0.2} & 0.6\pmstd{0.2} & 0.4\pmstd{0.2} & 0.4\pmstd{0.2} \\
2.0 & 17.75 & 1.8\pmstd{0.3} & 0.6\pmstd{0.2} & 0.8\pmstd{0.2} & 0.8\pmstd{0.2} & \textbf{0.4\pmstd{0.2}} & 0.6\pmstd{0.2} & 1.0\pmstd{0.2} & 0.5\pmstd{0.2} & 0.4\pmstd{0.2} \\
1.0 & 8.88 & 2.8\pmstd{0.4} & 1.1\pmstd{0.3} & 0.6\pmstd{0.2} & 1.8\pmstd{0.3} & 0.9\pmstd{0.2} & 0.7\pmstd{0.2} & 1.8\pmstd{0.3} & 0.8\pmstd{0.2} & \textbf{0.5\pmstd{0.2}} \\
0.5 & 4.44 & 3.9\pmstd{0.4} & 1.8\pmstd{0.3} & 1.8\pmstd{0.3} & 2.5\pmstd{0.4} & 1.1\pmstd{0.3} & 1.0\pmstd{0.2} & 2.3\pmstd{0.4} & 1.2\pmstd{0.2} & \textbf{0.9\pmstd{0.2}} \\
0.2 & 1.78 & 10.8\pmstd{0.7} & 4.2\pmstd{0.5} & 2.9\pmstd{0.4} & 6.4\pmstd{0.6} & 2.4\pmstd{0.4} & 2.0\pmstd{0.3} & 5.7\pmstd{0.5} & 1.8\pmstd{0.3} & \textbf{1.2\pmstd{0.3}} \\
0.1 & 0.89 & 17.5\pmstd{0.9} & 6.5\pmstd{0.6} & 5.3\pmstd{0.5} & 12.0\pmstd{0.8} & 4.1\pmstd{0.5} & \textbf{2.8\pmstd{0.4}} & 11.0\pmstd{0.7} & 3.7\pmstd{0.4} & 2.8\pmstd{0.4} \\
0.05 & 0.44 & 22.1\pmstd{1.0} & 12.3\pmstd{0.8} & 9.4\pmstd{0.7} & 17.3\pmstd{0.9} & 7.9\pmstd{0.6} & 5.3\pmstd{0.5} & 16.6\pmstd{0.9} & 6.2\pmstd{0.6} & \textbf{5.2\pmstd{0.5}} \\
0.02 & 0.18 & 30.8\pmstd{1.1} & 19.0\pmstd{0.9} & 16.2\pmstd{0.9} & 30.8\pmstd{1.1} & 14.0\pmstd{0.8} & \textbf{11.2\pmstd{0.7}} & 30.3\pmstd{1.1} & 13.9\pmstd{0.8} & 11.9\pmstd{0.8} \\
0.01 & 0.09 & 36.2\pmstd{1.1} & 26.5\pmstd{1.0} & 24.6\pmstd{1.0} & 36.0\pmstd{1.1} & 23.0\pmstd{1.0} & 19.0\pmstd{0.9} & 41.0\pmstd{1.2} & 23.7\pmstd{1.0} & \textbf{18.3\pmstd{0.9}} \\
0.005 & 0.04 & 38.8\pmstd{1.1} & 32.4\pmstd{1.1} & 34.2\pmstd{1.1} & 46.3\pmstd{1.3} & 30.8\pmstd{1.1} & \textbf{26.8\pmstd{1.1}} & 45.8\pmstd{1.2} & 33.9\pmstd{1.1} & 27.8\pmstd{1.1} \\
\bottomrule
\end{tabular}
\end{table*}

\begin{table*}[t]
\centering
\setlength{\tabcolsep}{2.8pt}
\small
\caption{Per-template wMAPE (\%) with bootstrap SE at $\varepsilon = 0.1$ ($\bar{B} = 0.89$). Budget $B = (3k+1) \cdot \varepsilon$. Best per row in \textbf{bold}.}
\label{tab:rq1_3kplus1_per_tmpl_wmape}
\begin{tabular}{l|ccc|ccc|ccc}
\toprule
& \multicolumn{3}{c|}{\textbf{Exponential}} & \multicolumn{3}{c|}{\textbf{Bounded Laplace}} & \multicolumn{3}{c}{\textbf{Staircase}} \\
Template & All & Roots & Opt & All & Roots & Opt & All & Roots & Opt \\
\midrule
ANEMIA & 2.6\pmstd{0.2} & 0.7\pmstd{0.0} & 0.5\pmstd{0.0} & 1.4\pmstd{0.1} & 0.4\pmstd{0.0} & 0.3\pmstd{0.0} & 1.3\pmstd{0.1} & 0.4\pmstd{0.0} & \textbf{0.2\pmstd{0.0}} \\
FIB4 & 26.1\pmstd{2.1} & 9.3\pmstd{0.9} & 5.9\pmstd{0.4} & 12.7\pmstd{1.3} & 4.5\pmstd{0.4} & \textbf{3.2\pmstd{0.2}} & 13.2\pmstd{1.0} & 3.9\pmstd{0.3} & 3.9\pmstd{0.5} \\
AIP & 52.9\pmstd{4.8} & 20.4\pmstd{2.3} & 10.3\pmstd{1.2} & 32.1\pmstd{3.5} & 10.6\pmstd{1.3} & 5.6\pmstd{0.6} & 30.3\pmstd{3.2} & 9.6\pmstd{1.1} & \textbf{5.5\pmstd{0.6}} \\
CONICITY & 2.7\pmstd{0.2} & 0.8\pmstd{0.1} & 0.7\pmstd{0.0} & 1.6\pmstd{0.1} & 0.5\pmstd{0.0} & \textbf{0.4\pmstd{0.0}} & 1.4\pmstd{0.1} & 0.4\pmstd{0.0} & 0.4\pmstd{0.0} \\
VASCULAR & 6.0\pmstd{0.4} & 2.0\pmstd{0.1} & 1.6\pmstd{0.1} & 3.5\pmstd{0.2} & 1.0\pmstd{0.1} & \textbf{0.9\pmstd{0.1}} & 3.0\pmstd{0.2} & 0.9\pmstd{0.1} & 1.0\pmstd{0.1} \\
TYG & 4.4\pmstd{0.3} & 2.0\pmstd{0.2} & 1.4\pmstd{0.1} & 3.1\pmstd{0.2} & 1.1\pmstd{0.1} & \textbf{0.7\pmstd{0.1}} & 3.3\pmstd{0.3} & 1.0\pmstd{0.1} & 0.8\pmstd{0.1} \\
HOMA & 30.7\pmstd{3.7} & 10.6\pmstd{1.2} & 7.9\pmstd{0.8} & 18.8\pmstd{2.5} & 5.4\pmstd{0.8} & \textbf{3.5\pmstd{0.5}} & 19.6\pmstd{2.7} & 5.4\pmstd{0.8} & 4.5\pmstd{0.4} \\
NLR & 46.4\pmstd{4.1} & 23.8\pmstd{2.6} & 16.5\pmstd{1.5} & 54.1\pmstd{5.7} & 10.8\pmstd{1.0} & \textbf{8.2\pmstd{0.6}} & 38.5\pmstd{3.8} & 11.1\pmstd{0.9} & 8.5\pmstd{0.9} \\
\bottomrule
\end{tabular}
\end{table*}